\documentclass[10pt,a4]{article}
\usepackage[top=0.85in,footskip=0.75in,marginparwidth=2in]{geometry}

\usepackage[utf8]{inputenc}

\usepackage{natbib}

\usepackage{nameref}
\usepackage[hidelinks]{hyperref}
\usepackage{listings}

\usepackage{microtype}
\DisableLigatures[f]{encoding = *, family = * }

\usepackage{lastpage,fancyhdr,graphicx}
\usepackage{epstopdf}
\fancyheadoffset[L]{2.25in}
\fancyfootoffset[L]{2.25in}

\usepackage{color}
\newcommand{\citepsup}{\citep}
\newcommand{\citetsup}{\citet}
\definecolor{Gray}{gray}{.25}

\usepackage{graphicx}

\usepackage{algorithm2e}
\RestyleAlgo{boxruled}
\LinesNumbered
\DontPrintSemicolon
\SetAlFnt{\normalsize}
\SetKwInOut{Input}{Input}

\usepackage{mathtools}
\usepackage{titlesec}

\usepackage{amsthm,amsmath,amssymb,mathrsfs}
\usepackage{framed, enumitem}
\usepackage{bm}

\DeclareMathOperator{\FDR}{FDR}
\DeclareMathOperator{\FDP}{FDP}
\DeclareMathOperator{\fdr}{fdr}
\DeclareMathOperator{\Cfdr}{Clfdr}
\DeclareMathOperator{\argmin}{argmin}
\DeclareMathOperator*{\argmax}{argmax}   % Jan Hlavacek
\DeclareMathOperator*{\expit}{expit}   
\DeclareMathOperator{\plim}{plim}

\DeclareMathOperator{\weightbudget}{WeightBudget}

\DeclareMathOperator{\clamp}{clamp}
\DeclarePairedDelimiter{\ceil}{\lceil}{\rceil}
\DeclarePairedDelimiter{\floor}{\lfloor}{\rfloor}

\newcommand{\RR}{\mathbb{R}}
\newcommand{\EE}{\mathbb{E}}
\newcommand{\PP}{\mathbb{P}}
\newcommand{\ind}{\mathbf{1}}
\newcommand{\EEs}[2][]{\mathbb{E}_{#1}\left[#2\right]}
\newcommand{\PPs}[2][]{\mathbb{P}_{#1}\left[#2\right]}

\newcommand{\Hnull}{\mathscr{H}_0}
\newcommand{\allpairs}{((P_i, X_i))_{i \in [m]}}

\newcommand{\allpvalues}{(P_i)_{i \in [m]}}

\newcommand{\norm}[1]{\left\lVert#1\right\rVert}
\newcommand{\p}[1]{\left(#1\right)}
\newcommand{\abs}[1]{\left\lvert#1\right\rvert}
\newcommand{\sqb}[1]{\left[#1\right]}
\newcommand{\cb}[1]{\left\{#1\right\}}
\newcommand{\simiid}{\,{\buildrel \text{iid} \over \sim\,}}
\newcommand{\cond}{\,\big|\,}
\newcommand{\Kproof}{\tilde{\kappa}}

\newcommand{\thresholdfunction}{s}
\newcommand{\thresholdfunctionfunction}{s}

\theoremstyle{definition} %-- WH 6 Mar 2018: making all such environments printed in roman
\newtheorem{thm}{Theorem}
\newtheorem{cor}{Corollary}
\newtheorem{lem}{Lemma}
\newtheorem{prop}{Proposition}
\newtheorem{assumption}{Assumption}

\newtheorem{remark}{Remark}
\newtheorem{example}{Example}
\newtheorem{defn}{Definition}
\newtheorem{specif}{Specification}

\theoremstyle{remark}

\newcommand{\supplementname}{Supplement}

\usepackage{soul}

\begin{document}
\vspace*{0.35in}

% title goes here:
\begin{flushleft}
{\Large\textbf{Covariate powered cross-weighted multiple testing}}
\newline
\\
Nikolaos Ignatiadis\textsuperscript{1}\\
Wolfgang Huber\textsuperscript{2}\\
\bigskip
\textbf{1} Department of Statistics, Stanford University, USA\\
ignat@stanford.edu\\[2ex]
\textbf{2} European Molecular Biology Laboratory, Heidelberg, Germany\\
wolfgang.huber@embl.org
\bigskip

\end{flushleft}

%------------------------------------------------------------
\section*{Summary}
%------------------------------------------------------------
A fundamental task in the analysis of datasets with many variables is screening for associations. This can be cast as a multiple testing task, where the objective is achieving high detection power while controlling type I error. We consider $m$ hypothesis tests represented by pairs $((P_i, X_i))_{1\leq i \leq m}$ of p-values $P_i$ and covariates $X_i$, such that $P_i \perp X_i$ if $H_i$ is null. Here, we show how to use information potentially available in the covariates about heterogeneities among hypotheses to increase power compared to conventional procedures that only use the $P_i$. To this end, we upgrade existing weighted multiple testing procedures through the Independent Hypothesis Weighting (IHW) framework to use data-driven weights that are calculated as a function of the covariates. Finite sample guarantees, e.g., false discovery rate (FDR) control, are derived from cross-weighting, a data-splitting approach that enables learning the weight-covariate function without overfitting as long as the hypotheses can be partitioned into independent folds, with arbitrary within-fold dependence. IHW has increased power compared to methods that do not use covariate information. A key implication of IHW is that hypothesis rejection in common multiple testing setups should not proceed according to the ranking of the p-values, but by an alternative ranking implied by the covariate-weighted p-values. 

\vspace*{1cm}

\noindent\textbf{Keywords:} Benjamini-Hochberg, Empirical Bayes, False Discovery Rate, Independent Hypothesis Weighting, Multiple Testing, p-value weighting

\vspace*{1cm}

%\newpage
%\tableofcontents

\newpage

%--------------------------------------------------
\section{Introduction}
%--------------------------------------------------
Screening large datasets for interesting associations is a basic operation in statistical data analysis.
A frequently taken approach is to enumerate all potential associations, set up a hypothesis test for each of them, summarize the results by the p-values $P_i$, and select as \emph{discoveries} all hypotheses with a small enough p-value; typically, this is a small fraction of all hypotheses. More formally, for some cutoff $\hat{t}$:

\begin{equation}
\label{eq:rejection}
\text{Reject hypothesis } i  \;\; \Longleftrightarrow \;\; P_i \leq \hat{t}
\end{equation}
The choice of the cutoff $\hat{t}$ may be data-driven and is determined by a multiple testing procedure, such as those proposed by \citet{bonferroni1935calcolo} or \citet{benjamini1995controlling}, which compute a $\hat{t}$ that provides a defined level of protection against spurious discoveries. Common objectives are control of the family-wise error rate ($\text{FWER}$) or the false discovery rate ($\FDR$).

These procedures operate solely on the list of p-values. Here, we consider situations in which beyond the p-value $P_i$, side information represented by a covariate $X_i$ is available for each hypothesis. Such side-information reflects heterogeneity among the tests and may ---more or less directly---carry information about their different power, or the different prior probabilities of their null hypothesis being true. Suitable covariates are often apparent to domain scientists or to statisticians. We will see that procedures that take into account such side information often have higher power, in the sense that they make more discoveries at the same level of type-I error.

To illustrate, we use a high-throughput genetics dataset by \citet{grubert2015genetic}, who aimed to discover associations between genetic polymorphisms (SNPs) in the human genome and the activity of genomic regions (H3K27ac peaks). The main idea of the analysis of these data, which is presented in more detail in Section~\ref{sec:hQTLexample}, is to carry out a hypothesis test for each pair of SNP and region on the same chromosome. On Chromosomes 1 and 2, $N_1 = 645452$ and $N_2 = 699343$ SNPs were recorded, and H3K27ac levels were measured in $K_1 = 12193$ and $K_2 = 11232$ regions, which amounts to nearly 16 billion ($N_1K_1+N_2K_2$) tests. Figure~\ref{fig:H3K27ac_histograms} illustrates how the p-value distributions differ as a function of the genomic distance between SNP and region. These differences are consistent with biological domain knowledge: associations across shorter distances are a-priori more plausible and empirically more frequent. Methods that are able to take into account this heterogeneity among the tests should be able to discover more associations at the same $\FDR$, compared to (\ref{eq:rejection}), which ignores such side information.

\begin{figure}
\centering
\includegraphics[width=0.8\textwidth]{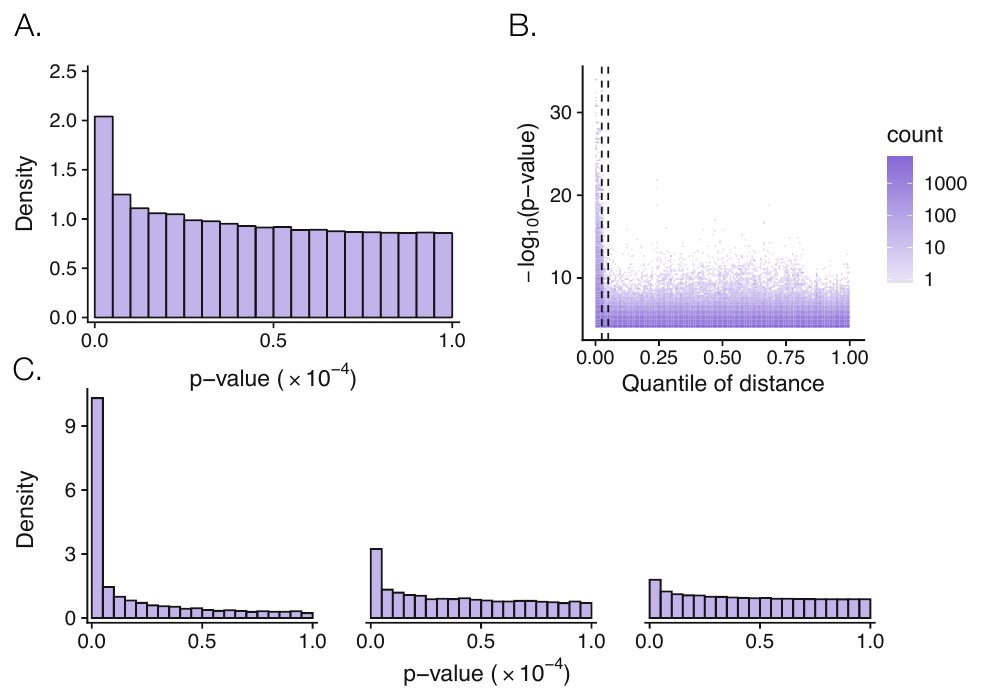}
\caption{\normalfont\textbf{Heterogeneous multiple hypothesis testing in a biological example.} For each hypothesis $(i=1,\ldots, m)$, a p-value $P_i$ is provided as well as a covariate $X_i$, which here is the genomic distance between the two features tested for association: a single nucleotide polymorphism (SNP) and a biochemical chromatin modification. \textbf{A. Histogram of p-values}:  we recognize the peak close to the origin, corresponding to enrichment of alternative hypotheses, and a near-uniform tail for larger p-values. Note that the displayed p-values are right-censored at $10^{-4}$, as is further explained in Section~\ref{sec:hQTLexample}, which provides more context on the data. \textbf{B. Two-dimensional heatmap of bin counts} of the joint empirical distribution $\left(-\log_{10}P_i, X_i\right)$: small p-values are enriched at lower distances. \textbf{C. Histograms of p-values stratified by the covariate:} upon partitioning our hypotheses at the boundaries denoted by dashed lines in panel B, we observe that at small distances the signal (peak at the left of the histograms) is pronounced, while for larger distances the histogram is dominated by background (uniform distribution of p-values from true null hypotheses).\label{fig:H3K27ac_histograms}}
\end{figure}

%------------------------------------------------------------
\subsection{Independent Hypothesis Weighting}
\label{subsec:intro_ihw}
%------------------------------------------------------------
In this paper, we present Independent Hypothesis Weighting (IHW), a flexible framework that can leverage hypothesis heterogeneity to improve power, while retaining finite-sample type-I error control. To explain the method, consider testing $m$ hypotheses $H_1,\dotsc,H_m$ based on p-values $P_1,\dotsc,P_m$ in the situation where we also have access to covariates $X_1,\dotsc,X_m$ such that each $X_i$ is independent of the p-value $P_i$ if $H_i$ is a null hypothesis; the codomain of the $X_i$ can be any space (the same for all $i$). We propose to use a decision rule of the following form in place of~\eqref{eq:rejection}:

\begin{equation}
\label{eq:crossweighted_rejection}
\text{Reject hypothesis } i  \;\; \Longleftrightarrow \;\; P_i \leq \hat{t} \cdot \widehat{W}^{-\ell}(X_i) \quad\quad  \text{where } i \in I_{\ell},
\end{equation}
and $I_{\ell}, \;\ell=1,\dotsc,K$ is a partition of the hypotheses into $K$ disjoint folds, such that the $(P_i,X_i)$ pairs are independent across folds.

There are two salient features to this rule: first, the decision boundary of hypothesis $i$ does not only depend on its p-value $P_i$ and the overall cutoff $\hat{t}$, but also on the weight function $\widehat{W}^{-\ell}:{\mathcal X}\to\mathbb{R}_{\ge0}$ of the covariate $X_i$, where ${\mathcal X}$ is the codomain of the $X_i$, and there is one such function for each fold $I_{\ell}$. Second, the notation $\widehat{W}^{-\ell}$ is used to denote that each of these functions is learned from the data with the proviso that only p-values from the $K-1$ folds excluding $I_{\ell}$ are used. We call this proviso \emph{cross-weighting}.

Conceptually, cross-weighting is related to cross-fitting \citep{schick1986asymptotically}, a method that has been successful in the fields of causal inference \citep{nie2017learning, chernozhukov2017double} and empirical Bayes \citep{ignatiadis2019covariate} for estimation with high-dimensional nuisance parameters. Analogous to findings in the cross-fitting literature, we will show that naively using plug-in estimators to obtain the  weight function tends to overfit, but cross-weighting salvages this at essentially no cost.

%------------------------------------------------------------
\subsection{Related work}
%------------------------------------------------------------
Previous work has shed light on optimal discovery thresholds in heterogeneous multiple testing. Similar to \eqref{eq:crossweighted_rejection}, these thresholds may take the form $\{P_i \leq \hat{t} w_i\}$ parametrized in terms of weights $w_i$ that are optimal for controlling the family-wise error rate (FWER) \citep{roeder2009genome, pena2011power, dobriban2015optimal} or the false discovery rate ($\FDR$) \citep{roquain2009optimal, durand2019adaptive}. Furthermore, in the case of $\FDR$ control, optimal decision thresholds are known to take the form of contours of equal local false discovery rate \citep{cai2009simultaneous, cai2016cars, efron2010large, ferkingstad2008unsupervised, ochoa2014beyond, ploner2006multidimensional, scott2014false}. Nevertheless, all of these optimal procedures are not implementable, as they depend on unknown properties of the data-generating mechanism. Instead, it has been proposed to apply a plug-in principle: the thresholds are estimated from the data at hand.

Such plug-in approaches however have no guarantees of type-I error control or only do so in an asymptotic limit, as the number of tested hypotheses goes to infinity \citep{cai2009simultaneous, cai2016cars, durand2019adaptive, ignatiadis2016data}. More importantly, with finite samples, these plug-in methods often exceed the claimed type-I error; we will  demonstrate this in Sections~\ref{sec:ihw_gbh} and~\ref{sec:numerical_study}. This has motivated the provision of case-by-case, ad-hoc modifications, which however still do not provide finite-sample guarantees. For example, \citet{durand2017adaptive} recommends conducting a global test first and only proceeding with multiple testing if the global null hypothesis can be rejected. \citet*{cai2016cars} use a conservative modification of the density estimator employed by their (asymptotically valid) plug-in approach and show that this controls $\FDR$ in simulations with sparse signals. Furthermore, they suggest using the global screen of \citet{durand2017adaptive} first. \citet*{ignatiadis2016data} use cross-weighting (described above) as a heuristic to maintain $\FDR$ control in finite samples.

Dispensing with heuristics, several authors have recently provided procedures that are formally justified under full independence of the hypotheses: \citet{li2019multiple} propose SABHA, a data-driven, weighted procedure for $\FDR$ control which directly confronts potential overfitting. The authors prove finite sample $\FDR$ control at an elevated level compared to the nominal $\alpha$; i.e., at $(1+\varepsilon)\alpha$ for some $\varepsilon > 0$. However, their guarantee only applies for their specific weighting scheme, which furthermore is suboptimal even under knowledge of the data-generating process \citep{lei2016adapt}. \citet*{zhang2017neuralfdr} and \citet*{zhang2018adafdr} use a variant of hypothesis splitting to guarantee high-probability bounds on the false discovery proportion, however their proposals require a minimum number of rejections; otherwise an empty list of discoveries is declared. Closer to our approach is AdaPT \citep{lei2016adapt}, which uses covariate information to learn covariate-modulated decision boundaries and provides finite sample $\FDR$ guarantees. Its construction is based on a variant of the optimal stopping theorem developed by \citet*{barber2015controlling}, which provides the analyst with considerable flexibility in learning these boundaries from the data, while masking information that could lead to overfitting. However, AdaPT has no theoretical guarantees outside of full p-value independence, is tied to $\FDR$ control and suffers from a large variance of the false discovery proportion \citep{korthauer2018practical}.

Here we propose a general and flexible framework that goes beyond these previous approaches. We formalize hypothesis weighting with weights as a function of covariates $X_i$ and demonstrate that such weights can be learned from the data without overfitting (i.e., losing type-I error control) if we use  cross-weighting as in~\eqref{eq:crossweighted_rejection}. Hence we build upon the hypothesis-splitting idea of \citet{ignatiadis2016data} and demonstrate that it can be used not merely as a heuristic, but instead as a theoretically grounded and principled way of conducting multiple testing with side-information that has far reaching applications. The Independent Hypothesis Weighting method provides finite sample guarantees for multiple type-I measures, such as the $\FDR$, the FWER and the $k$-FWER, unlike previous proposals that are tied to the $\FDR$. IHW provides a clean way to deal with dependent settings, as it allows arbitrary dependence within folds. Finally, IHW provides the researcher with flexibility in choosing any weighting scheme that would be appropriate for the data at hand, but we also recommend a default scheme and provide a software implementation in the form of an R package.

%------------------------------------------------------------
\subsection{Outline}
%------------------------------------------------------------
In Section~\ref{sec:control_guarantees}, we provide an overview of weighted multiple testing and explain our proposal in the context of $\FDR$ control under full independence of hypothesis tests. Section~\ref{sec:IHW_FDR_dependence} extends the results to dependence, and to control of the $k$-FWER. Section~\ref{sec:powerful_weighting_rules} describes a framework for learning weighting rules. Section~\ref{sec:numerical_study} provides simulation results, and Section~\ref{sec:hQTLexample} presents the high-throughput biology example from Figure~\ref{fig:H3K27ac_histograms}. Section~\ref{sec:related_work} discusses further relationships to previous work, and Section~\ref{sec:disc} concludes with a discussion.

% ------------------------------------------------------------
\section{Weighted and cross-weighted multiple testing}
\label{sec:control_guarantees}
%------------------------------------------------------------
A multiple testing procedure operates on data for $m$ hypotheses $H_1,\dotsc,H_m$ and declares $R$ hypotheses as rejections (``discoveries''). Among these, $V$ will be nulls, i.e., the procedure will commit $V$ type I errors. The goal is to make as many discoveries as possible while retaining (stochastic) guarantees that $V$ is acceptable. Concretely, one possible objective is to control the family-wise error rate, defined as FWER $\coloneqq \PP[V\geq 1]$, or the $k$-FWER $\coloneqq \PP[V\geq k]$. In exploratory situations, a typically less stringent objective is to control the false discovery rate ($\FDR$), i.e., the expectation of the false discovery proportion (FDP), namely $\FDR \coloneqq \EE[\text{FDP}] \coloneqq \EE\left[\frac{V}{R \lor 1}\right]$ \citep{benjamini1995controlling}.

Typically the data for each hypothesis are summarized into a single number, the p-value $P_i$, and a rule of form~\eqref{eq:rejection} is applied. However, in the presence of heterogeneity across tests, it might be suboptimal to use such a decision rule that treats all hypotheses exchangeably. Weighted multiple testing \citep*{genovese2006false} is a flexible way of encoding prior information and differentially prioritizing the hypotheses. Multiple testing weights are defined as non-negative numbers $w_i$ such that $\sum_{i=1}^m w_i/m = 1$. Then, a weighted multiple testing decision rule takes the following form:

\begin{equation}
\label{eq:weighted_rejection}
\text{Reject hypothesis } i  \;\; \Longleftrightarrow \;\; P_i \leq \min\{ w_i \cdot \hat{t}, \tau\}
\end{equation}
Here $\tau \in (0,1]$ is a fixed number, of which more below, and as in~\eqref{eq:rejection}, the cutoff $\hat{t}$ may be data-driven. A larger $w_i$ implies that it is easier to reject hypothesis $i$. We first review two procedures for choosing $\hat{t}$.

\begin{defn}[Weighted $k$-Bonferroni]
\label{defn:wbonf}
The $k$-FWER  can be controlled at level $\alpha \in (0,1)$ by applying the weighted $k$-Bonferroni procedure \citep{romano2010balanced}, which takes the form~\eqref{eq:weighted_rejection} with deterministic cutoff $\hat{t} = k\alpha/m$ and $\tau =1$. The case $k=1$ is the weighted Bonferroni procedure proposed by \citet{genovese2006false}.
\end{defn}

\begin{defn}[$\tau$-censored, weighted Benjamini-Hochberg]
\label{defn:tau-wbh}
The $\FDR$  can be controlled at level $\alpha \in (0,1)$ by applying the $\tau$-censored, weighted Benjamini-Hochberg procedure, which takes the form~\eqref{eq:weighted_rejection} with $\tau \in (0,1]$ fixed and data-driven cutoff $\hat{t}$ specified as:

\begin{equation}
\hat{t} = \frac{\alpha \hat{k}}{m},\;\; \hat{k} = \max\left\{ k\in \mathbb N_{\geq 0} \mid P_i \leq \left(\frac{\alpha w_i k}{m} \right) \land \tau  \text{ for at least } k \text{ p-values} \right\}
\end{equation}
\end{defn}

The weighted Benjamini-Hochberg (BH) procedure of \citet*{genovese2006false} is the special case $\tau=1$. The more general form was proposed by \citet{li2019multiple} and will be employed for our theoretical guarantees in the following. The number of rejections of $\tau$-censored BH is non-decreasing in $\tau$, so that a procedure with smaller $\tau$ will never make more discoveries. However, for large $\tau$, say $\tau \geq 0.5$, the discovery set will be equal to that with $\tau=1$, as long as weighted BH with $\tau=1$ did not reject a p-value $\geq 0.5$.

In decision rule~\eqref{eq:weighted_rejection}, the weights $w_i$ are denoted by lower-case letters. This reflects the fact that existing results treat these weights as deterministic---as prior knowledge that a researcher has to specify before seeing the p-values \citep*{genovese2006false, blanchard2008two,habiger2014weighted,roquain2009optimal, ramdas2017unified}.  The main goal of this work is to let the weights depend on the data at hand---they are thus denoted as random variables $W_i$---while providing finite-sample guarantees. Such data-dependent weighting has been recognized as an important open problem \citep{benjamini2008, roquain2009optimal} that is essential for dealing with large scale multiple testing. To the best of our knowledge, no solution has been provided so far. Existing proposals for data-driven weighting either explicitly account for overfitting by establishing $\FDR$ control at an elevated level compared to nominal \citep{li2019multiple} or only provide guarantees in the asymptotic limit \citep*{hu2010false, ignatiadis2016data, durand2019adaptive, zhao2014weighted, wang2018weighted, roeder2007improving}.

%---------------------------------------------------------------------------------
\subsection{Example: Group Benjamini-Hochberg with cross-weighting}
\label{sec:ihw_gbh}
%---------------------------------------------------------------------------------
We first provide a rudimentary version of our method that is applicable to situations with categorical (or suitably categorized) covariates $X_i \in \cb{1,\dotsc,G}$. This setting is called multiple testing with groups; each group consists of hypotheses whose covariate $X_i$ takes on the same value. Our method builds upon the Group Benjamini-Hochberg (GBH) method proposed by \citet{hu2010false} to improve power compared to BH by using the group structure. GBH consists of first estimating the proportion of null hypotheses $\pi_0(g)$ in each group by $\widehat{\pi}_0(g)$, weighting each hypothesis proportionally to $(1-\widehat{\pi}_0(g))/\widehat{\pi}_0(g)$ and finally applying the weighted BH procedure. Algorithm~\ref{alg:GBH} describes the method in detail\footnote{A simplification is that in Algorithm~\ref{alg:GBH}, the weights are specified so that $\sum_i W_i = m$. In contrast, in the original GBH paper \citep{hu2010false}, the weights are less conservative and satisfy $\sum_i \widehat{\pi}_0(X_i) W_i = m$. This inflation ensures that in the oracle case of known $\pi_0(\cdot)$, the $\FDR$ of GBH is exactly equal to $\alpha$. We return to the issue of null proportion adaptivity in Section~\ref{subsec:ihw_fdr_control} and Theorem~\ref{thm:IHW-Storey}; in the case of GBH it may be regained by employing the optional step in Algorithm~\ref{alg:GBH}, cf.~\citet{ramdas2017unified}.}, using the estimator of \citet{storey2004strong} applied to the grouped setting, analogous to~\citet{sankaran2014structssi}.

\begin{figure}
\centering
\includegraphics[width=0.8\textwidth]{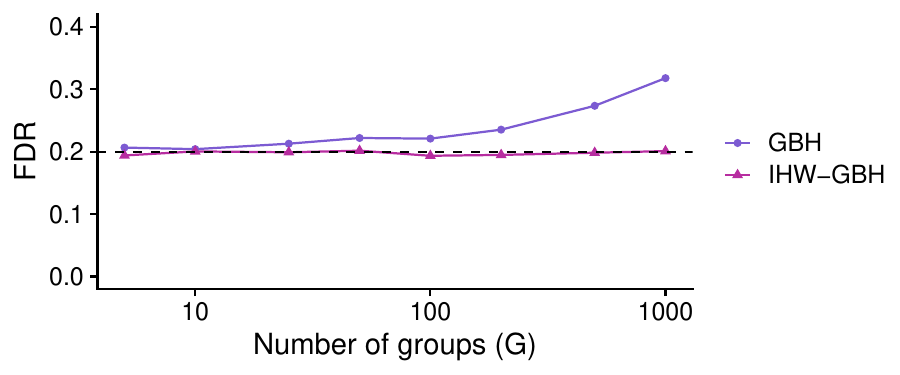}
\caption{\normalfont\textbf{The need for cross-weighting:}
We simulated under the global null with $m=10000$ independent $P_i \sim U[0,1]$ and $X_i \equiv i \pmod{G}$, where the number of groups $G$ is a simulation parameter shown on the $x$-axis. Then we applied the GBH and IHW-GBH (with a random partition into 5 folds) methods (described in Algorithms~\ref{alg:GBH} and \ref{alg:IHW-GBH}, with $\tau=0.5$ and without the null-proportion adaptivity step) at level $\alpha=0.2$. The plot shows the $\FDR$ (obtained by averaging over 12000 Monte Carlo replicates) versus $G$. GBH does not control the $\FDR$, and $\FDR$ increases as $G$ increases, while IHW-GBH controls the $\FDR$ for all $G$.}
\label{fig:null_gbh}
\end{figure}

\citet*{hu2010false} provide the following guarantees for GBH: in the oracle situation where the $\pi_0(g)$ are known, GBH controls the $\FDR$. In the asymptotic limit where the number of groups is fixed, the number of hypotheses in each group grows to infinity and $\plim_{m \to \infty}\widehat{\pi_0}(g) \geq \pi_0(g)$ for all $g$, GBH controls the $\FDR$. Furthermore, sufficient conditions are given so that asymptotically GBH is at least as powerful as BH. The asymptotics, however, do not necessarily apply for finite {$m/G$, the number of hypotheses per group, as shown by  simulations summarized in Figure~\ref{fig:null_gbh}. Intuitively, the reason is that some groups will randomly be enriched for smaller than expected p-values (and some for larger than expected ones), and the method further up-weights the former set of null p-values.

\LinesNotNumbered
%\IncMargin{6em}

% enable numbering again
%\DecMargin{6em}
\newcommand{\hrulealg}[0]{\vspace{1mm} \hrule \vspace{1mm}}
\noindent
\begin{minipage}[t]{0.46\textwidth}
  \vspace{0pt} 
\begin{algorithm}[H]
    \label{alg:GBH}
    \caption{The Group Benjamini-Hochberg (GBH) algorithm}
  \Input{%
    $(P_1,\dotsc,P_m) \in [0,1]^m$\newline
    $(X_1,\dotsc,X_m) \in \cb{1,\dotsc,G}^m$\newline \newline
    a nominal level $\alpha \in (0,1)$\newline
    a censoring level $\tau \in (0,1)$}
    \hrulealg
    \For{$g = 1,\dotsc,G$}{
     $$\widehat{\pi}_0(g) := \frac{1+\sum_{i: X_i=g} \ind\p{P_i > \tau}}{\abs{\cb{i: X_i=g}}(1-\tau)}\; \land \;1$$
    }
    \For{$i = 1,\dotsc,m$}{
     $$ W_i := \frac{1-\widehat{\pi}_0(X_i)}{\widehat{\pi}_0(X_i)}\bigg /\sum_{i=1}^m  \frac{1-\widehat{\pi}_0(X_i)}{m \cdot \widehat{\pi}_0(X_i)} $$
    }\;
    \hrulealg
     \textbf{Optional (null prop. adaptivity):} $$\hat{\pi}'_{0, W}  := \frac{ \displaystyle\max_{i=1,\dotsc,m} W_i + \sum_{i=1}^m W_i \ind(P_i > \tau)}{m(1-\tau)}$$\\
     Update $W_i := W_i/\hat{\pi}'_{0,W}$\\
     \hrulealg
     \vspace{0.31cm}
    $\;$Apply weighted BH (Def.~\ref{defn:tau-wbh}) with p-values $P_i$ and weights $W_i$.
  \end{algorithm}
\end{minipage}
\begin{minipage}[t]{0.54\textwidth}
  \vspace{0pt} 
\begin{algorithm}[H]
    \label{alg:IHW-GBH}
    \caption{The cross-weighted GBH  (IHW-GBH) algorithm}
  \Input{%
    $(P_1,\dotsc,P_m) \in [0,1]^m$\newline
    $(X_1,\dotsc,X_m) \in \cb{1,\dotsc,G}^m$\newline 
    a partition $I_1, \dotsc, I_K$ of $\cb{1,\dotsc,m}$
    \newline
    a nominal level $\alpha \in (0,1)$\newline
    a censoring level $\tau \in (0,1)$}
    \hrulealg
    \For{$\ell=1, \dotsc, K$}{
    \For{$g = 1,\dotsc,G$}{
     $$\widehat{\pi}_0^{-\ell}(g) :=  \frac{1+\sum_{i \notin I_{\ell}: X_i=g} \ind\p{P_i > \tau}}{\abs{\cb{i \notin I_{\ell}: X_i=g}}(1-\tau)} \; \land \;1$$
    }   \For{$i \in I_{\ell}$}{
     $$ W_i := \frac{1-\widehat{\pi}_0^{-\ell}(X_i)}{\widehat{\pi}^{-\ell}_0(X_i)}\bigg /\sum_{ \smash{i \in I_{\ell}}}  \frac{1-\widehat{\pi}_0^{-\ell}(X_i)}{|I_{\ell}| \cdot \widehat{\pi}_0^{-\ell}(X_i)} $$ \\
     \hrulealg
     \textbf{Optional (null prop. adaptivity):} $$\hat{\pi}'_{0,W,{\ell}}  := \frac{ \displaystyle\max_{i \in I_{\ell}} W_i + \sum_{i \in I_{\ell}} W_i \ind(P_i > \tau)}{|I_{\ell}|(1-\tau)}$$\\
     Update $W_i := W_i/\hat{\pi}'_{0,W,{\ell}}$\\
     \hrulealg
    }}
    $\;$Apply the $\tau$-censored, weighted BH procedure (Def.~\ref{defn:tau-wbh}) with p-values $P_i$ and weights $W_i$.
  \end{algorithm}
\end{minipage}

Our solution is to use cross-weighting. We assign each hypothesis to one of $K$ folds -- randomly and independently of its p-value $P_i$ and covariate $X_i$-- and then calculate weights out-of-fold, as elaborated in Algorithm~\ref{alg:IHW-GBH}. With cross-weighting, a null p-value that is small by chance cannot lead to an upweighting of itself. $\FDR$ control is restored, as shown in Figure~\ref{fig:null_gbh}. On the other hand, if the weights are determined not just by noise, but by true signal, then IHW-GBH, just as GBH, has increased power compared to BH, as we show in a more comprehensive simulation study in Section~\ref{subsec:group_mtp}. If $G$ furthermore remains fixed as $m \to \infty$, then GBH and IHW-GBH are asymptotically equivalent (Corollary~\ref{cor:ihw_gbh_asymptotic}).

%------------------------------------------------------------------------
\subsection{IHW: A family of multiple testing procedures}
%------------------------------------------------------------------------

\begin{algorithm}[t]
    \label{alg:IHW_general}
    \caption{The general IHW algorithm}
  \Input{%
    $\mathbf{P}=(P_1,\dotsc,P_m) \in [0,1]^m$\newline
    $\mathbf{X}=(X_1,\dotsc,X_m) \in \mathcal{X}^m$\newline 
    a partition $I_1, \dotsc, I_K$ of $[m]$
    \newline
    a nominal level $\alpha \in (0,1)$}
    \hrulealg
    \For{$\ell=1, \dotsc, K$}{
      Learn a weight function $\widehat{W}^{-\ell}: \mathcal{X} \to \RR_{\geq 0}$ from the pairs $(P_i, X_i), \;i\notin I_{\ell}$\;
      \For{$i \in I_{\ell}$}{
       Let $W_i$ be a suitable rescaling of $\widehat{W}^{-\ell}(X_i)$,
        $$
        W_i := \frac{\abs{I_{\ell}} \widehat{W}^{-\ell}(X_i)}{\sum_{i \in I_{\ell}} \widehat{W}^{-\ell}(X_i) },
        \text{ if }\;\;\sum_{i \in I_{\ell}} \widehat{W}^{-\ell}(X_i) >0,  \text{ else } W_i :=1
        $$
    }}
    Run a weighted multiple testing procedure with p-values $P_i$ and weights $W_i$.
\end{algorithm}
We now generalize the IHW-GBH procedure beyond categorical covariates, the GBH weighting scheme and the weighted BH procedure~(Def.~\ref{defn:tau-wbh}): we seek a general way of applying weighted multiple testing methods with \emph{data-driven} weights $W_i$ when covariates $X_i$---not necessarily categorical---are available. Our approach consists of two ingredients: first, we only consider weights that are functions of the covariates $X_i$, i.e., $W_i = W(X_i)$. The second ingredient is \emph{cross-weighting}: we partition our $m$ hypotheses into $K$ disjoint folds\footnote{Our baseline proposal is to construct the partition by splitting the set $[m]=\{1,\ldots,m\}$ into $K$ (the default in the IHW software package is $K=5$) equally sized folds randomly. Alternatively, domain specific knowledge can be used to derive folds that minimize across-fold dependence, cf.\ the example in Section~\ref{sec:hQTLexample}.} $I_1,\dotsc,I_K$. Then, in determining the weight $W_i$ for hypothesis $i \in I_{\ell}$, we set $W_i \propto \widehat{W}^{-\ell}(X_i)$, where the weight function $\widehat{W}^{-\ell}$ is learned from data \emph{outside} fold $I_\ell$ and the weights are normalized, typically such that $\sum_{i \in I_{\ell}} W_i = \abs{I_{\ell}}$. This overall framework is summarized in Algorithm~\ref{alg:IHW_general}.

In Sections~\ref{subsec:ihw_fdr_control} and~\ref{subsec:kfwer_control} we provide formal guarantees of finite-sample type-I error control for the IHW algorithm, under the condition that the weighted multiple testing procedure is weighted BH with $\tau$-censoring or weighted $k$-Bonferroni. We will discuss how to learn weight functions for general (non-categorical) covariates in Section~\ref{sec:powerful_weighting_rules}.
%-------------------------------------------------------------------------------
\subsection{Finite-sample FDR control with cross-weighting under independence}
\label{subsec:ihw_fdr_control}
%-------------------------------------------------------------------------------
To derive formal guarantees for Algorithm~\ref{alg:IHW_general}, we set out with a sufficient distributional assumption that contains several independence relationships. In Section~\ref{sec:IHW_FDR_dependence}, we will consider more general dependence structures.

\begin{assumption}[Distributional setting under independence]
\label{assumption:distrib}
Let $(P_i, X_i)$, $i\in [m]$ be\footnote{We use the notation $[m] = \cb{1,\dotsc,m}$.} (p-value, covariate) pairs and $\Hnull \subset [m]$ be the index set of null hypotheses. We assume that:
\begin{enumerate}[label=(\alph*)]
\item[(a$_1$)] The null pairs $((P_i, X_i))_{i \in \Hnull}$ are jointly independent.
\item[(a$_2$)] The null pairs $((P_i, X_i))_{i \in \Hnull}$ are independent of the alternative pairs $((P_i, X_i))_{i \notin \Hnull}$.
\item[(b )] For $i \in \Hnull$, it holds that $P_i$ is independent of $X_i$.
\item[(c )]  For $i \in \Hnull$, $P_i$ is super-uniform, i.e., $\PP[P_i \leq t] \leq t$ for all $t\in[0,1]$.
\end{enumerate}
\end{assumption}

To parse this assumption, let us first consider two important special cases: (i) marginalizing over the $X_i$, so that we only have access to p-values, and (ii) deterministic $X_i$. In both cases, Assumption~\ref{assumption:distrib} reduces to (a$_1$') $(P_i)_{i \in \Hnull}$ are jointly independent, (a$_2$') independent of the alternative p-values $(P_i)_{i \notin \Hnull}$ and (c). Of these, (a$_1$') and (a$_2$'), while admittedly strong, are a typical starting point for proving finite-sample results for multiple testing procedures, even in the absence of covariates: \citet{liang2012adaptive} call it the null independence assumption. In the setting with covariates, these are also assumptions made by \citet[Theorem 1]{li2019multiple} and \citet[Theorem 1]{lei2016adapt}. \citet{cai2016cars} also assume full independence of hypotheses. The super-uniformity; also called conservativeness,  of the null p-values (c) is also a standard assumption in multiple testing \citep{blanchard2008two}. \citet{li2019multiple} make a stronger assumption than (c). 

The case of deterministic $X_i$ is important, since for example the genomic distance between SNPs and peaks in our motivating example in Figure~\ref{fig:H3K27ac_histograms} is a deterministic covariate. See \supplementname~\ref{subsec:domain_spec_covariates} for additional examples. Nevertheless, we formulate results for the more general case to also handle situations in which the covariate $X_i$ is calculated from the same data that are used to calculate the p-value $P_i$. For instance, \citet{cai2016cars} consider simultaneous two-sample testing, and construct an ancillary $X_i$ that is independent of the $t$-statistic (and thus also the p-value) under the null hypothesis; we revisit their construction in the simulation study of Section~\ref{subsec:cars_sims}. Assumption~\ref{assumption:distrib}(b) is crucial in ensuring that knowledge of $X_i$ does not influence the null distribution. \citet{cai2016cars} call it a "principle for information extraction"; cf.\ \citet{bourgon2010independent, boca2018direct} for further elaborations on this assumption and \supplementname~\ref{subsec:stat_covariates} for more examples of random covariates.

Next, we state two specifications on the weighting mechanism used. Unlike Assumption~\ref{assumption:distrib}, the applicability of which depends on the generally unknown data-generating mechanism, these are entirely under the control of the analyst.

\begin{specif}[Honest weighting]
\label{assumption:honest_weights}
Consider a partition of $[m]$ into $K$ folds $I_1,\dotsc,I_K$, i.e., $\bigcup_{\ell} I_{\ell} = [m]$ and $\left(I_{\ell}\right)_\ell$ are disjoint, and define $I_{\ell}^c = [m] \setminus I_{\ell}$. The partition is assigned independently of $\allpairs$. Then, the data-driven weights $(W_i)_{i \in [m]}$ are honest with respect to the partition $I_1,\dotsc,I_K$ if:

\begin{enumerate}[label=(\alph*)]
\item $W_i$ is a function of only $(P_j)_{j \in I_{\ell}^c}$ and $(X_j)_{j \in [m]}$ for all $\ell \in [K]$ and all $i \in I_{\ell}$.
\item The weights in fold $I_{\ell}$ average to $1$, i.e., $\sum_{i \in I_{\ell}} W_i = |I_{\ell}|$ for all $\ell \in [K]$.
\item $W_i \geq 0$ for all $i$.
\end{enumerate}
\end{specif}
We call this specification ``honest weighting'', borrowing terminology from the honest tree construction of \citet{wager2018estimation}, who call a regression tree honest if the set of observations used to determine its structure is disjoint from the set of observations used for prediction in the leaves. Specification~\ref{assumption:honest_weights} encapsulates our idea of cross-weighting. Informally, it says that the weight $W_i$ of hypothesis $i$ should not depend on its p-value $P_i$. As already shown in Figure~\ref{fig:null_gbh}, without honesty it is easy to overfit the data. Part (b) of the definition encapsulates a fixed weighting budget \citep*{genovese2006false}. Instead of merely requiring $\sum_{i=1}^m W_i =m$, the budget is restricted within each fold, to prevent information leakage across folds through the total magnitude of the weights.

Honesty suffices to guarantee type-I error control in some cases, for example for the weighted $k$-Bonferroni procedure (Section~\ref{subsec:kfwer_control} and Theorem~\ref{thm:IHW-bonf}). However, for the $\tau$-censored, weighted BH procedure with data-driven weights, we require one further condition on the weights, which was proposed by \citet{li2019multiple} and states that the magnitude of p-values less than or equal to $\tau$ must be concealed from the weighting algorithm.

\begin{specif}[$\tau$-censored weighting]
\label{assumption:tau_stopped_weights}
The weights $W_i$ are called $\tau$-censored for $\tau \in (0,1]$ if they depend on the p-values $\allpvalues$ only through $(P_i\;\ind(P_i > \tau))_{i \in [m]}$.
\end{specif}
We are ready to state the first result:
\begin{thm}[IHW-BH controls the $\FDR$ under honesty and $\tau$-censored weighting]\label{thm:IHWc} Let $\allpairs$ satisfy Assumption~\ref{assumption:distrib}. Furthermore assume that we construct data-driven weights $W_i$ that are honest (Specification~\ref{assumption:honest_weights}) and $\tau$-censored (Specification~\ref{assumption:tau_stopped_weights}) for some $\tau \in (0,1]$. Then the $\tau$-censored, weighted BH procedure (Definition~\ref{defn:tau-wbh}) with p-values $P_i$ and weights $W_i$ controls the $\FDR$ at the nominal level $\alpha$.
\end{thm}
The intuition for this theorem is the following: in the weighted BH algorithm (Definition~\ref{defn:tau-wbh}), the rejection threshold of a null p-value $P_i$ depends on its weight $W_i$ and the total number of rejections $R$. Assumption~\ref{assumption:distrib} and honest weighting (Specification~\ref{assumption:honest_weights}) ensure that a null p-value cannot influence its own weight. However, tests can coordinate adversarially by weighting each other in a way that increases $R$ and potentially leads to their own rejection. \supplementname~\ref{sec:counterexample} provides an example of how such adversarial coordination can break $\FDR$-control guarantees, even though honesty holds. However, under $\tau$-censoring, the only p-values that can coordinate through weight assignment are the ones $>\tau$. These p-values are also excluded from being rejected and so $\FDR$ control is restored.\\ 

\noindent As a corollary, we get the following result:

\begin{cor}[IHW-GBH controls the $\FDR$]
\label{cor:ihw_bl_fdr}
Let $\allpairs$ satisfy Assumption~\ref{assumption:distrib}, then the IHW-GBH procedure (without the null proportion adaptivity step) described in Algorithm~\ref{alg:IHW-GBH} controls the $\FDR$ at the nominal level $\alpha$.
\end{cor}

\begin{proof}
By construction, the weights $W_i$ of IHW-GBH are honest and $\tau$-censored.
\end{proof}

A shortcoming of IHW-BH with weights that satisfy $\sum_{i=1}^m W_i = m$ is that FDR is controlled at $\pi_{0,W}' \alpha \leq \alpha$, where $\pi'_{0,W} \coloneqq (\sum_{i \in \Hnull} \EE[W_i]) / m$ and IHW-BH can thus be needlessly conservative. Motivated by null-proportion adaptive methods for unweighted BH~\citep*{storey2004strong} and weighted BH with deterministic weights~\citep*{habiger2014weighted,ramdas2017unified}, we estimate $\pi'_{0,W}$ within fold $I_{\ell}$ by
\begin{equation}
\label{eq:pi0_adaptive}
\hat{\pi}_{0,W,\ell}' = \frac{ \displaystyle\left(\max_{i \in I_{\ell}}W_i\right) + \sum_{i \in I_{\ell}} W_i \,\ind(P_i > \tau')}{|I_{\ell}|(1-\tau')} \text{ with } \tau'  \in [\tau,1), \footnote{We suggest $\tau' = 0.5$ as a default choice.}
\end{equation}
\noindent and use these estimates to inflate the weights $W_i$. We have the following result:
\begin{thm}[IHW-Storey controls the $\FDR$ under honesty and $\tau$-censored weighting]
\label{thm:IHW-Storey}
Assume that all assumptions of Theorem~\ref{thm:IHWc} are satisfied. Next let \smash{$\hat{\pi}_{0,W,\ell}'$} be defined as in~\eqref{eq:pi0_adaptive} and define null-proportion adaptive weights as \smash{$W_i^{\text{Storey}} := W_i\,/\,\hat{\pi}_{0,W,\ell}'$} for $i \in I_{\ell}$. Then the $\tau$-censored, weighted BH procedure (Definition~\ref{defn:tau-wbh}) with p-values $P_i$ and weights \smash{$W_i^{\text{Storey}}$} controls the $\FDR$ at the nominal level $\alpha$.
\end{thm}

\noindent A direct application of this theorem is that the statement of Corollary~\ref{cor:ihw_bl_fdr} also holds for the null-proportion adaptive version of IHW-GBH (cf.\ Algorithm~\ref{alg:IHW-GBH}). This provides power gains in situations where the null proportion is substantially smaller than 1 at least in some regions of the covariate space, since then it will be the case that $\sum  W_i^{\text{Storey}} > \sum  W_i$, thus increasing the total weight budget. 

%-----------------------------------------------------------------------
\subsection{FDR asymptotics with cross-weighting under independence}
\label{subsec:asymptotics}
%-----------------------------------------------------------------------
While the primary focus of this paper is on finite-sample guarantees and performance in simulations, in this section we provide asymptotic results for $m \to \infty$ that serve three purposes: first, they demonstrate how cross-weighting enables a streamlined proof of asymptotic $\FDR$ control under standard assumptions on $(P_i,X_i)$ while dispensing of requirements on the class of weight functions. Second, they show that in situations in which there is sufficient signal and the data-driven weight function has approached its asymptotic limit, no power is lost by using cross-weighting. Third, they show that in an asymptotic regime, IHW-BH controls the $\FDR$ without a need for $\tau$-censoring (Specification~\ref{assumption:tau_stopped_weights}). On the other hand, our aim here is not to provide the sharpest asymptotics under the weakest conditions, but just to provide these conceptual insights.

We develop the asymptotics using the following Bayesian model~\citep{ferkingstad2008unsupervised, lei2016adapt, deb2018two}, which we call the conditional two-groups model and which extends the two-groups model of \citet{storey2003positive} and \citet*{efron2001empirical}:
\begin{equation}
\label{eq:conditional_twogroups}
\begin{aligned}
&X_i \sim \PP^X,\;\; H_i \mid (X_i=x) \sim \text{Bernoulli}(1-\pi_0(x)),\\
&P_i \mid (H_i = 0, X_i=x) \sim U[0,1],\;\; P_i \mid (H_i = 1, X_i=x) \sim F_{\text{alt}}(\cdot \mid X_i=x)
\end{aligned}
\end{equation}
We also define $F(t \mid X_i=x) = \pi_0(x)t + (1-\pi_0(x))F_{\text{alt}}(t \mid X_i=x)$: the distribution of $P_i$ given $X_i=x$. The distribution $F(t \mid X_i=x)$ can vary from test to test because of varying null probabilities $\pi_0(x)$ and/or alternative distributions $F_{\text{alt}}(\cdot \mid X_i=x)$, depending on the value of its covariate $X_i$.

Since $m$ is a changing parameter in the asymptotics, it is useful to formalize what ``learning a weight function'' entails and use more involved notation:
\begin{specif}[Weighting scheme] A weighting scheme \smash{$\widehat{W}^{(\cdot)}$} is a mechanism that, for any finite subset $I \subset \mathbb N_{>0}$, uses samples $((P_i, X_i))_{i \in I}$ to learn a weight function \smash{$\widehat{W}^{(I)}: \mathcal{X} \to \RR_{\geq 0}$}. We assume that the learned weight function \smash{$\widehat{W}^{(I)}$} does not excessively upweight individual hypotheses, i.e., there exists $\Gamma < \infty$ such that
\begin{equation}
\label{eq:weights_technical_condition}
\int\widehat{W}^{(I)}(x)^2 d\PP^X(x) \leq \Gamma \cdot \left(\int\widehat{W}^{(I)}(x) d\mathbb P^X(x)\right)^2\quad\quad\text{ for all subsets } I \subset \mathbb N.
\end{equation}
\label{specif:wt_scheme}
\end{specif}
Given $m$ independent draws $(P_i,X_i)$ from~\eqref{eq:conditional_twogroups} and a weighting scheme (Specification~\ref{specif:wt_scheme}), we seek to apply learned weights in conjunction with weighted BH~(Definition~\ref{defn:tau-wbh}). We consider two possibilities:
\begin{enumerate}[nosep]
    \item \textbf{Naive weighted BH:} We use all data $\allpairs$ to learn $\widehat{W}^{([m])}$ and let $W_i \propto \widehat{W}^{([m])}(X_i)$ for $i=1,\dotsc,m$, such that the weights average to $1$ (i.e., $\sum_{i=1}^m W_i = m$). Then we apply the weighted BH procedure with p-values $P_i$ and weights $W_i$.
    \item \textbf{IHW-BH:} We partition $[m]$ into $K$ disjoint folds $I_1,\dotsc,I_K$, independently of $\allpairs$. Then we apply Algorithm~\ref{alg:IHW_general} in conjunction with weighted BH, i.e., for each fold $\ell$, we apply the weighting scheme on $[m]\setminus I_{\ell}$ and for $i\in I_{\ell}$ set weight $W_i \propto \widehat{W}^{([m]\setminus I_{\ell})}(X_i)$ and such that the weights average to $1$ in that fold (i.e., $\sum_{i \in I_{\ell}} W_i=1$). Then we apply weighted BH with p-values $P_i$ and weights $W_i$. We note that the data-driven weights $W_i$ are honest (Specification~\ref{assumption:honest_weights}) by construction. However, for the asymptotics, we do not require $\tau$-censoring~(Specification~\ref{assumption:tau_stopped_weights}), but instead require the mild technical condition~\eqref{eq:weights_technical_condition}.
\end{enumerate}

\begin{prop} Let $(P_i,X_i)$ be i.i.d.\ from the conditional two-groups model~\eqref{eq:conditional_twogroups} satisfying regularity Assumption~\ref{assumption:conditional_twogroups_asymptotics} (in \supplementname~\ref{sec:asymp_proof}). If the partition satisfies $|I_{\ell}|/m \to \gamma_{\ell} \in (0,1)$ as $m \to \infty$ for all $\ell$, then\footnote{See \supplementname~\ref{sec:asymp_proof} for the proof and formal statements.}:
\begin{enumerate}[nosep, label=(\alph*)]
    \item There exists a weighting scheme satisfying Specification~\ref{specif:wt_scheme}, such that the naive weighted BH procedure asymptotically does not control the $\FDR$.
    \item For any weighting scheme satisfying Specification~\ref{specif:wt_scheme}, the IHW-BH procedure asymptotically controls the $\FDR$.
    \item Consider a weighting scheme that converges in probability to a deterministic limiting weight function $W^*: \mathcal{X} \to \RR_{\geq 0}$,
\begin{equation*}
\norm{\widehat{W}^{([m])}(\cdot) - W^*(\cdot)}_{\infty} \overset{\PP}{\longrightarrow} 0 \text{  as  } m \to \infty, \;\;\;  \int W^*(x) d\PP^X(x) = 1,\; \int W^*(x)^2 d\PP^X(x)< \infty
\end{equation*}
Then, the naive weighted BH and IHW-BH procedures have the same power asymptotically.
\end{enumerate}
\label{prop:asymp}
\end{prop}

\begin{proof}[Proof idea for (a) and (b):] The proof of \citet{storey2004strong} for asymptotic $\FDR$ control of BH argues that by the Glivenko-Cantelli theorem, $\sup_{t}\abs{\frac{1}{m}\sum_{i=1}^m \sqb{\ind(P_i \leq t) - \PP[P_i \leq t]}} \stackrel{\PP}{\to} 0$ and similarly for the subset of null hypotheses. A consequence is that the BH estimator of the false discovery rate is asymptotically uniformly conservative over all thresholds $\geq \delta>0$, which in turn implies asymptotic $\FDR$ control. Extending this argument to the weighted case requires uniform convergence: \smash{$\sup_{t}\abs{\frac{1}{m}\sum_{i=1}^m \sqb{\ind(P_i \leq t W_i) - \PP[P_i \leq tW_i]}} \stackrel{\PP}{\to} 0$}. 

For data-driven weights, this can be achieved by learning the weight function from a suitably restricted class $\mathcal{W}$. \citet{du2014single, ignatiadis2016data, durand2019adaptive} all use $\mathcal{W}$ such that the functions $\{ (p,x) \mapsto \ind(p \leq t W(x)) \mid t \in (0,1], W(\cdot) \in \mathcal{W} \}$ are $\PP$-Glivenko-Cantelli \citep{van2000asymptotic}. Similarly, \citet{li2019multiple} consider $\mathcal{W}$ with low Rademacher complexity. On the other hand, if convergence is not uniform (e.g., if we are free to choose any weights satisfying Specification~\ref{specif:wt_scheme}), then we can find regions of $\mathcal{X}$-space that are enriched for small p-values merely by chance, upweight them, and violate $\FDR$ control (cf.\ Figure~\ref{fig:null_gbh}).

Instead, through cross-weighting, the richness of $\mathcal{W}$ is irrelevant: upon conditioning on other folds, $P_i/\widehat{W}^{([m]\setminus I_{\ell})}(X_i)$ in fold $I_{\ell}$ are i.i.d., and thus the one-dimensional Glivenko-Cantelli result applies. 
\end{proof}

In words, while data-driven weights can lead to overfitting (a), cross-weighting universally alleviates this (b). A further upshot of (b) is that it dispenses with the requirement for $\tau$-censored weights (Specification~\ref{assumption:tau_stopped_weights}).
Finally, the objection may be raised to cross-weighting that it drops data and should thus be less powerful than a procedure that uses all the data. However, (c) shows that asymptotically one loses no power by using cross-weighting if the weighting procedure is well-behaved, i.e., the weights asymptotically converge to a limit. 

As a corollary of Proposition~\ref{prop:asymp}, we have that:	

\begin{cor}[IHW-GBH asymptotics]
\label{cor:ihw_gbh_asymptotic}
Under the assumptions of Proposition~\ref{prop:asymp} with $\mathcal{X} = [G]$ for fixed $G \in \mathbb N$, the GBH and IHW-GBH procedures without null proportion adaptivity, described in Algorithms~\ref{alg:GBH} and~\ref{alg:IHW-GBH}, have the same power asymptotically.
\end{cor}
\begin{proof}
In \supplementname~\ref{subsec:ihw_gbh_asymptotics_proof}, we verify~\eqref{eq:weights_technical_condition} and the condition from part (c) of Proposition~\ref{prop:asymp}.
\end{proof}

At this point, we note that \citet{durand2019adaptive}, motivated by a preprint version of this work, derived the following related and elegant result: in the setting with $\mathcal{X}$ a finite discrete space, \citet[Theorem 7.1.]{durand2019adaptive} constructs a cross-weighted procedure that asymptotically controls the $\FDR$ and simultaneously achieves the power of the \emph{optimal} weighted procedure.

%----------------------------------------------------------------------------------
\section{Extension to dependence}
\label{sec:IHW_FDR_dependence}
%-----------------------------------------------------------------------------------
\subsection{The key assumption: Independence across folds, dependence within}
%-----------------------------------------------------------------------------------
Assumption~\ref{assumption:distrib} made the strong assumption of joint independence of all null p-values and was sufficient for the results presented in Section~\ref{sec:control_guarantees}. Real data commonly deviate from this assumption. The consequences of such deviations on the applicability of results derived using independence assumptions are typically difficult to reason about. It is therefore desirable to construct guarantees that can be derived from weaker assumptions that are closer to realistic patterns of dependence.

\begin{assumption}[Distributional setting with dependence]
\label{assumption:distrib_dep}

Let $(P_i, X_i)$, $i\in [m]$ be (p-value, covariate) pairs, $I_1,\dotsc, I_K$ be folds of a partition of $[m]$ that is defined based on information independent of $\allpairs$, and let $\Hnull \subset [m]$ the index set of null hypotheses. We assume that:

\begin{enumerate}[label=(\alph*)]
\item[(a)] The (p-value, covariate) pairs are independent across folds $I_1,\dotsc,I_K$, but may be dependent within each fold. Formally, $((P_i, X_i))_{i \in I_{\ell}}, \ell \in [K]$ are jointly independent.
\item[(b)]  For $i \in \Hnull$, it holds that $P_i$ is independent of $(X_j)_{j \in [m]}$.
\item[(c)]  For $i \in \Hnull$, $P_i$ is super-uniform, i.e., $\PP[P_i \leq t] \leq t$ for all $t\in[0,1]$.
\end{enumerate}
\end{assumption}

Let us compare Assumption~\ref{assumption:distrib_dep} to Assumption~\ref{assumption:distrib}. Parts~\ref{assumption:distrib_dep}(b, c) are mild. Part~\ref{assumption:distrib_dep}(c) is identical to \ref{assumption:distrib}(c) and standard in multiple testing. Part~\ref{assumption:distrib_dep}(b) is analogous to \ref{assumption:distrib}(b), albeit stronger, since we are conditioning on the full vector of $X_i$. Nevertheless, \ref{assumption:distrib_dep}(b) is implied by \ref{assumption:distrib}(a,b). In the important case where the $X_i$ are deterministic, \ref{assumption:distrib}(b) trivially holds. But it also allows for situations where, for instance, the $X_i$ are random spatial locations. In this case, we may expect p-values with similar $X_i$ to be correlated. Assumption~\ref{assumption:distrib_dep}(b) then means that knowing the locations $X_i$ of all hypotheses provides no information about a \emph{single} null p-value $P_i$.

The critical assumption is~\ref{assumption:distrib_dep}(a). Without covariates, the assumption implies that $I_1,\dotsc,I_K$ is a partition of p-values into independent blocks. This is not an assumption typically encountered in the multiple testing literature, although it has appeared e.g., in \citet{heesen2015inequalities, guo2016adaptive}. It is fundamental to the cross-weighting approach, the core idea of which is to avoid any dependence between each individual null p-value $P_i$ and its data-driven weight $W_i$. Cross-weighting ensures that $W_i$ is determined based on $X_i$ and p-values from the other folds, but not $P_i$. This would no longer be true with dependence \emph{across} folds. This observation is analogous to a similar phenomenon in cross-validation. In Chapter 7.1 of the \emph{Elements of Statistical Learning}, \citet*{hastie2009elements} caution practitioners to split data into independent folds when evaluating a supervised learning method by cross-validation (CV): if the folds are not independent, the CV estimates of prediction error are not reliable.

From the application perspective, the assumption is practical: domain experts often have sufficient understanding of their data to find suitable partitions of the hypotheses into independent blocks. In the example from Figure~\ref{fig:H3K27ac_histograms}, further detailed in Section~\ref{sec:hQTLexample}, it is plausible to assume that the data for hypotheses located on different chromosomes are independent, or at least that any potential dependences are negligible. As another example, for covariates $X_i$ that correspond to spatial or temporal positions, hypotheses that are sufficiently far away from each other will be independent if the dependences are mediated by spatial or temporal proximity.

We note that all other existing methods for multiple testing with covariates that provide $\FDR$ control assume either full independence \citep{lei2016adapt, cai2016cars}, weak dependence \citep{li2019multiple} or the ability to consistently estimate the joint distribution of all hypotheses \citep{sun2009large}. Thus, Assumption~\ref{assumption:distrib_dep} is a practical starting point towards dealing with common patterns of dependence encountered in real data.

Next, we describe two multiple testing methods with data-driven weights that have provable type-I error guarantees under dependence.
%-----------------------------------------------------------------------------------
\subsection{$k$-FWER control with cross-weighting under dependence}
\label{subsec:kfwer_control}
%-----------------------------------------------------------------------------------
$k$-FWER control is achieved by applying cross-weighting in conjunction with the weighted $k$-Bonferroni procedure of Definition~\ref{defn:wbonf}. We are not aware of existing procedures with data-driven weights and finite-sample $k$-FWER control. Existing proposals provide asymptotic guarantees \citep{wang2018weighted}.

The proof is direct and without technical complications. We provide it here in the main text, since it  shows the key idea behind cross-weighting: each null p-value $P_i$ is independent of its weight $W_i$, and this protects against overfitting.

\begin{thm}
\label{thm:IHW-bonf}
Let $\allpairs$ satisfy Assumption~\ref{assumption:distrib_dep} (or Assumption~\ref{assumption:distrib}) with respect to the partition $I_1,\dotsc,I_K$. Furthermore assume that we construct data-driven weights $W_i$ that are honest w.r.t.\ $I_1,\dotsc, I_K$ (Specification~\ref{assumption:honest_weights}). Then the weighted $k$-Bonferroni procedure (Definition~\ref{defn:wbonf}) with p-values $P_i$ and weights $W_i$ controls the $k\text{-FWER}$ at the nominal level $\alpha$.
\end{thm}

\begin{proof}
We first show that $P_i$ is independent of $W_i$ ($P_i \perp W_i$) for any $i \in \Hnull$. Without loss of generality, $i \in \Hnull \cap I_{\ell}$. By honesty, $W_i$ is a function only of the p-values in the other folds, $(P_j)_{j \in I_{\ell}^c}$ and all covariates $\mathbf{X}= (X_j)_{j \in [m]}$. It thus suffices to argue that $P_i$ is independent of  $((P_j)_{j \in I_{\ell}^c}, \mathbf{X})$. This follows from Assumption~\ref{assumption:distrib_dep} (resp. Assumption~\ref{assumption:distrib}). We next bound the $k\text{-FWER}$.

\begin{equation*}
\begin{aligned}
k\text{-FWER}&= \PP[V \geq k]  \leq \frac{1}{k} \EE[V] = \frac{1}{k}\sum_{i \in \Hnull} \PP\left[P_i \leq  \frac{k\alpha W_i}{m} \right]\\  &= \frac{1}{k}\sum_{i \in \Hnull} \EE\left[ \PP\left[P_i \leq  \frac{k\alpha W_i}{m} \mid W_i\right]\right] 
\stackrel{(*)}{\leq} \frac{1}{k}\sum_{i \in \Hnull} \EE\left[ \frac{k\alpha W_i}{m}\right]  = \frac{\alpha}{m}\EE\left[\sum_{i \in \Hnull} W_i\right] \leq \alpha\,.
\end{aligned}
\end{equation*}
Note that in $(*)$, we used the fact that for $i \in \Hnull$ it holds that $P_i$ is super-uniform and $P_i$ is independent of $W_i$. In the last step we used that honesty ensures that $\sum_{i} W_i = m$.
\end{proof}

%-----------------------------------------------------------------------------------
\subsection{FDR control with cross-weighting under dependence}
%-----------------------------------------------------------------------------------
We recall the basic procedure for controlling $\FDR$ with (deterministic) weights under arbitrary dependence: 
\begin{defn}[Weighted Benjamini-Yekutieli (wBY) \citep{benjamini2001control, blanchard2008two}]
\label{defn:wby}
Consider p-values $P_1,\dotsc,P_m$ with arbitrary dependence such that the null p-values are super-uniform. Furthermore consider deterministic weights $w_i \geq 0$ such that $\sum_{i=1}^m w_i = m$. Then the $\FDR$ is controlled at level $\alpha \in (0,1)$ by applying the weighted Benjamini-Yekutieli procedure at level $\alpha$, i.e., the weighted Benjamini-Hochberg procedure (Definition~\ref{defn:tau-wbh}) with $\tau=1$ at level $\alpha/\sum_{k=1}^m \frac{1}{k}$.
\end{defn}

We now show that applying the weighted BY procedure with cross-weighting controls the $\FDR$ under Assumption~\ref{assumption:distrib_dep}.

\begin{thm}[IHW-BY controls the $\FDR$ under honesty and independent folds]
\label{thm:ihw-by}
Let $\allpairs$ satisfy Assumption~\ref{assumption:distrib_dep} with respect to the partition $I_1,\dotsc,I_K$. Furthermore assume that we construct data-driven weights $W_i$ that are honest w.r.t.\ $I_1,\dotsc, I_K$ (Specification~\ref{assumption:honest_weights}). Then the weighted BY procedure (Definition~\ref{defn:wby}) with p-values $P_i$ and weights $W_i$ controls the $\FDR$ at the nominal level $\alpha$.
\end{thm}

To demonstrate that honesty is essential for the result of Theorem~\ref{thm:ihw-by}, we next describe two plausible candidate methods for FDR control with covariates that do not control $\FDR$:
\begin{example}[BY with arbitrary data-driven weights does not control $\FDR$ under Assumption~\ref{assumption:distrib_dep}]
Theorem~\ref{thm:ihw-by} may appear as a consequence of Theorem 4.2.\ of \citet{blanchard2008two}, who extended the results of \citet{benjamini2001control} and proved that the weighted BY procedure (Definition~\ref{defn:wby}) controls the $\FDR$ for any choice of weights and any p-value distribution. However, their result holds only for deterministic weights and not for data-driven weights, as we now demonstrate.

\begin{proof}
We generate $\allpairs$ satisfying Assumption~\ref{assumption:distrib_dep} and under the global null as follows: fix $m=2m'$ for $m' \in \mathbb{N}$. We consider deterministic covariates $X_i=i$ and the partition $I_1 = \cb{1,\dotsc,m'}, I_2=\cb{m'+1,\dotsc,m}$.  We first draw a permutation $\mathcal{\sigma}$ from the uniform measure on the permutation group of $\left\{1,\dotsc,m'\right\}$. Next we independently draw: $U_i \sim U[(i-1)/m', i/m']$  for $i=1,\dotsc,m'$ and let $P_{i} = U_{\mathcal{\sigma}(i)}$. Finally we draw independent $P_{m'+1}, \dotsc, P_m \sim U[0,1]$. Weights are chosen as follows: Let $i^* \in \argmin_{i}\left\{P_i \right\}$ and then let $W_i = W(X_i) = m\ind\left(X_i = i^*\right)$. Then the $\FDR$ of weighted BY at $\alpha$ is equal to $1$ as long as $m/\sum_{k=1}^m \frac{1}{k} > 2/\alpha$, as we now show:

Since the smallest p-value in $I_1$ is uniformly distributed on $U[0,1/m']$, it follows that with probability $1$, $P_{i^*} \leq 1/m'$ and hence $P_{i^*}/ W_{i^*} \leq 2/m^2 < \alpha/(m \sum_{k=1}^m \frac{1}{k})$. $H_{i^*}$ gets rejected and so $\FDP = 1$ almost surely. 
\end{proof}
\noindent In contrast, $\FDR$ control would be guaranteed, had we used weights derived through cross-weighting. BY with $\tau$-censored weights (Specification~\ref{assumption:tau_stopped_weights}) also does not control $\FDR$, cf. \supplementname~\ref{subsec:counterexample_by_tau}.
\end{example}

\begin{example}[AdaPT with BY correction does not control $\FDR$ under Assumption~\ref{assumption:distrib_dep}] \citet{lei2016adapt} prove $\FDR$ control for AdaPT under full independence (cf.\ Assumption~\ref{assumption:distrib}). Here we demonstrate that even with the Benjamini-Yekutieli correction, i.e., at level $\alpha/\sum_{k=1}^m \frac{1}{k}$ and $\tau$-censoring (Specification~\ref{assumption:tau_stopped_weights}), AdaPT does not control $\FDR$ under Assumption~\ref{assumption:distrib_dep}.

\begin{proof}
We generate $\allpairs$ satisfying Assumption~\ref{assumption:distrib_dep} and under the global null as follows: we fix \smash{$m=2m', m' \in \mathbb{N}$} and consider the partition $I_1 = \cb{1,\dotsc,m'}, I_2=\cb{m'+1,\dotsc,m}$. We take constant covariates \smash{$X_i=1$} for all $i$ and draw \smash{$P_1,P_{m'+1} \simiid U[0,1]$}. Finally we set \smash{$P_2,\dotsc,P_{m'} = P_1$} and $P_{m'+2},\dotsc,P_m = P_{m'+1}$. We then run the AdaPT algorithm at level $\alpha/\sum_{k=1}^m \frac{1}{k}$ with the initialization specified in \citet{lei2016adapt}. Then $\FDR \geq 0.2925$ as long as $m/\sum_{k=1}^m \frac{1}{k} > 2/\alpha$, as we now show:

As specified in Section 4.4.1 of \citet{lei2016adapt}, the AdaPT algorithm is initialized at threshold $0.45$. Now call $A$ the event that $\cb{P_1 \leq 0.45, P_{m'+1} < 0.55}$. On the event $A$, on the first step of the algorithm, AdaPT estimates the FDP (cf. \eqref{eq:fdphat_bc}) as $\p{ 1+\sum_i \ind\p{ P_i \geq 1-0.45}}/\sum_i \ind\p{ P_i \leq 0.45}$, which is equal to $1/m'$ if $P_{m'+1} > 0.45$ and equal to $1/m$ otherwise. In both cases, the estimated FDP is less or equal than $1/m'$ and thus less than $\alpha/\sum_{k=1}^m \frac{1}{k}$ under our assumption on $m,\alpha$. Thus AdaPT immediately terminates, rejecting all p-values in $I_1$, and so $\FDP=1$. Similarly $\FDP=1$ on the event $A'=\cb{P_1 < 0.55, P_{m'+1} \leq 0.45}$ and $\FDR \geq \mathbb P[A \cup A']= 0.2925$. Finally, note that the above procedure is $\tau$-censored with $\tau=0.45$.
\end{proof}
\end{example}

%------------------------------------------------------------
\section{Learning powerful weighting rules}
\label{sec:powerful_weighting_rules}
%------------------------------------------------------------
Sections~\ref{sec:control_guarantees} and \ref{sec:IHW_FDR_dependence} focused on sufficient conditions for type-I error control, but did not address power. These conditions leave considerable flexibility in the choice of the class of possible weight functions, and in the method of selecting (or ``learning'') these functions, given the data. This flexibility gives the analyst the opportunity to use domain-specific as well as statistical knowledge to make choices that have desirable type-II error properties. Nevertheless, it is useful to provide a default algorithm that works well across a range of settings. To this end, here we describe two schemes for learning weight functions, one for weighted $k$-Bonferroni and one for weighted BH. Both rely on positing the approximate applicability of model~\eqref{eq:conditional_twogroups}, estimating quantities appearing therein and solving a convex program to find a weight function that optimizes the expected number of discoveries.

%-----------------------------------------------------
\subsection{Learning weights for IHW $k$-Bonferroni}
\label{subsec:ihw_bonf_weights}
%-----------------------------------------------------
The weighted $k$-Bonferroni procedure with weight function $W(\cdot)$ rejects hypotheses that satisfy $P_i \leq k\alpha/m W(X_i)$. Under Model~\eqref{eq:conditional_twogroups}, a weight function maximizing the expected number of discoveries is one that maximizes \smash{$\sum_i \PP[ P_i \leq k\alpha/m W(X_i) \mid X_i] = \sum_i F(k\alpha/m W(X_i) \; | \; X_i)$}.  To derive honest weights (Specification~\ref{assumption:honest_weights}) that approximately maximize this objective, we learn \smash{$\widehat{W}^{-\ell}$} for each fold $\ell$ separately as follows: first we estimate $F(t \mid x)$ from Model~\eqref{eq:conditional_twogroups} by \smash{$\widehat{F}^{-\ell}(t \mid x)$} using only p-values and covariates outside of fold $\ell$. Next,  identifying \smash{$\widehat{W}^{-\ell}(\cdot)$} with the function's values evaluated at the $X_i$, i.e. \smash{$W_i = \widehat{W}^{-\ell}(X_i),\; i \in I_{\ell}$} we solve the $\abs{I_{\ell}}$-dimensional problem with optimization variables $\mathbf{w} = (w_i)_{i \in I_{\ell}}$:

\begin{equation}
\label{eq:wbonf_opt}
(W_i)_{i \in I_{\ell}} \; \in \; \argmax_{\mathbf{w} \in [0,\infty)^{\abs{I_{\ell}}}} \left\{\sum_{i \in I_{\ell}}\widehat{F}^{-\ell}\left(k\alpha/m \cdot w_i \,\mid\, X_i \right) \;\; \middle| \;\; w_i \geq 0,\;\; \sum_{i \in I_{\ell}} w_i = \abs{I_{\ell}} \right\}.
\end{equation}
This setting allows for conditional distributions $\widehat{F}^{-\ell}(t \mid X_i)$ that are different for tests with different covariates $X_i$. We consider estimators \smash{$\widehat{F}^{-\ell}(t \mid x)$} that are concave in $t$ for all $x$. This has the advantage of turning~\eqref{eq:wbonf_opt} into a convex optimization program, which is often tractable. Concavity of the distribution of p-values is a reasonable assumption and often provides a good fit to multiple testing datasets \citep*{strimmer2008unified, genovese2006false}. However, the procedure works even when the concavity assumption does not hold: given any (potentially non-concave) pilot estimator of the conditional distribution function \smash{$t \mapsto F(t \mid x)$}, we can project it onto the set of concave distribution functions and solve the optimization problem with the projected distribution functions. We interpret the resulting procedure as a convex relaxation of~\eqref{eq:wbonf_opt} that makes computation tractable.

With this setup, we are ready to state a concrete weighting scheme, which proceeds in three steps: first, discretize the $X_i$ into a finite number of bins defined, e.g., by quantile slicing or as the leaves of a tree. Second, estimate \smash{$\widehat{F}^{-\ell}(t \mid \text{bin})$} by the Grenander estimator \citep{grenander1956theory}, i.e., the least concave majorant of the empirical cumulative distribution function of the p-values $P_i$ with $i \in I_{\ell}^c$ and $X_i \in \text{bin}$. Third, solve~\eqref{eq:wbonf_opt} for each $\ell$ by linear programming. The reason that~\eqref{eq:wbonf_opt} may be expressed as a linear program is that the Grenander estimator is always concave in $t$ and piecewise linear. We provide the details of the estimation and optimization procedures in \supplementname~\ref{subsec:ihw_grenander}; the computational complexity scales as $O(\log(m) \cdot m)$.

An alternative ansatz is to specify $\pi_0(x)$ and $F_{\text{alt}}(\cdot \mid X_i=x)$ in the conditional two-groups model~\eqref{eq:conditional_twogroups} parametrically. For instance, we may consider for $X_i \in\mathbb{R}^p$
\begin{equation}
\begin{aligned}
\label{eq:betamix_simulation_model}
 & \pi_0(x) = \expit(a_0 + a^\top x), \quad\quad\text{ where } \expit(u)=\exp(u)/(1+\exp(u))\\
 & F_{\text{alt}}(\cdot \mid X_i=x) = \text{Beta}(\beta(x), 1),\;\;\beta(x) =  b_0 + b^\top x.
\end{aligned}
\end{equation}
Such a Beta-Uniform mixture model has been considered in the setting without covariates, e.g., by \citet{allison2002mixture, klaus2011learning} and with covariates by \citet{lei2016adapt}. In \supplementname~\ref{subsec:ihw_betamix} we explain how to learn the parameters of the model using the expectation-maximization algorithm and how to optimize~\eqref{eq:wbonf_opt}.

%--------------------------------------------------------
\subsection{Learning weights for IHW Benjamini-Hochberg}
\label{subsec:bh_weights}
%--------------------------------------------------------
Our starting point for deriving powerful weight functions for the weighted BH procedure (Definition~\ref{defn:tau-wbh}) is again the conditional two-groups model~\eqref{eq:conditional_twogroups}. We seek a threshold function $\thresholdfunctionfunction:\mathcal{X} \to [0,1]$, such that the multiple testing procedure that rejects hypotheses with $P_i \leq \thresholdfunctionfunction(X_i)$ satisfies the following two properties: first, the marginal $\FDR$, defined as $\text{mFDR}(\thresholdfunctionfunction) := \PP[ H_i =0 \mid P_i \leq \thresholdfunctionfunction(X_i)]$ is bounded by $\alpha$, i.e., $\text{mFDR}(\thresholdfunctionfunction) \leq \alpha$ and second, the expected number of discoveries \smash{$\sum F(\thresholdfunction(X_i) \; | \; X_i)$} is large\footnote{Such a Bayesian, Neyman-Pearson type procedure is motivated by the asymptotic equivalence between the frequentist $\FDR$ and the $\text{mFDR}$ \citep{genovese2004stochastic, sun2007oracle, cai2009simultaneous, cai2016cars}.}. Similarly to our Bonferroni construction, we learn the threshold function \smash{$\widehat{\thresholdfunctionfunction}^{-\ell}$} for each fold $\ell$ separately. To this end, we estimate \smash{$\widehat{F}^{-\ell}(t \mid x)$} and $\widehat{\pi}_0^{-\ell}(x)$ out of fold. Noting that $\text{mFDR}(\thresholdfunction) \leq \alpha$ is implied by $\EEs{\pi_0(X_i) \thresholdfunction(X_i)} \leq \alpha \EEs{F(\thresholdfunction(X_i) \mid X_i)}$, we propose solving:

\begin{equation}
\label{eq:wbh_opt}
\begin{aligned}
 \mathbf{t}= (t_i)_{i \in I_{\ell}} \; \in \; \argmax_{\mathbf{t} \in [0,1]^{\abs{I_{\ell}}}} \left\{\sum_{i \in I_{\ell}}\widehat{F}^{-\ell}\left(t_i \mid X_i \right) \;\; \cond \;\; t_i \geq 0,\;\; \sum_{i \in I_{\ell}} \widehat{\pi}_0^{-\ell}(X_i)t_i \leq \alpha \sum_{i \in I_{\ell}}\widehat{F}^{-\ell}\left(t_i \mid X_i \right) \right\}.
\end{aligned}
\end{equation}
As our goal is to apply the weighted BH procedure, we convert these thresholds $t_i$ into weights $W_i$ through normalization: for $i \in I_{\ell}$, set $W_i = \abs{I_{\ell}}\cdot t_i/(\sum_{i \in I_{\ell}} t_i)$, unless the denominator is $0$, in which case $W_i=1$. A few remarks are in order: similarly to optimization problem~\eqref{eq:wbonf_opt}, \eqref{eq:wbh_opt} is also a convex program if \smash{$\widehat{F}^{-\ell}(t \mid x)$} is concave in $t$ for all $x$, and may be expressed as a linear program if the Grenander estimator is used. We thus again suggest to discretize $X_i$ and estimate distributions with the Grenander estimator. If the weights will be applied in conjunction with the weighted BH algorithm, we suggest to simply set \smash{$\widehat{\pi_0}^{-\ell} \equiv 1$}. This optimization and estimation scheme was proposed by \citet{ignatiadis2016data}. Alternatively, \smash{$\widehat{\pi_0}^{-\ell}(x)$} may be estimated by applying Storey's null proportion estimator \citep{storey2004strong}  to all hypotheses outside fold $I_{\ell}$ that fall into the same bin as $x$. Details of the estimation and optimization procedures are provided in \supplementname~\ref{subsec:ihw_grenander}.

The weights $W_i$ constructed above are honest (Specification~\ref{assumption:honest_weights}). Yet, in view of Theorems~\ref{thm:IHWc} and~\ref{thm:IHW-Storey}, it might appear unsatisfying that $W_i$ do not satisfy the $\tau$-censored weights condition (Specification~\ref{assumption:tau_stopped_weights}). In our experience, the proposed procedure with the Grenander estimator does not overfit and controls the $\FDR$. This is corroborated by extensive simulations below and by the asymptotic guarantees of Proposition~\ref{prop:asymp}.

Our alternative proposal, which satisfies $\tau$-censoring (Specification~\ref{assumption:tau_stopped_weights}), is to fit the Beta-Uniform mixture model~\eqref{eq:betamix_simulation_model}. The EM algorithm may be modified to accommodate for censored knowledge of $P_i \leq \tau$; cf.\ \citet{markitsis2010censored} in the setting without covariates. Furthermore, under model~\eqref{eq:betamix_simulation_model}, the solution to problem~\eqref{eq:wbh_opt} lies on a contour of equal conditional local fdr (cf.\ Theorem 2 in \citet{lei2016adapt}), and this fact facilitates the optimization. We describe the steps in more detail in \supplementname~\ref{subsec:ihw_betamix}.

Finally, we use the same framework to derive weights for the weighted Benjamini-Yekutieli procedure (Definition~\ref{defn:wby}): we proceed as for weighted BH but solve~\eqref{eq:wbh_opt} with $\alpha$ replaced by $\alpha/\sum_{k=1}^m \frac{1}{k}$. In this case, honesty suffices for $\FDR$ control (Theorem~\ref{thm:ihw-by}).

%------------------------------------------------------------
\section{Numerical experiments}
\label{sec:numerical_study}
%------------------------------------------------------------
Our goal in this section is to corroborate through simulations of three important settings---grouped multiple testing, multiple testing with continuous covariates and simultaneous two-sample tests---the following claims: first, some methods with asymptotic FDR control guarantees do not control FDR in finite samples. Second, IHW is a flexible framework for multiple testing, its main advantage over other methods being finite sample error control (due to cross-weighting), while remaining competitive in terms of power. Throughout this section, we define power as

\begin{equation}
\label{eq:power_defn}
\text{Power} := \EEs{ \frac{ \sum_{i \notin \Hnull} \ind\p{i\text{  rejected}}}{\max\cb{1,m-\abs{\Hnull}}}}
\end{equation}
The expectation, just as the $\FDR$, is evaluated through averaging over Monte Carlo replicates.

%---------------------------------------------------------------
\subsection{Grouped multiple testing}
\label{subsec:group_mtp}
% ---------------------------------------------------------------
We first consider the multiple testing problem with groups, i.e, with categorical covariates $X_i \in [G]$. In each simulation we generate $(P_i, H_i, X_i),\;i=1,\dotsc,20000$, independently, as follows
\begin{align}
\label{eq:grouped_sim}
&\tilde{X}_i = \floor{40 \cdot (i-1)/m},\;\;X_i = \ceil{\tilde{X}_i/40 \cdot G}    \notag    \\
&H_i \mid \tilde{X}_i \sim \text{Bernoulli}(1-\pi_0(\tilde{X}_i)),\;\; \pi_0(\tilde{X}_i) = (0.2 + 0.8\tilde{X}_i/36)\cdot \ind(\tilde{X}_i=0 \text{ mod }4) \;+\;\ind(\tilde{X}_i \neq 0\text{ mod }4) \notag  \\
&Z_i \mid H_i, \tilde{X}_i \sim \mathcal{N}(H_i \cdot \mu(\tilde{X}_i),1),\;\; \mu(\tilde{X}_i) = 2.5 - 2\tilde{X}_i/36 \notag \\
&P_i = 1-\Phi(Z_i),\;\; \Phi \text{ is the standard Normal CDF}
\end{align}
In words, there are 40 latent groups defined by $\tilde{X}_i$, each with 500 hypotheses. A quarter of the groups has non-nulls, three quarters do not. The alternative signal strength $\mu(\cdot)$ and null proportion $\pi(\cdot)$ vary linearly across non-null groups. Parameters are chosen so that the overall proportion of nulls is $0.9$. We then coarsen $\tilde{X}_i$ to $X_i = \ceil{\tilde{X}_i/40 \cdot G }$, with $G$ varying across simulations; $X_i$ is non-latent, i.e., visible to the algorithm. For example, for $G=2$, $X_i$ takes on only two levels (2 groups), while for $G=40$, $X_i = \tilde{X}_i$ takes on all 40 levels. We also use the above configuration of covariates and simulate under the global null by drawing all p-values from the uniform distribution.

We compare the following seven methods:
\begin{enumerate}[nosep]
    \item The \textbf{Benjamini-Hochberg (BH)} method \citep{benjamini1995controlling}, which ignores the covariates $X_i$.
    \item The \textbf{stratified BH procedure  (SBH)} \citep{sun2006stratified, efron2008simultaneous}, wherein the BH procedure is applied $G$ times separately to p-values corresponding to different levels of $X_i$.
    \item The \textbf{Clfdr (conditional local fdr)} procedure of  \citet{cai2009simultaneous}, which applies an optimal decision rule that rejects hypotheses with a low value of the group-wise local fdr (cf.\ Algorithm~\ref{alg:Cfdr} in \supplementname~\ref{sec:fdr}). We apply a data-driven version of the oracle rule by estimating local fdrs within each group with the \texttt{fdrtool} CRAN Package \citep{strimmer2008fdrtool}, which estimates marginal densities with the Grenander estimator.
    \item The \textbf{Group Benjamini-Hochberg (GBH)} procedure of \citet{hu2010false} with null-proportion adaptivity, as described in Algorithm~\ref{alg:GBH} ($\tau=0.5$).
    \item The \textbf{IHW-GBH} procedure with null-proportion adaptivity, as described in Algorithm~\ref{alg:IHW-GBH} ($\tau=0.5$) with hypotheses randomly split into 5 folds.
    \item The \textbf{IHW-Storey-Grenander} procedure: the IHW-Storey method (Theorem~\ref{thm:IHW-Storey}) with hypotheses randomly split into 5 folds and data-driven weights based on the Grenander estimator described in Section~\ref{subsec:bh_weights} and \supplementname~\ref{subsec:ihw_grenander}.
    \item The \textbf{Structure Adaptive Benjamini-Hochberg algorithm (SABHA)} by \citet{li2019multiple}: SABHA first estimates $\widehat{\pi}_0(\cdot)$ for each group by solving a joint convex optimization problem. Then, the $\tau$-censored, weighted BH procedure is applied with weights $W_i = 1/\widehat{\pi}_0(X_i)$. We set the tuning parameters of group-wise SABHA to $\tau=0.5,\; \varepsilon = 0.1$ following Section 7.1 of~\citet{li2019multiple}.
\end{enumerate}
All of the above methods provably control $\FDR$ asymptotically, as $m \to \infty$, the number of groups remains fixed and there is signal in the data, but only BH and IHW-GBH have provable finite-sample $\FDR$ control at $\alpha$ and SABHA at $\alpha(1 + 10\sqrt{G/m})$~\citep[Lemma 2]{li2019multiple}.

\begin{figure}
\centering
\includegraphics[width=\textwidth]{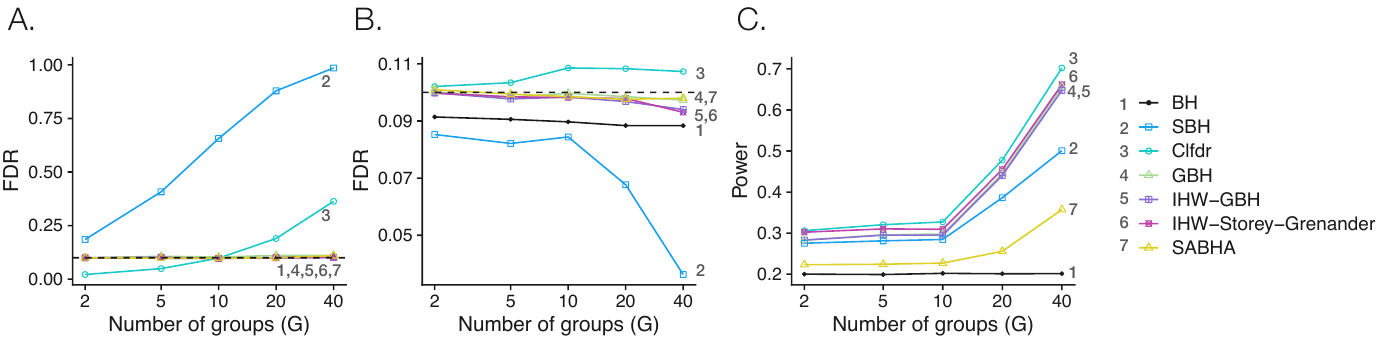}
\caption{\normalfont\textbf{Grouped multiple testing simulation: A. False discovery rate under the global null} in Model~\eqref{eq:grouped_sim} (averaged over 10000 Monte Carlo replicates) for seven methods for multiple testing with groups. \textbf{B.,C. False discovery rate and power} in Model~\eqref{eq:grouped_sim} (averaged over 200 Monte Carlo replicates) when there is signal (average null proportion is $0.9)$ for the same seven methods. The nominal $\alpha$ is equal to 0.1 throughout.\label{fig:grouped_sim}}
\end{figure}
Results are shown in Figure~\ref{fig:grouped_sim}. Under the global null (Fig.~\ref{fig:grouped_sim}A), SBH strongly overfits, since under the global null the $\FDR$ is equivalent to the $\text{FWER}$, so it would need to pay a Bonferroni correction to apply BH separately to each group. Clfdr has $\FDR$ much below nominal for a small number of groups (the oracle local fdr procedure would not reject anything under the global null), but as the number of groups increases, it no longer controls $\FDR$. We further discuss this below. All other methods control $\FDR$ in this setting. For GBH, however, recall Fig.~\ref{fig:null_gbh} for a situation where it does display a pronounced loss of $\FDR$ control. 

For the simulations with signal (Fig.~\ref{fig:grouped_sim}B, C) we make the following observations: as $G$ increases, the covariates become more informative, hence in principle power can be increased. Indeed this is precisely what we observe (Fig.~\ref{fig:grouped_sim}C) for the grouped methods that do not directly estimate the distribution in each group (all methods except Clfdr and IHW-Storey-Grenander). The power of BH remains constant. After BH, the least powerful procedure appears to be SABHA; the suboptimality of its weighting scheme has been previously pointed out \citep{lei2016adapt}. We also observe that IHW-GBH matches the power of GBH and has the added advantage of provable finite-sample $\FDR$ control. Regarding the methods that estimate the distribution, when $G$ is small relative to $m$, then the Grenander estimator can precisely estimate the distribution in each bin. This translates into the Clfdr procedure and IHW-Storey-Grenander outperforming the other methods at small $G$; indeed Clfdr is provably asymptotically the most powerful procedure in this setting. However, as $G$ increases and the amount of data in each group decreases, the distributions are not estimated as accurately. The consequence for Clfdr is loss of $\FDR$ control, while IHW-Storey-Grenander retains $\FDR$ control due to cross-weighting. In conclusion, in this set of simulations, IHW is the most powerful method of those that control $\FDR$.

%----------------------------------------------------------
\subsection{Multiple testing with continuous covariates}
\label{subsec:adapt_simulations}
%----------------------------------------------------------
In this section we explore a setting with a two-dimensional, continuous covariate $X_i$. We seek to compare IHW, AdaPT and local fdr based methods with an emphasis on understanding behavior under model-misspecification (to be made precise momentarily). We simulate independent $(X_i, H_i, P_i), i=1,\dotsc,10000$ from the conditional two-groups model~\eqref{eq:conditional_twogroups} with the following choices for $\PP^X, \pi_0(x)$ and $F_{\text{alt}}(\cdot \mid X_i=x)$:
\begin{equation}
\begin{aligned}
\label{eq:adapt_simulation_model}
& \PP^X = U[0,1]^2,\;\; \pi_0(x) = 0.98\cdot \ind\p{x_1^2 + x_2^2 \leq 1} +  0.6 \cdot \ind\p{x_1^2 + x_2^2 > 1},\;\; (\EE[\pi_0(X_i)] \approx 0.9)\\
& F_{\text{alt}}(\cdot \mid X_i=x) = \text{Beta}(\beta(x), 1),\;\;\beta(x) =  1\big/\max\cb{1.3, \bar{\beta}\cdot(\sqrt{x_1}+\sqrt{x_2})} 
\end{aligned}
\end{equation}
$\bar{\beta} \in [1,3]$ is a parameter that varies across simulation settings. The two-dimensional covariates $X_i$ modulate both the null proportion $\pi_0(X_i)$ and the signal in the alternative density. We compare six methods. 
\begin{enumerate}[nosep]
    \item The \textbf{Benjamini-Hochberg (BH)} \citep{benjamini1995controlling} method ignoring $X_i$.
    \item The \textbf{oracle Clfdr procedure (Clfdr-oracle)} that rejects hypotheses with a small conditional local $\fdr$, \smash{$\fdr(P_i | X_i) := \PP[H_i = 0 | X_i, P_i]$} with a threshold chosen through Algorithm~\ref{alg:Cfdr} in \supplementname~\ref{sec:fdr}. This procedure achieves an optimal trade-off between the false nondiscovery rate and the false discovery rate, cf.\ \citet{sun2007oracle, cai2009simultaneous}. Clfdr-oracle, however, would not be available to an analyst, as it assumes oracle knowledge of the components~\eqref{eq:adapt_simulation_model} in model~\eqref{eq:conditional_twogroups}.
    \item The \textbf{IHW-BH-Grenander} procedure, similarly to the previous section, but without null-proportion adaptivity (i.e., with IHW-BH instead of IHW-Storey). The covariates $X_i \in [0,1]^2$ are binned into $5 \times 5$ equal volume bins.
\end{enumerate}
Furthermore, we compare three methods that fit Model~\eqref{eq:betamix_simulation_model} as a misspecified working model for the true model~\eqref{eq:adapt_simulation_model} using the EM algorithm (details in \supplementname~\ref{subsec:ihw_betamix}).
\begin{enumerate}[nosep]
    \setcounter{enumi}{3}
    \item \textbf{Clfdr-EM:} this is the same as Clfdr-oracle, but instead of true quantities we use the ones estimated by maximum likelihood on the misspecified model~\eqref{eq:betamix_simulation_model}. We employ the EM algorithm since the status $H_i \in \cb{0,1}$ is unknown.
    \item \textbf{IHW-Storey-BetaMix:} this is the IHW-Storey method with hypotheses split randomly into 5 folds and weights derived from optimization problem~\eqref{eq:wbh_opt} based on the (out-of-fold) estimated working model~\eqref{eq:betamix_simulation_model}. Here the EM algorithm deals with both unknown $H_i$ and unknown value of censored p-values $P_i \leq \tau$ with $\tau=0.1$.
    \item \textbf{AdaPT}, as implemented in the \texttt{adaptMT} CRAN package, wherein in each iteration the working model~\eqref{eq:betamix_simulation_model} is fitted. The EM algorithm deals with unknown $H_i$ and for a subset of hypotheses (``masked hypotheses'') the algorithm only has access to $\min\cb{P_i, 1-P_i}$ instead of $P_i$.
\end{enumerate}

\begin{figure}
\centering
\includegraphics[width=0.85\textwidth]{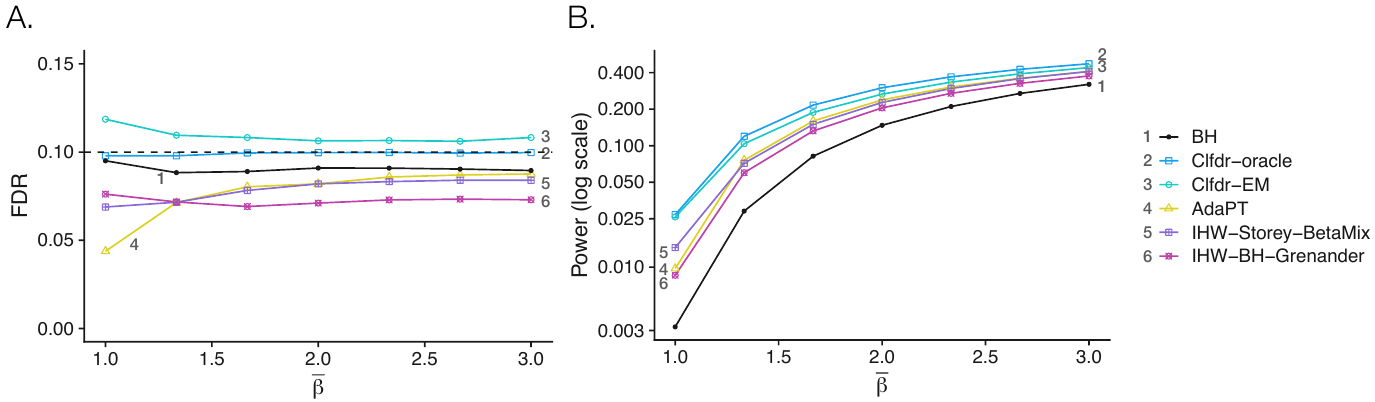}
\caption{\normalfont\textbf{Simulation for multiple testing with a continuous covariate: A. False discovery rate} in model~\eqref{eq:adapt_simulation_model} for six methods. The $x$-axis corresponds to a simulation parameter that is monotonically related to the strength of the signal for the alternatives.  \textbf{B. Power} in model~\eqref{eq:adapt_simulation_model} for the same six methods. The nominal $\alpha$ is equal to 0.1 throughout and results are averaged over 400 Monte Carlo replicates.\label{fig:betamix_sim}}
\end{figure}

The results are shown in Fig.~\ref{fig:betamix_sim}. As expected from theory, Clfdr-oracle controls the $\FDR$ and is most powerful. Clfdr-EM is also powerful, however because of misspecification in model~\eqref{eq:betamix_simulation_model}, it does not control the $\FDR$. All other algorithms control the $\FDR$. Among these, AdaPT is most powerful, closely followed by IHW-Storey-BetaMix and then by IHW-BH-Grenander; all of these procedures improve substantially upon BH. 

\paragraph{Breaking AdaPT:} Fig.~\ref{fig:betamix_sim} demonstrates that AdaPT is very powerful for multiple testing in model~\eqref{eq:adapt_simulation_model}. However, under two conditions (more of which, below), AdaPT's power (but not $\FDR$ control guarantees) can be diminished, even under independence. To explain these two conditions, we first provide a summary of how AdaPT works. In iteration $j$ of AdaPT, a candidate rejection function $\thresholdfunction_j: \mathcal{X} \to [0,1]$ is maintained and hypotheses that satisfy $P_i \leq \thresholdfunction_j(X_i)$ are in the provisional rejection set. The false discovery proportion at step $j$ is estimated by the \citet*{barber2015controlling} estimator (cf.\ \citet{arias2017distribution}):
\begin{equation}
\label{eq:fdphat_bc}
\widehat{\FDP}_{j} = \frac{ 1 + \abs{\{ i: P_i \geq 1-\thresholdfunction_j(X_i) \}}}{\abs{ \{ i: P_i \leq \thresholdfunction_j(X_i) \}}}.
\end{equation}
If $\widehat{\FDP}_{j} \leq \alpha$, the algorithm terminates and returns the current rejection set. Otherwise the rejection region $\thresholdfunction_j$ is further shrunk to $\thresholdfunction_{j+1}$ with $\thresholdfunction_{j+1}(x) \leq \thresholdfunction_j(x)$ for all $x$. The iteration continues until either the stopping criterion is satisfied or the empty set is returned.

The first complication of~\eqref{eq:fdphat_bc} is that AdaPT must reject at least $1/\alpha$ hypotheses or none at all. For example, for $\alpha = 0.05$, if there are 19 very small p-values, AdaPT may not be able to reject them, even if BH could. Hence AdaPT has low power in situations with very sparse signals, where the best one could hope for is to detect a handful of hypotheses. This is apparent in Figure~\ref{fig:betamix_sim}, in the lowest signal situation $(\bar{\beta}=1.0)$. There, AdaPT has $\FDR$ substantially below the nominal $\alpha$ and furthermore has lower power than IHW-Storey-BetaMix.

\begin{figure}
\centering
\includegraphics[width=\textwidth]{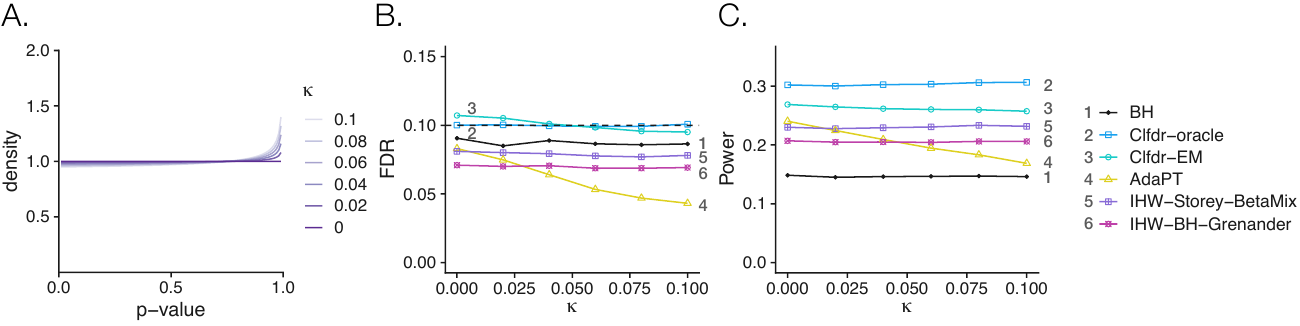}
\caption{\normalfont\textbf{Simulation for multiple testing with a strictly super-uniform null distribution: A. Density of null p-values} drawn from $(1-\kappa)U[0,1] + \kappa\text{Beta}(1, 0.5)$ for varying $\kappa$. \textbf{B, C. FDR control and power} under same simulation setting as Fig.~\ref{fig:betamix_sim}, but with $\bar{\beta} =2$ fixed and $\kappa$ varying (Fig.~\ref{fig:betamix_sim} corresponds to $\kappa=0$).}
\label{fig:betamix_onesided_sim}
\end{figure}

The second complication is that AdaPT can be conservative when the null p-value distribution is strictly super-uniform instead of uniform, because the numerator in~\eqref{eq:fdphat_bc} will overestimate the false discoveries. In applications, a strictly super-uniform distribution is typically caused by discrete p-values or when the researcher is testing for a one-sided alternative using a test calibrated to effect size zero, but many nulls have an effect in the opposite direction. To explore such enrichment of large p-values, we repeat the previous simulation with $P_i \mid (H_i = 0) \;\sim\; (1-\kappa)\,U[0,1] + \kappa\,\text{Beta}(1, 0.5)$, varying $\kappa$ $\in [0,0.1]$ and fixed $\bar{\beta}=2$. Our previous simulations correspond to $\kappa=0$, which yields the uniform null distribution. Fig.~\ref{fig:betamix_onesided_sim}A shows the null density as $\kappa$ varies, and panels B,C show the results of the simulation. We see that as $\kappa$ increases, the $\FDR$ of AdaPT quickly drops below the nominal $\alpha$ and as a consequence, power deteriorates.

%-------------------------------------------------------------------------------
\subsection{Simultaneous two-sample testing}
\label{subsec:cars_sims}
%-------------------------------------------------------------------------------
In this section we provide an example of a covariate $X_i$ that is random and arises from statistical  (rather than domain-specific) considerations. We study simultaneous two-sample testing for equality of means following \citet{cai2016cars}. For the $i$-th hypothesis we observe
\newcommand{\grA}{Y}
\newcommand{\grB}{V}
\begin{equation}
\label{eq:two_sample}
\grA_{i,1}, \dotsc, \grA_{i,n} \sim \mathcal{N}(\mu_{\grA,i}, \sigma_i^2) \quad \text{and}\quad 
\grB_{i,1}, \dotsc, \grB_{i,n} \sim \mathcal{N}(\mu_{\grB,i}, \sigma_i^2)
\end{equation}
(everything jointly independent). We are interested in testing $H_i: \mu_{\grA,i} = \mu_{\grB,i}$, $i=1,\dotsc,m$ and assume the variances $\sigma_i^2$ are known\footnote{The results extend to unequal sample sizes and to unknown variance. We refer the reader to \supplementname~\ref{subsubsec:ttest} and \citet{bourgon2010independent,liu2014incorporation,cai2016cars}.}. The optimal test statistic (in single hypothesis testing \citep{lehmann2005testing}) for this situation is the two-sample $z$-statistic $Z_i :=  \sqrt{n/2}\p{\overline{\grA_i}-\overline{\grB_i}}/\sigma_i$, where $\overline{\grA_i}$ and $\overline{\grB_i}$ are the sample means in each group. The p-values can be calculated as $P_i=2\p{1-\Phi(\abs{Z_i})}$, where $\Phi$ is the Standard Normal CDF. A basic multiple testing approach consists of applying BH to the p-values $P_i$.

In addition, denote by $ \hat{\mu}_i \coloneqq \left(\overline{\grA_i} + \overline{\grB_i}\right)/2$ the pooled average and let $X_i \coloneqq \sqrt{2n}\hat{\mu}_i/\sigma_i$. A direct covariance calculation reveals that $\text{Cov}(X_i, Z_i)=0$, and so $X_i$ and $Z_i$ are independent (note the joint normality). Hence we may apply the IHW framework with p-values $P_i$ and covariates $X_i$.

In single hypothesis testing, there is nothing to be gained from $X_i$ and its usefulness only emerges in the multiple testing setup. $X_i$ is a test statistic for the null hypothesis $\mu_{\grA,i} = \mu_{\grB,i} = 0$. If we believe a-priori that for many of the hypotheses $i$ with $\mu_{\grA,i} = \mu_{\grB,i}$, a sparsity condition holds, so that in fact $\mu_{\grA,i} = \mu_{\grB,i} = 0$, then large absolute values of this statistic are more likely to correspond to alternatives. Note that we did not actually re-specify our null hypothesis from $\mu_{\grA,i} = \mu_{\grB,i}$ to $\mu_{\grA,i} = \mu_{\grB,i} = 0$. We just assumed properties of the null hypotheses to motivate a choice of covariate, and are still testing for $\mu_{\grA,i} = \mu_{\grB,i}$.

In the simulation, which is similar to simulations in~\citet{cai2016cars}, we generate data from model~\eqref{eq:two_sample} with $m=10000$, $n=50$, $\sigma_i=1$ for all $i$. Furthermore, we vary $m_1$, the number of alternatives and let
$$ \mu_{\grA,i}= \left\{\begin{matrix}
 0.5,&  i=1,\dotsc,m_1\\ 
 0.25, & i=m_1+1,\dotsc, 2m_1\\ 
 0, & \text{otherwise}
\end{matrix}\right.,\;\;\;\;\mu_{\grB,i}= \left\{\begin{matrix}
 0,&  i=1,\dotsc,m_1\\ 
 0.25, & i=m_1+1,\dotsc, 2m_1\\ 
 0, & \text{otherwise}
\end{matrix}\right.$$
That is, only the first $m_1$ hypotheses are alternatives. The next $m-m_1$ hypotheses are nulls with the last $m-2m_1$ also being nulls with respect to the screening null $\mu_{\grA,i} = \mu_{\grB,i}=0$. We compare five methods.

\begin{enumerate}[nosep]
    \item The \textbf{Benjamini-Hochberg (BH)} procedure applied to $P_i$ and ignoring $X_i$.
    \item The \textbf{CARS} procedure (covariate-assisted ranking and screening) \citep{cai2016cars}: CARS is a multiple testing procedure designed specifically for simultaneous two-sample tests based on $Z_i$ and $X_i$. At a high level, CARS learns a function $(z,x) \mapsto$ $\hat{s}_{\text{CARS}}(z,x)$ and a threshold $\hat{t}_{\text{CARS}}$ and rejects all hypotheses such that $\hat{s}_{\text{CARS}}(Z_i,X_i)$ $\leq \hat{t}_{\text{CARS}}$. Asymptotically, CARS controls the $\FDR$ and learns the optimal decision boundary. We use the default settings of the \texttt{CARS} function (\texttt{option="regular"}) in the \texttt{CARS} R package.
    \item \textbf{CARS-sparse}: a modification of CARS, also proposed by \citet{cai2016cars}, that is more conservative and empirically alleviates loss of $\FDR$ control in situations with sparse signals (\texttt{option="sparse"} in the \texttt{CARS} package).
    \item \textbf{IHW-Storey-CARS:} we use IHW-Storey (Theorem~\ref{thm:IHW-Storey}) in conjunction with a honest (but not $\tau$-censored) weighting heuristic based on CARS. We partition hypotheses randomly into 5 folds $I_1,\dotsc,I_5$. To choose weights for $I_{\ell}$ we proceed as follows: first, we run CARS on the remaining 4 folds and get \smash{$\hat{s}^{-\ell}_{\text{CARS}}(\cdot,\cdot)$} and \smash{$\hat{t}_{\text{CARS}}^{-\ell}$}. Then, for $i \in I_{\ell}$, we let $t_i$ be the smallest threshold at which $H_i$ would get rejected,
    $$t_i := \inf\cb{z \geq 0 : \; \hat{s}^{-\ell}_{\text{CARS}}(z, X_i) \leq \hat{t}_{\text{CARS}}^{-\ell}}.$$
    Then we let $\tilde{W}_i = 2\p{1-\Phi(t_i)}$, $W_i = \abs{I_{\ell}}\tilde{W}_i/\sum_{j \in I_{\ell}} \tilde{W}_j$ and finally apply the IHW-Storey procedure from Theorem~\ref{thm:IHW-Storey}.     
    \item \textbf{IHW-Storey-Grenander}, as in the grouped multiple testing simulations of Section~\ref{subsec:group_mtp}; we discretize the covariate $X_i$ into 10 groups with 1000 observations each. 
\end{enumerate}

\begin{figure}
\centering
\includegraphics[width=\textwidth]{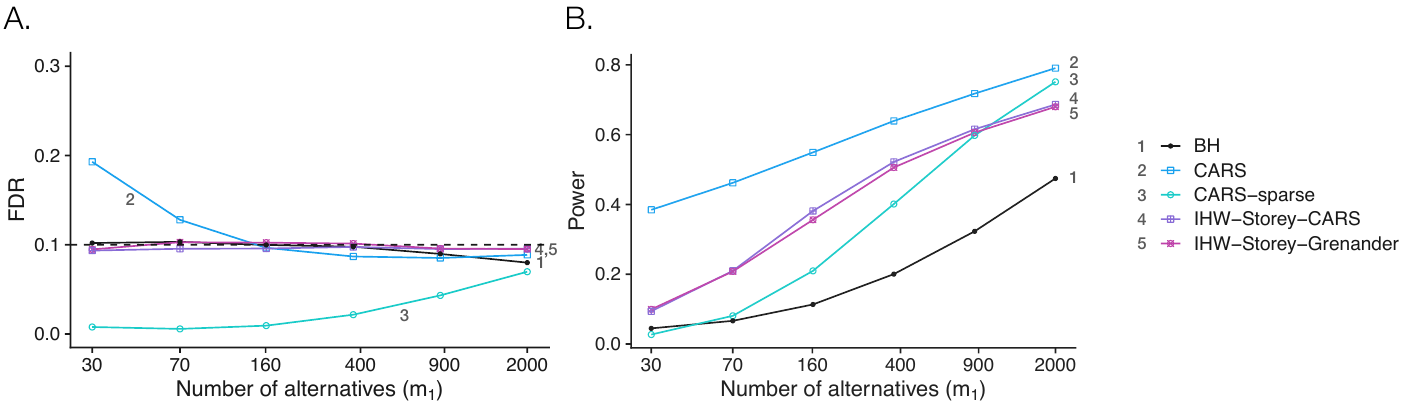}
\caption{\normalfont\textbf{Simulation for simultaneous two-sample testing: A. False discovery rate} and \textbf{B. Power} in model~\eqref{eq:two_sample} for five methods for simultaneous two-sample testing. The nominal $\alpha$ is equal to 0.1 throughout, and results were averaged over 400 Monte Carlo replicates.\label{fig:cars_sim}}
\end{figure}
The results are shown in Fig.~\ref{fig:cars_sim}. With sparse signal (small $m_1$), CARS fails to control the $\FDR$. This observation had also been made by \citet{cai2016cars}, who therefore proposed a modification, CARS-sparse, which indeed controls $\FDR$ in our simulation, as do all other methods. On the other hand, IHW-Storey-CARS is easy to implement---using existing software for CARS---and turns out to have more power in the simulations than CARS-sparse. IHW-Storey-Grenander also has more power than CARS-sparse.

%----------------------------------------------------------------------------
\section{Application example: biological high-throughput data}
\label{sec:hQTLexample}
%----------------------------------------------------------------------------
\citet{grubert2015genetic} assayed cell lines derived from 75 human individuals for the status of their single nucleotide polymorphisms (SNPs, i.e., differences that exist between the genome sequences of individuals) and a biochemical modification of DNA-associated molecules called H3K27ac. We tested all within-chromosome associations by marginal regression of the quantitative readout from the ChiP-seq assay for H3K27ac on the polymorphisms, which are encoded as categorical variables with levels \emph{aa}, \emph{ab}, \emph{bb}, using the software \texttt{Matrix eQTL} \citep{shabalin2012matrix}. Here we restrict ourselves to associations in Chromosomes 1 and 2, for which Grubert et al.\ reported the status of $N_1= 645452$ and $N_2=699343$ SNPs and the H3K27ac levels at $K_1 = 12193$ and $K_2=11232$ genomic positions (``peaks'') on these chromosomes. This results in a total of approximately 16 billion hypotheses ($m = N_1\times K_1+N_2\times K_2 \approx 1.6\cdot 10^{10}$)\footnote{We note that computing and storing 16 billion p-values puts notable demands on computing infrastructure. Therefore, a common choice made by implementations such as Matrix eQTL \citep{shabalin2012matrix} to reduce storage requirements is to only report p-values below
some threshold (e.g., in this case, below $10^{-4}$). Benjamini-Hochberg/Yekutieli and IHW-BH/BY can deal with this seamlessly by operating as if the right-censored p-values were equal to $1$. In contrast, AdaPT depends on the large p-values to estimate the $\FDR$, cf.\ \eqref{eq:fdphat_bc}.}. Figure~\ref{fig:H3K27ac_histograms} shows the marginal histogram of the p-values and illustrates how these p-values are related to the genomic distance between SNP and H3K27ac peak. This covariate is motivated from biological domain knowledge: associations across shorter distances are a-priori more plausible and empirically more frequent.

We compare two different approaches of dealing with the multiplicity, while controlling the FDR:

\begin{enumerate}[nosep]
\item The \textbf{Benjamini-Yekutieli (BY)} procedure on the $m$ p-values (at level $\alpha=0.01$): such a conservative procedure is justified, since p-values for the same H3K27Ac peak and different, but genetically linked SNPs will be strongly dependent. 
\item The \textbf{IHW-BY-Grenander} method (at level $\alpha=0.01$) using as covariate the genomic distance between SNP and H3K27ac peak and weights based on the Grenander estimator after binning based on genomic distance; cf.\ Section~\ref{subsec:bh_weights} and \supplementname~\ref{subsec:ihw_grenander} for a description of the algorithm and \supplementname~\ref{sec:hqtl_suppl} for application-specific details. To satisfy Assumption~\ref{assumption:distrib_dep} and hence have guaranteed $\FDR$ control by Theorem~\ref{thm:ihw-by}, we partition p-values into two folds corresponding to the different chromosomes. The data for these are, to sufficient approximation, independent.
\end{enumerate}

\begin{figure}
\centering
\includegraphics[width=\textwidth]{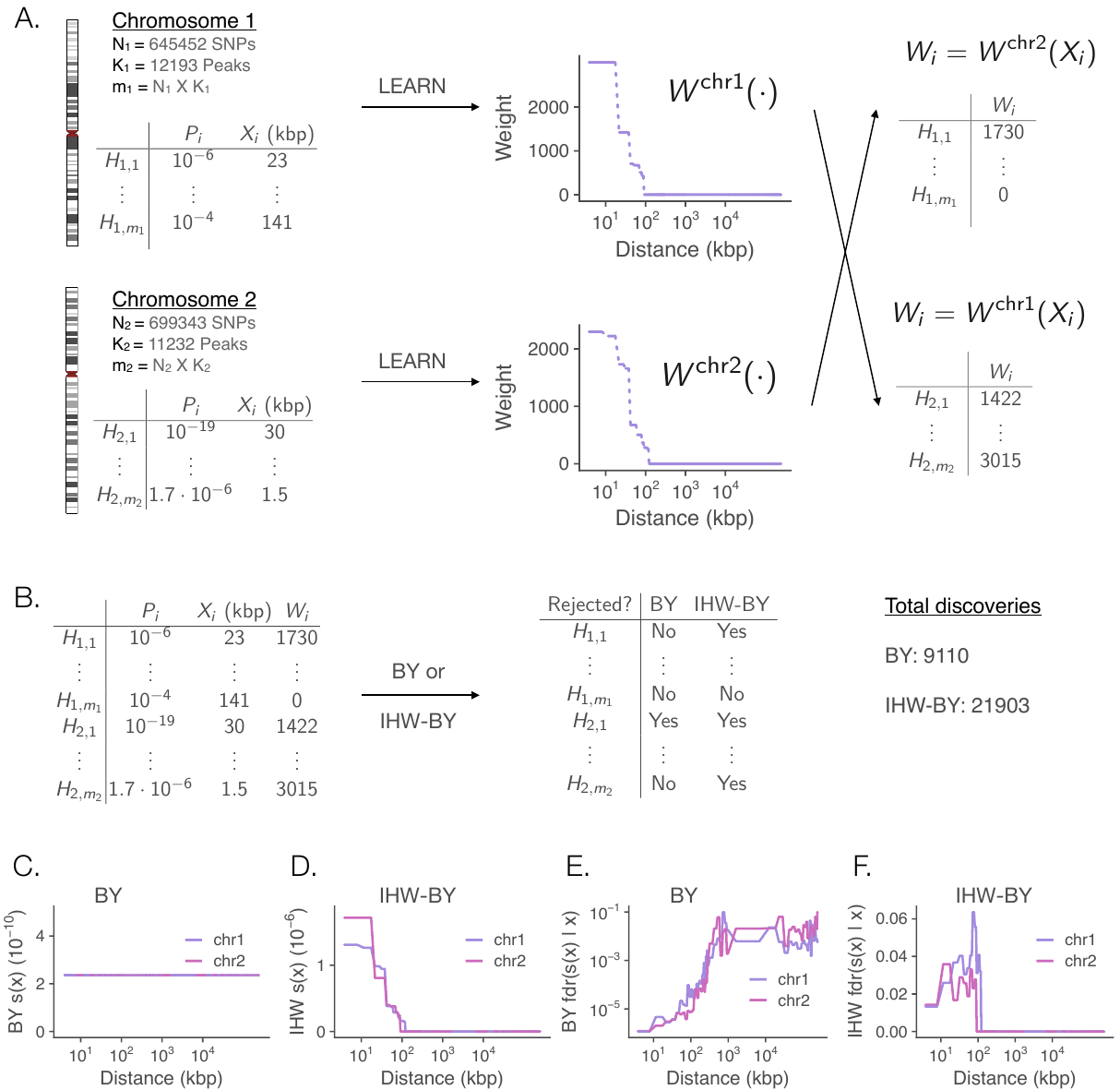}
\caption{\normalfont\textbf{Biological data example revisited: A. Schematic representation of cross-weighting:} we consider a multiple testing situation with $m=m_1+m_2$ hypotheses that can be partitioned into two independent folds (here: two chromosomes). Besides the p-value $P_i$, a covariate $X_i$ is available for each hypothesis $(i=1,\ldots m)$, which here is the genomic distance between SNP and peak. For each fold we learn the optimal weight function and assign weights to hypotheses from fold 1 using the function learned from the $((P_i,X_i))_i$ of fold 2, and vice versa. \textbf{B. Data-driven weighting increases power:} Upon merging the two tables of hypotheses, we apply the Benjamini-Yekutieli (BY) method at $\alpha=0.01$ to the p-values, or the weighted BY method with the learned weights (IHW-BY). Each method returns a list of rejected hypotheses. IHW more than doubles the total number of discoveries.  \textbf{C, D. Decision boundaries for BY and IHW-BY:} BY rejects all hypotheses with p-value $P_i$ below a fixed threshold, while IHW-BY rejects hypotheses with $P_i \leq s_l(X_i)$, where $l \in \{1,2\}$ denotes the fold, and the threshold depends on the covariate $X_i$. The threshold is more lenient for hypotheses with smaller genomic distance $X_i$. For larger $X_i$, the threshold becomes smaller (more stringent); in this example application, it reaches $0$ for very large $X_i$. \textbf{E, F. Estimated conditional local fdr at the BY and IHW-BY rejection thresholds.} We observe that for BY the conditional local fdr varies widely, while for IHW-BY it is approximately balanced at the non-zero thresholds (note the different scales of the y-axis in panels E,F). The conditional density $f(t \mid x)$ is estimated by binning along $X_i$ and applying the Grenander estimator within each bin. We set $f(0 \mid x) = \infty$, so that the conditional local fdr is $0$ when $s(x)=0$. \label{fig:H3K27ac}}
\end{figure}

The results are shown in Figure~\ref{fig:H3K27ac}. IHW more than doubles the discoveries compared to the unweighted procedure while maintaining all formal guarantees of FDR control. Panel A shows the learned weight functions for the two folds. Upon applying the weighted BY procedure, the weights translate into thresholds for rejection: hypothesis $i$ is rejected if $P_i \leq W_i \; \hat{t}_{\text{IHW}}^*$ for some common choice of $\hat{t}^*_{\text{IHW}}$ and hypothesis-dependent $W_i$ (Panel D). In contrast, the BY procedure uses the same rejection threshold $\hat{t}_{\text{BY}}^*$ for all hypotheses (Panel C). As a consequence, the BY procedure had to be relatively stringent throughout, while IHW could be permissive at smaller and stringent only at higher distances.

There is another interpretation explaining why IHW increases power: it attempts to set thresholds in a way that balances the conditional local false discovery rate ($\fdr$), at least among the non-zero thresholds. This is shown in Panel F. Indeed, under certain assumptions, the optimal decision boundary is one of constant local $\fdr$, cf.~\citet[Theorem 2]{lei2016adapt}.  On the other hand, since BY thresholds only depend on the p-values, the local fdr varies widely and increases as a function of genomic distance, as seen in Panel E.

Finally, we note that the estimation method for the local fdr in Panels E and F is the same that was used to derive the weights. The local fdr estimates appear to be  noisy; even inaccurate estimates of the local fdr can lead to powerful weights (increase in number of discoveries). Furthermore, the frequentist guarantees of type-I error control of IHW are independent of and unaffected by (in)accuracies of the local fdr estimate.

%-----------------------------------------------
\section{Further relations to previous work}
\label{sec:related_work}
%-----------------------------------------------
Throughout this manuscript we have emphasized the relationship of the present research to previous work. In particular, in our numerical study in Section~\ref{sec:numerical_study} we compared IHW to previously developed methods for grouped multiple testing, multiple testing with continuous covariates and simultaneous two-sample testing. In this section we provide some further connections of IHW to previous work.

%-------------------------------------------------------------
\subsection{\citet*{ignatiadis2016data}}
\label{subsec:nmeth2016}
%-------------------------------------------------------------
The idea of cross-weighting for $\FDR$ control was introduced as one of three empirically promising heuristics by \citet*{ignatiadis2016data}; the other two heuristics being convex relaxations and regularization of the weights towards unity and/or low total variation. The contribution of this paper relative to \citet{ignatiadis2016data} is to clarify essential versus circumstantial concepts (e.g.,  \citet{ignatiadis2016data} only considered one possibility for weighting hypotheses through the Grenander estimator) and to establish formal, finite-sample FDR control for IHW-BH. We also show how the fundamental idea of cross-weighting applies beyond independence and introduce cross-weighted variants of the $k$-Bonferroni and BY procedures for $k$-FWER and FDR control under dependence.

%------------------------------------------------------------
\subsection{Sample splitting}
\label{subsec:datasplit}
%------------------------------------------------------------
One of the initial attempts at data-driven weights \citep{rubin2006method} used another form of data-splitting: consider the setting  where we start with a $m \times n$ data-matrix from which we get our p-values $P_i$ by calculating the test statistic in a row-wise fashion, say by applying a $t$-test for each row. Then one can calculate $m$ ``prior'' p-values $P_i''$ based on $n_1<n$ columns and derive prior weights $W_i$ based on $P_i''$. The remaining $n-n_1$ columns are used to compute p-values $P_i'$. Finally, a weighted multiple testing procedure is applied with p-values $P_i'$ and weights $W_i$. However, the authors then show that in this case it is more powerful to simply use an unweighted procedure with p-values $P_i$ calculated based on the whole dataset, rather than a weighted procedure with sample-splitting. \citet{habiger2014compound} pursue a similar approach. For IHW, we instead split horizontally (on hypotheses) rather than vertically (on samples), and the p-values $P_i$ are unaltered. 

%----------------------------------------------------------------------
\subsection{The weighted False Discovery Rate}
\label{subsec:discussion:hetero}
%----------------------------------------------------------------------
In this work, we have studied heterogeneous multiple testing with the aim of increasing power, while controlling the $k$-$\text{FWER}$ or the $\FDR$. However, in light of non-exchangeability, the cost of a false discovery to the researcher may not be uniform, but vary across hypotheses; e.g., it may be equal to $a_i \geq 0$ for hypothesis $H_i$. Then it is of scientific interest to control the weighted $\FDR$ of \citet{benjamini1997multiple} defined as
\begin{equation*}
\text{wFDR}(\textbf{a}) := \EE\left[ \frac{\sum_{i \in \Hnull} a_i \ind(H_i \text{ rejected})}{\sum_{i =1}^m a_i \ind(H_i \text{ rejected})} \ind\p{ \sum_{i =1}^m a_i \ind(H_i \text{ rejected}) >0 }  \right].
\end{equation*}
Similarly, the utility (benefit) $b_i$ of a true discovery may vary across hypotheses. Then, instead of maximizing the expected number of discoveries (cf.\ Section~\ref{sec:powerful_weighting_rules}), it may be more pertinent to maximize the expected total benefit. \citet{basu2018weighted} study optimal oracle procedures that achieve this optimization goal subject to control of $\text{wFDR}(\textbf{a})$, as well as data-driven procedures that achieve the same goal asymptotically. In future work it would be of interest to study whether cross-weighting may be applied to derive flexible and powerful procedures with finite-sample control of $\text{wFDR}(\textbf{a})$. We expect this to be tractable --for example by leveraging the results of \citet{ramdas2017unified}-- and useful if the utility $b_i$ is a function of the covariates, i.e., $b_i = b(X_i)$.

%------------------------------------------------
\section{Discussion}
\label{sec:disc}
%------------------------------------------------
Despite the ubiquitous uptake by the natural sciences of the concepts of multiple testing (and in particular the FDR), and despite ever growing volumes of data and possible hypothesis tests, surprisingly little attention has been paid to systematic approaches to account for hypothesis heterogeneity in order to increase detection power. While this may be justifiable in situations where power is large anyway, in many cases the costs of the underlying experiments or studies are substantial and increase with sample size, and the question of power decides over success or failure. In such cases, an approach that increases power compared to a baseline analysis, at no cost and by purely computational means, should be of interest.

Our approach is an instance of the value of large scale data \citep{efron2010large}: due to dataset size, modeling and inference opportunities open up that were previously irrelevant or impossible. In addition to the p-values $P_i$, our approach uses two further inputs: the covariates $X_i$ and the fold assignment. These are different concepts and their construction is unrelated to each other. The $X_i$ are informative about power and/or prior probability of the tests, but independent of $P_i$ under the null hypothesis. Meanwhile, the folds are constructed as a device for the cross-weighting scheme, in order to achieve type-I error control: we want independence of folds so that the weights do not lead to overfitting. Their choice is unrelated to power. Random folds are an easy default, but to get independent folds, it is then necessary to require global independence (Assumption~\ref{assumption:distrib}). When global independence cannot be assumed, the dependences are in many application scenarios---loosely speaking---``local'' (under some suitable choice of metric on the set of hypotheses). This can be used to construct folds that are independent, at least to sufficient approximation. Making such loose speak more precise requires specification of individual application scenarios and the associated domain knowledge, as in the example of Section~\ref{sec:hQTLexample}.

If, for a dataset at hand, independent folds cannot be achieved by any available fold-splitting scheme, it is possibly better not to try to address the dependences at the level of the multiple testing procedure, but upstream: strong, dataset-wide dependences often signal the need for a fundamental rethink of the analysis approach.

Sometimes, dataset-wide dependences are caused by so-called \emph{batch effects}. They are undesirable, uninteresting with respect to the scientific question, and can be reduced or avoided by good experimental design~\citep{leek2010batches}. Once they are a matter of fact, it is sometimes possible to remove them by mapping the data to a new set of properly ``normalized'' and ``batch-corrected'' variables \citep{leek2008general, stegle2010peer, wang2017confounder}.

If avoiding dependence by modifying the analysis upstream of the multiple testing treatment is not possible, the analyst should also consider whether multiple marginal hypothesis tests are indeed more appropriate than, say, dimension reduction, or a multivariate model with $\FDR$ guarantees \citep{candes2018panning,sesia2019gene, ren2020knockoffs}.

%---------------------------------------------------
\section*{Code availability and reproducibility}
%---------------------------------------------------
 The study is made fully third-party reproducible, and we provide its code in Github under the link \url{https://github.com/Huber-group-EMBL/covariate-powered-cross-weighted-multiple-testing}. The Bioconductor package IHW (\url{http://bioconductor.org/packages/IHW}) provides a user-friendly implementation of IHW-BH/Storey based on the Grenander estimator.
%--------------------------------------------------
\section*{Acknowledgments}
%--------------------------------------------------
We thank Judith Zaugg for making available data for the example in Section~\ref{sec:hQTLexample}, and Edgar Dobriban, William Fithian, Susan Holmes, Lihua Lei, Michael Love, Gesthimani Roumpani, Stelios Serghiou, Michael Sklar, Youngtak Sohn, Oliver Stegle, Mark van de Wiel and Britta Velten for helpful discussions and critical comments on the manuscript. We thank Stefan Wager, an anonymous associate editor and two anonymous reviewers for feedback that motivated us to substantially improve the manuscript. Michael Sklar proposed the counterexample from \supplementname~\ref{sec:counterexample}. W.H. acknowledges support from the German Federal Ministry of Education and Research, Grant MOFA, under grant contract No. 031L0171A. N.I. acknowledges support from a Ric Weiland Graduate Fellowship.

%--------------------------------------------------
\bibliographystyle{plainnat}
	
\bibliography{ddhw}

\begin{thebibliography}{86}
\providecommand{\natexlab}[1]{#1}
\providecommand{\url}[1]{\texttt{#1}}
\expandafter\ifx\csname urlstyle\endcsname\relax
  \providecommand{\doi}[1]{doi: #1}\else
  \providecommand{\doi}{doi: \begingroup \urlstyle{rm}\Url}\fi

\bibitem[Allison et~al.(2002)Allison, Gadbury, Heo, Fern{\'a}ndez, Lee, Prolla,
  and Weindruch]{allison2002mixture}
David~B Allison, Gary~L Gadbury, Moonseong Heo, Jos{\'e}~R Fern{\'a}ndez,
  Cheol-Koo Lee, Tomas~A Prolla, and Richard Weindruch.
\newblock A mixture model approach for the analysis of microarray gene
  expression data.
\newblock \emph{Computational Statistics \& Data Analysis}, 39\penalty0
  (1):\penalty0 1--20, 2002.

\bibitem[Arias-Castro and Chen(2017)]{arias2017distribution}
Ery Arias-Castro and Shiyun Chen.
\newblock Distribution-free multiple testing.
\newblock \emph{Electronic Journal of Statistics}, 11\penalty0 (1):\penalty0
  1983--2001, 2017.

\bibitem[Baillie et~al.(2018)Baillie, Bretherick, Haley, Clohisey, Gray,
  Neyton, Barrett, Stahl, Tenesa, Andersson, Brown, Faulkner, Lizio, Schaefer,
  Daub, Itoh, Kondo, Lassmann, Kawai, Mole, Bajic, Heutink, Rehli, Kawaji,
  Sandelin, Suzuki, Satsangi, Wells, Hacohen, Freeman, Hayashizaki, Carninci,
  Forrest, and and]{baillie2018_shared_activity_patterns}
J.~Kenneth Baillie, Andrew Bretherick, Christopher~S. Haley, Sara Clohisey,
  Alan Gray, Lucile P.~A. Neyton, Jeffrey Barrett, Eli~A. Stahl, Albert Tenesa,
  Robin Andersson, J.~Ben Brown, Geoffrey~J. Faulkner, Marina Lizio, Ulf
  Schaefer, Carsten Daub, Masayoshi Itoh, Naoto Kondo, Timo Lassmann, Jun
  Kawai, Damian Mole, Vladimir~B. Bajic, Peter Heutink, Michael Rehli, Hideya
  Kawaji, Albin Sandelin, Harukazu Suzuki, Jack Satsangi, Christine~A. Wells,
  Nir Hacohen, Thomas~C. Freeman, Yoshihide Hayashizaki, Piero Carninci,
  Alistair R.~R. Forrest, and David A.~Hume and.
\newblock Shared activity patterns arising at genetic susceptibility loci
  reveal underlying genomic and cellular architecture of human disease.
\newblock \emph{{PLOS} Computational Biology}, 14\penalty0 (3):\penalty0
  e1005934, 2018.

\bibitem[Barber and Cand{\`e}s(2015)]{barber2015controlling}
Rina~Foygel Barber and Emmanuel~J Cand{\`e}s.
\newblock Controlling the false discovery rate via knockoffs.
\newblock \emph{The Annals of Statistics}, 43\penalty0 (5):\penalty0
  2055--2085, 2015.

\bibitem[Basu et~al.(2018)Basu, Cai, Das, and Sun]{basu2018weighted}
Pallavi Basu, T~Tony Cai, Kiranmoy Das, and Wenguang Sun.
\newblock Weighted false discovery rate control in large-scale multiple
  testing.
\newblock \emph{Journal of the American Statistical Association}, 113\penalty0
  (523):\penalty0 1172--1183, 2018.

\bibitem[Benjamini(2008)]{benjamini2008}
Yoav Benjamini.
\newblock Comment: Microarrays, empirical {Bayes} and the two-groups model.
\newblock \emph{{S}tatistical {S}cience}, 23\penalty0 (1):\penalty0 23--28,
  2008.

\bibitem[Benjamini and Hochberg(1995)]{benjamini1995controlling}
Yoav Benjamini and Yosef Hochberg.
\newblock Controlling the false discovery rate: a practical and powerful
  approach to multiple testing.
\newblock \emph{Journal of the Royal Statistical Society. Series B (Statistical
  Methodology)}, pages 289--300, 1995.

\bibitem[Benjamini and Hochberg(1997)]{benjamini1997multiple}
Yoav Benjamini and Yosef Hochberg.
\newblock Multiple hypotheses testing with weights.
\newblock \emph{Scandinavian Journal of Statistics}, 24\penalty0 (3):\penalty0
  407--418, 1997.

\bibitem[Benjamini and Yekutieli(2001)]{benjamini2001control}
Yoav Benjamini and Daniel Yekutieli.
\newblock The control of the false discovery rate in multiple testing under
  dependency.
\newblock \emph{The Annals of Statistics}, pages 1165--1188, 2001.

\bibitem[Blanchard and Roquain(2008)]{blanchard2008two}
Gilles Blanchard and Etienne Roquain.
\newblock Two simple sufficient conditions for {FDR} control.
\newblock \emph{Electronic journal of Statistics}, 2:\penalty0 963--992, 2008.

\bibitem[Boca and Leek(2018)]{boca2018direct}
Simina~M Boca and Jeffrey~T Leek.
\newblock A direct approach to estimating false discovery rates conditional on
  covariates.
\newblock \emph{PeerJ}, 6:\penalty0 e6035, 2018.

\bibitem[Bonferroni(1935)]{bonferroni1935calcolo}
Carlo~E Bonferroni.
\newblock \emph{Il calcolo delle assicurazioni su gruppi di teste}.
\newblock Studi in Onore del Professore Salvatore Ortu Carboni, Rome, Italy,
  1935.

\bibitem[Bourgon et~al.(2010)Bourgon, Gentleman, and
  Huber]{bourgon2010independent}
Richard Bourgon, Robert Gentleman, and Wolfgang Huber.
\newblock Independent filtering increases detection power for high-throughput
  experiments.
\newblock \emph{Proceedings of the National Academy of Sciences}, 107\penalty0
  (21):\penalty0 9546--9551, 2010.

\bibitem[Cai and Sun(2009)]{cai2009simultaneous}
T~Tony Cai and Wenguang Sun.
\newblock Simultaneous testing of grouped hypotheses: Finding needles in
  multiple haystacks.
\newblock \emph{Journal of the American Statistical Association}, 104\penalty0
  (488), 2009.

\bibitem[Cai et~al.(2019)Cai, Sun, and Wang]{cai2016cars}
T~Tony Cai, Wenguang Sun, and Weinan Wang.
\newblock Covariate-assisted ranking and screening for large-scale two-sample
  inference.
\newblock \emph{Journal of the Royal Statistical Society: Series B (Statistical
  Methodology)}, 81\penalty0 (2):\penalty0 187--234, 2019.

\bibitem[Cand{\`e}s et~al.(2018)Cand{\`e}s, Fan, Janson, and
  Lv]{candes2018panning}
Emmanuel Cand{\`e}s, Yingying Fan, Lucas Janson, and Jinchi Lv.
\newblock Panning for gold: ‘model-{X}’ knockoffs for high dimensional
  controlled variable selection.
\newblock \emph{Journal of the Royal Statistical Society: Series B (Statistical
  Methodology)}, 2018.

\bibitem[Chernozhukov et~al.(2017)Chernozhukov, Chetverikov, Demirer, Duflo,
  Hansen, Newey, and Robins]{chernozhukov2017double}
Victor Chernozhukov, Denis Chetverikov, Mert Demirer, Esther Duflo, Christian
  Hansen, Whitney Newey, and James Robins.
\newblock Double/debiased machine learning for treatment and structural
  parameters.
\newblock \emph{The Econometrics Journal}, 2017.

\bibitem[Deb et~al.(2021)Deb, Saha, Guntuboyina, and Sen]{deb2018two}
Nabarun Deb, Sujayam Saha, Adityanand Guntuboyina, and Bodhisattva Sen.
\newblock Two-component mixture model in the presence of covariates.
\newblock \emph{Journal of the American Statistical Association}, pages 1--35,
  2021.

\bibitem[Dobriban et~al.(2015)Dobriban, Fortney, Kim, and
  Owen]{dobriban2015optimal}
Edgar Dobriban, Kristen Fortney, Stuart~K Kim, and Art~B Owen.
\newblock {O}ptimal multiple testing under a {G}aussian prior on the effect
  sizes.
\newblock \emph{Biometrika}, 102\penalty0 (4):\penalty0 753--766, 2015.

\bibitem[Du and Zhang(2014)]{du2014single}
Lilun Du and Chunming Zhang.
\newblock Single-index modulated multiple testing.
\newblock \emph{The Annals of Statistics}, 42\penalty0 (4):\penalty0 30--79,
  2014.

\bibitem[Durand(2017)]{durand2017adaptive}
Guillermo Durand.
\newblock Adaptive p-value weighting with power optimality.
\newblock \emph{arXiv preprint arXiv:1710.01094v1}, 2017.

\bibitem[Durand(2019)]{durand2019adaptive}
Guillermo Durand.
\newblock Adaptive $p$-value weighting with power optimality.
\newblock \emph{Electronic Journal of Statistics}, 13\penalty0 (2):\penalty0
  3336--3385, 2019.

\bibitem[Efron(2008)]{efron2008simultaneous}
Bradley Efron.
\newblock Simultaneous inference: When should hypothesis testing problems be
  combined?
\newblock \emph{The Annals of Applied Statistics}, pages 197--223, 2008.

\bibitem[Efron(2010)]{efron2010large}
Bradley Efron.
\newblock \emph{Large-scale inference: {E}mpirical {Bayes} methods for
  estimation, testing, and prediction}.
\newblock Cambridge University Press, 2010.

\bibitem[Efron et~al.(2001)Efron, Tibshirani, Storey, and
  Tusher]{efron2001empirical}
Bradley Efron, Robert Tibshirani, John~D Storey, and Virginia Tusher.
\newblock Empirical {Bayes} analysis of a microarray experiment.
\newblock \emph{Journal of the American Statistical Association}, 96\penalty0
  (456):\penalty0 1151--1160, 2001.

\bibitem[Ferkingstad et~al.(2008)Ferkingstad, Frigessi, Rue, Thorleifsson, and
  Kong]{ferkingstad2008unsupervised}
Egil Ferkingstad, Arnoldo Frigessi, H{\aa}vard Rue, Gudmar Thorleifsson, and
  Augustine Kong.
\newblock Unsupervised empirical {Bayesian} multiple testing with external
  covariates.
\newblock \emph{The Annals of Applied Statistics}, pages 714--735, 2008.

\bibitem[Fortney et~al.(2015)Fortney, Dobriban, Garagnani, Pirazzini, Monti,
  Mari, Atzmon, Barzilai, Franceschi, Owen, and Kim]{fortney2015genome}
Kristen Fortney, Edgar Dobriban, Paolo Garagnani, Chiara Pirazzini, Daniela
  Monti, Daniela Mari, Gil Atzmon, Nir Barzilai, Claudio Franceschi, Art~B
  Owen, and Stuart~K Kim.
\newblock Genome-wide scan informed by age-related disease identifies loci for
  exceptional human longevity.
\newblock \emph{PLoS Genetics}, 11\penalty0 (12):\penalty0 e1005728, 2015.

\bibitem[Genovese and Wasserman(2004)]{genovese2004stochastic}
Christopher Genovese and Larry Wasserman.
\newblock A stochastic process approach to false discovery control.
\newblock \emph{The Annals of Statistics}, pages 1035--1061, 2004.

\bibitem[Genovese et~al.(2006)Genovese, Roeder, and
  Wasserman]{genovese2006false}
Christopher~R Genovese, Kathryn Roeder, and Larry Wasserman.
\newblock False discovery control with p-value weighting.
\newblock \emph{Biometrika}, 93\penalty0 (3):\penalty0 509--524, 2006.

\bibitem[Grenander(1956)]{grenander1956theory}
Ulf Grenander.
\newblock On the theory of mortality measurement.
\newblock \emph{Scandinavian Actuarial Journal}, 1956\penalty0 (1):\penalty0
  70--96, 1956.

\bibitem[Grubert et~al.(2015)Grubert, Zaugg, Kasowski, Ursu, Spacek, Martin,
  Greenside, Srivas, Phanstiel, Pekowska, et~al.]{grubert2015genetic}
Fabian Grubert, Judith~B Zaugg, Maya Kasowski, Oana Ursu, Damek~V Spacek,
  Alicia~R Martin, Peyton Greenside, Rohith Srivas, Doug~H Phanstiel,
  Aleksandra Pekowska, et~al.
\newblock Genetic control of chromatin states in humans involves local and
  distal chromosomal interactions.
\newblock \emph{Cell}, 162\penalty0 (5):\penalty0 1051--1065, 2015.

\bibitem[Guo and Sarkar(2019)]{guo2016adaptive}
Wenge Guo and Sanat Sarkar.
\newblock Adaptive controls of {FWER} and {FDR} under block dependence.
\newblock \emph{Journal of Statistical Planning and Inference}, 2019.

\bibitem[Habiger(2017)]{habiger2014weighted}
Joshua~D Habiger.
\newblock Adaptive false discovery rate control for heterogeneous data.
\newblock \emph{Statistica Sinica}, pages 1731--1756, 2017.

\bibitem[Habiger and Pe{\~n}a(2014)]{habiger2014compound}
Joshua~D Habiger and Edsel~A Pe{\~n}a.
\newblock Compound p-value statistics for multiple testing procedures.
\newblock \emph{Journal of multivariate analysis}, 126:\penalty0 153--166,
  2014.

\bibitem[Hastie et~al.(2009)Hastie, Tibshirani, and
  Friedman]{hastie2009elements}
T~Hastie, R~Tibshirani, and J~Friedman.
\newblock \emph{The Elements of Statistical Learning: Data Mining, Inference,
  and Prediction, Second Edition}.
\newblock Springer Series in Statistics. Springer, 2009.
\newblock ISBN 9780387848587.

\bibitem[Heesen and Janssen(2015)]{heesen2015inequalities}
Philipp Heesen and Arnold Janssen.
\newblock Inequalities for the false discovery rate ({FDR}) under dependence.
\newblock \emph{Electronic Journal of Statistics}, 9\penalty0 (1):\penalty0
  679--716, 2015.

\bibitem[Hu et~al.(2010)Hu, Zhao, and Zhou]{hu2010false}
James~X Hu, Hongyu Zhao, and Harrison~H Zhou.
\newblock False discovery rate control with groups.
\newblock \emph{Journal of the American Statistical Association}, 105\penalty0
  (491), 2010.

\bibitem[Ignatiadis and Wager(2019)]{ignatiadis2019covariate}
Nikolaos Ignatiadis and Stefan Wager.
\newblock Covariate-powered empirical {B}ayes estimation.
\newblock In \emph{Advances in {N}eural {I}nformation {P}rocessing {S}ystems},
  pages 9620--9632, 2019.

\bibitem[Ignatiadis et~al.(2016)Ignatiadis, Klaus, Zaugg, and
  Huber]{ignatiadis2016data}
Nikolaos Ignatiadis, Bernd Klaus, Judith~B Zaugg, and Wolfgang Huber.
\newblock Data-driven hypothesis weighting increases detection power in
  genome-scale multiple testing.
\newblock \emph{Nature Methods}, 2016.

\bibitem[Kim(2020)]{lpsymphony}
Vladislav Kim.
\newblock \emph{lpsymphony: Symphony integer linear programming solver in R},
  2020.
\newblock http://R-Forge.R-project.org/projects/rsymphony,
  https://projects.coin-or.org/SYMPHONY,
  http://www.coin-or.org/download/source/SYMPHONY/.

\bibitem[Klaus and Strimmer(2011)]{klaus2011learning}
Bernd Klaus and Korbinian Strimmer.
\newblock Learning false discovery rates by fitting sigmoidal threshold
  functions.
\newblock \emph{Journal de la Société Française de Statistique},
  152\penalty0 (2):\penalty0 39--50, 2011.

\bibitem[Korthauer et~al.(2019)Korthauer, Kimes, Duvallet, Reyes, Subramanian,
  Teng, Shukla, Alm, and Hicks]{korthauer2018practical}
Keegan Korthauer, Patrick~K Kimes, Claire Duvallet, Alejandro Reyes, Ayshwarya
  Subramanian, Mingxiang Teng, Chinmay Shukla, Eric~J Alm, and Stephanie~C
  Hicks.
\newblock A practical guide to methods controlling false discoveries in
  computational biology.
\newblock \emph{Genome biology}, 20\penalty0 (1):\penalty0 118, 2019.

\bibitem[Leek and Storey(2008)]{leek2008general}
Jeffrey~T Leek and John~D Storey.
\newblock A general framework for multiple testing dependence.
\newblock \emph{Proceedings of the National Academy of Sciences}, 105\penalty0
  (48):\penalty0 18718--18723, 2008.

\bibitem[Leek et~al.(2010)Leek, Scharpf, Bravo, Simcha, Langmead, Johnson,
  Geman, Baggerly, and Irizarry]{leek2010batches}
Jeffrey~T Leek, Robert~B Scharpf, H{\'{e}}ctor~Corrada Bravo, David Simcha,
  Benjamin Langmead, W~Evan Johnson, Donald Geman, Keith Baggerly, and Rafael~A
  Irizarry.
\newblock Tackling the widespread and critical impact of batch effects in
  high-throughput data.
\newblock \emph{Nature Reviews Genetics}, 11\penalty0 (10):\penalty0 733--739,
  2010.

\bibitem[Lehmann and Romano(2005)]{lehmann2005testing}
EL~Lehmann and JP~Romano.
\newblock \emph{Testing Statistical Hypotheses}.
\newblock Springer Texts in Statistics. Springer, 2005.
\newblock ISBN 9780387988641.

\bibitem[Lei and Fithian(2018)]{lei2016adapt}
Lihua Lei and William Fithian.
\newblock {AdaPT}: an interactive procedure for multiple testing with side
  information.
\newblock \emph{Journal of the Royal Statistical Society: Series B (Statistical
  Methodology)}, 2018.

\bibitem[Li and Barber(2019)]{li2019multiple}
Ang Li and Rina~Foygel Barber.
\newblock Multiple testing with the structure-adaptive {B}enjamini--{H}ochberg
  algorithm.
\newblock \emph{Journal of the Royal Statistical Society: Series B (Statistical
  Methodology)}, 81\penalty0 (1):\penalty0 45--74, 2019.

\bibitem[Liang and Nettleton(2012)]{liang2012adaptive}
Kun Liang and Dan Nettleton.
\newblock Adaptive and dynamic adaptive procedures for false discovery rate
  control and estimation.
\newblock \emph{Journal of the Royal Statistical Society: Series B (Statistical
  Methodology)}, 74\penalty0 (1):\penalty0 163--182, 2012.

\bibitem[Liu(2014)]{liu2014incorporation}
Weidong Liu.
\newblock Incorporation of sparsity information in large-scale multiple
  two-sample $ t $ tests.
\newblock \emph{arXiv preprint arXiv:1410.4282}, 2014.

\bibitem[Lougee-Heimer(2003)]{lougee2003coinor}
Robin Lougee-Heimer.
\newblock The common optimization interface for operations research: Promoting
  open-source software in the operations research community.
\newblock \emph{IBM Journal of Research and Development}, 47\penalty0
  (1):\penalty0 57--66, 2003.

\bibitem[Markitsis and Lai(2010)]{markitsis2010censored}
Anastasios Markitsis and Yinglei Lai.
\newblock A censored beta mixture model for the estimation of the proportion of
  non-differentially expressed genes.
\newblock \emph{Bioinformatics}, 26\penalty0 (5):\penalty0 640--646, 2010.

\bibitem[Nie and Wager(2020)]{nie2017learning}
Xinkun Nie and Stefan Wager.
\newblock {Quasi-Oracle Estimation of Heterogeneous Treatment Effects}.
\newblock \emph{Biometrika}, 09 2020.
\newblock asaa076.

\bibitem[Ochoa et~al.(2015)Ochoa, Storey, Llinás, and Singh]{ochoa2014beyond}
Alejandro Ochoa, John~D Storey, Manuel Llinás, and Mona Singh.
\newblock Beyond the {E}-value: Stratified statistics for protein domain
  prediction.
\newblock \emph{{PLoS} {C}omputational {B}iology}, 11\penalty0 (11):\penalty0
  e1004509, 11 2015.

\bibitem[Pe{\~n}a et~al.(2011)Pe{\~n}a, Habiger, and Wu]{pena2011power}
Edsel~A Pe{\~n}a, Joshua~D Habiger, and Wensong Wu.
\newblock Power-enhanced multiple decision functions controlling family-wise
  error and false discovery rates.
\newblock \emph{The Annals of Statistics}, 39\penalty0 (1):\penalty0 556--583,
  2011.

\bibitem[Ploner et~al.(2006)Ploner, Calza, Gusnanto, and
  Pawitan]{ploner2006multidimensional}
Alexander Ploner, Stefano Calza, Arief Gusnanto, and Yudi Pawitan.
\newblock Multidimensional local false discovery rate for microarray studies.
\newblock \emph{Bioinformatics}, 22\penalty0 (5):\penalty0 556--565, 2006.

\bibitem[Ramdas et~al.(2019)Ramdas, Barber, Wainwright, and
  Jordan]{ramdas2017unified}
Aaditya~K Ramdas, Rina~F Barber, Martin~J Wainwright, and Michael~I Jordan.
\newblock A unified treatment of multiple testing with prior knowledge using
  the p-filter.
\newblock \emph{The Annals of Statistics}, 47\penalty0 (5):\penalty0
  2790--2821, 2019.

\bibitem[Ren and Cand{\`e}s(2020)]{ren2020knockoffs}
Zhimei Ren and Emmanuel Cand{\`e}s.
\newblock Knockoffs with side information.
\newblock \emph{arXiv preprint arXiv:2001.07835}, 2020.

\bibitem[Rockafellar(1970)]{rockafellar1970convex}
R~Tyrrell Rockafellar.
\newblock \emph{Convex analysis}.
\newblock Number~28 in Princeton Landmarks in Mathematics and Physics.
  Princeton university press, 1970.

\bibitem[Roeder and Wasserman(2009)]{roeder2009genome}
Kathryn Roeder and Larry Wasserman.
\newblock Genome-wide significance levels and weighted hypothesis testing.
\newblock \emph{{S}tatistical {S}cience}, 24\penalty0 (4):\penalty0 398, 2009.

\bibitem[Roeder et~al.(2007)Roeder, Devlin, and Wasserman]{roeder2007improving}
Kathryn Roeder, Bernie Devlin, and Larry Wasserman.
\newblock Improving power in genome-wide association studies: weights tip the
  scale.
\newblock \emph{Genetic {E}pidemiology}, 31\penalty0 (7):\penalty0 741--747,
  2007.

\bibitem[Romano and Wolf(2010)]{romano2010balanced}
Joseph~P Romano and Michael Wolf.
\newblock Balanced control of generalized error rates.
\newblock \emph{The {A}nnals of {S}tatistics}, 38\penalty0 (1):\penalty0
  598--633, 2010.

\bibitem[Roquain and Van De~Wiel(2009)]{roquain2009optimal}
Etienne Roquain and Mark Van De~Wiel.
\newblock Optimal weighting for false discovery rate control.
\newblock \emph{Electronic Journal of Statistics}, 3:\penalty0 678--711, 2009.

\bibitem[Rubin et~al.(2006)Rubin, Dudoit, and Van~der Laan]{rubin2006method}
Daniel Rubin, Sandrine Dudoit, and Mark Van~der Laan.
\newblock A method to increase the power of multiple testing procedures through
  sample splitting.
\newblock \emph{Statistical Applications in Genetics and Molecular Biology},
  5\penalty0 (1), 2006.

\bibitem[Sankaran and Holmes(2014)]{sankaran2014structssi}
Kris Sankaran and Susan Holmes.
\newblock struct{SSI}: Simultaneous and selective inference for grouped or
  hierarchically structured data.
\newblock \emph{Journal of statistical software}, 59\penalty0 (13):\penalty0 1,
  2014.

\bibitem[Schick(1986)]{schick1986asymptotically}
Anton Schick.
\newblock On asymptotically efficient estimation in semiparametric models.
\newblock \emph{The Annals of Statistics}, pages 1139--1151, 1986.

\bibitem[Scott et~al.(2015)Scott, Kelly, Smith, Zhou, and Kass]{scott2014false}
James~G Scott, Ryan~C Kelly, Matthew~A Smith, Pengcheng Zhou, and Robert~E
  Kass.
\newblock False discovery rate regression: an application to neural synchrony
  detection in primary visual cortex.
\newblock \emph{Journal of the American Statistical Association}, 110\penalty0
  (510):\penalty0 459--471, 2015.

\bibitem[Sesia et~al.(2019)Sesia, Sabatti, and Cand{\`e}s]{sesia2019gene}
Matteo Sesia, Chiara Sabatti, and Emmanuel~J Cand{\`e}s.
\newblock Gene hunting with knockoffs for hidden markov models.
\newblock \emph{Biometrika}, 106:\penalty0 1–18, 2019.

\bibitem[Shabalin(2012)]{shabalin2012matrix}
Andrey~A Shabalin.
\newblock Matrix e{QTL}: ultra fast e{QTL} analysis via large matrix
  operations.
\newblock \emph{Bioinformatics}, 28\penalty0 (10):\penalty0 1353--1358, 2012.

\bibitem[Stegle et~al.(2010)Stegle, Parts, Durbin, and Winn]{stegle2010peer}
Oliver Stegle, Leopold Parts, Richard Durbin, and John Winn.
\newblock A {B}ayesian framework to account for complex non-genetic factors in
  gene expression levels greatly increases power in {eQTL} studies.
\newblock \emph{{PLoS} Computational Biology}, 6\penalty0 (5):\penalty0
  e1000770, 2010.

\bibitem[Storey(2003)]{storey2003positive}
John~D Storey.
\newblock The positive false discovery rate: A {Bayesian} interpretation and
  the q-value.
\newblock \emph{The Annals of Statistics}, pages 2013--2035, 2003.

\bibitem[Storey(2007)]{storey2007optimal}
John~D Storey.
\newblock The optimal discovery procedure: a new approach to simultaneous
  significance testing.
\newblock \emph{Journal of the Royal Statistical Society: Series B (Statistical
  Methodology)}, 69\penalty0 (3):\penalty0 347--368, 2007.

\bibitem[Storey et~al.(2004)Storey, Taylor, and Siegmund]{storey2004strong}
John~D Storey, Jonathan~E Taylor, and David Siegmund.
\newblock Strong control, conservative point estimation and simultaneous
  conservative consistency of false discovery rates: a unified approach.
\newblock \emph{Journal of the Royal Statistical Society: Series B (Statistical
  Methodology)}, 66\penalty0 (1):\penalty0 187--205, 2004.

\bibitem[Strimmer(2008{\natexlab{a}})]{strimmer2008fdrtool}
Korbinian Strimmer.
\newblock fdrtool: a versatile {R} package for estimating local and tail
  area-based false discovery rates.
\newblock \emph{Bioinformatics}, 24\penalty0 (12):\penalty0 1461--1462,
  2008{\natexlab{a}}.

\bibitem[Strimmer(2008{\natexlab{b}})]{strimmer2008unified}
Korbinian Strimmer.
\newblock A unified approach to false discovery rate estimation.
\newblock \emph{{BMC} {B}ioinformatics}, 9\penalty0 (1):\penalty0 303,
  2008{\natexlab{b}}.

\bibitem[Sun et~al.(2006)Sun, Craiu, Paterson, and Bull]{sun2006stratified}
Lei Sun, Radu~V Craiu, Andrew~D Paterson, and Shelley~B Bull.
\newblock Stratified false discovery control for large-scale hypothesis testing
  with application to genome-wide association studies.
\newblock \emph{Genetic {E}pidemiology}, 30\penalty0 (6):\penalty0 519--530,
  2006.

\bibitem[Sun and Cai(2007)]{sun2007oracle}
Wenguang Sun and T~Tony Cai.
\newblock Oracle and adaptive compound decision rules for false discovery rate
  control.
\newblock \emph{Journal of the American Statistical Association}, 102\penalty0
  (479):\penalty0 901--912, 2007.

\bibitem[Sun and Cai(2009)]{sun2009large}
Wenguang Sun and T~Tony Cai.
\newblock Large-scale multiple testing under dependence.
\newblock \emph{Journal of the Royal Statistical Society: Series B (Statistical
  Methodology)}, 71\penalty0 (2):\penalty0 393--424, 2009.

\bibitem[Tusher et~al.(2001)Tusher, Tibshirani, and
  Chu]{tusher2001significance}
Virginia~Goss Tusher, Robert Tibshirani, and Gilbert Chu.
\newblock Significance analysis of microarrays applied to the ionizing
  radiation response.
\newblock \emph{Proceedings of the National Academy of Sciences}, 98\penalty0
  (9):\penalty0 5116--5121, 2001.

\bibitem[van De~Wiel et~al.(2016)van De~Wiel, Lien, Verlaat, van Wieringen, and
  Wilting]{wiel2016better}
Mark~A van De~Wiel, Tonje~G Lien, Wina Verlaat, Wessel~N van Wieringen, and
  Saskia~M Wilting.
\newblock Better prediction by use of co-data: adaptive group-regularized ridge
  regression.
\newblock \emph{Statistics in Medicine}, 35\penalty0 (3):\penalty0 368--381,
  2016.

\bibitem[van~der Vaart(2000)]{van2000asymptotic}
AW~van~der Vaart.
\newblock \emph{Asymptotic Statistics}.
\newblock Cambridge Series in Statistical and Probabilistic Mathematics.
  Cambridge University Press, 2000.
\newblock ISBN 9781107268449.

\bibitem[Wager and Athey(2018)]{wager2018estimation}
Stefan Wager and Susan Athey.
\newblock Estimation and inference of heterogeneous treatment effects using
  random forests.
\newblock \emph{Journal of the American Statistical Association}, 113\penalty0
  (523):\penalty0 1228--1242, 2018.

\bibitem[Wang et~al.(2017)Wang, Zhao, Hastie, and Owen]{wang2017confounder}
Jingshu Wang, Qingyuan Zhao, Trevor Hastie, and Art~B Owen.
\newblock Confounder adjustment in multiple hypothesis testing.
\newblock \emph{The Annals of Statistics}, 45\penalty0 (5):\penalty0
  1863--1894, 2017.

\bibitem[Wang(2018)]{wang2018weighted}
Li~Wang.
\newblock Weighted multiple testing procedure for grouped hypotheses with
  k-{FWER} control.
\newblock \emph{Computational Statistics}, pages 1--25, 2018.

\bibitem[Zhang et~al.(2017)Zhang, Xia, Zou, and Tse]{zhang2017neuralfdr}
Martin~J Zhang, Fei Xia, James~Y Zou, and David Tse.
\newblock {N}eural{FDR}: Learning discovery thresholds from hypothesis
  features.
\newblock In \emph{Advances in {N}eural {I}nformation {P}rocessing {S}ystems},
  pages 1540--1549, 2017.

\bibitem[Zhang et~al.(2019)Zhang, Xia, and Zou]{zhang2018adafdr}
Martin~J Zhang, Fei Xia, and James Zou.
\newblock Fast and covariate-adaptive method amplifies detection power in
  large-scale multiple hypothesis testing.
\newblock \emph{Nature {C}ommunications}, 10\penalty0 (1):\penalty0 1--11,
  2019.

\bibitem[Zhao and Zhang(2014)]{zhao2014weighted}
Haibing Zhao and Jiajia Zhang.
\newblock Weighted p--value procedures for controlling {FDR} of grouped
  hypotheses.
\newblock \emph{Journal of Statistical Planning and Inference}, 2014.

\end{thebibliography}
%--------------------------------------------------

\newpage

\titleformat{\section}{\normalfont\large\bfseries}{Supplement \thesection:}{0.3em}{}
\setcounter{section}{0}
\setcounter{page}{1}

\renewcommand{\thepage}{S\arabic{page}}
\renewcommand{\thesection}{S\arabic{section}}

% ------------------------------------------------------------------------------------
\section{Finite-sample results for FDR control of IHW}
\label{sec:finite_sample_proofs}
%------------------------------------------------------------------------------------

Throughout Supplementary Section~\ref{sec:finite_sample_proofs}, the weights $W_i$ are considered random. Occasionally we explicitly condition on the weights; in which case we verify how the conditioning on (subsets of) weights influences conditional distributions.

%------------------------------------------------------------------------------------
\subsection{A preliminary lemma}
%------------------------------------------------------------------------------------
They key property of IHW that enables finite-sample type-I error control is the following: cross-weighting makes the p-values and their weights independent of each other. This was already demonstrated in the beginning of the proof of Theorem~\ref{thm:IHW-bonf} in Section~\ref{subsec:kfwer_control}. Here we formalize this result through the following Lemma:

\begin{lem}
\label{lem:Indep}
Let $(W_i)_{i \in [m]}$ be honest weights (Specification~\ref{assumption:honest_weights}) w.r.t. the partition $I_1, \dotsc, I_K$ of $[m]$. 

If $((P_i, X_i))_{i \in [m]}$ satisfy Assumption~\ref{assumption:distrib_dep}, then:
\begin{enumerate}[label=(\alph*)]
 \item For all $\ell \in [K]$ and all $i \in \Hnull \cap I_{\ell}$, $P_i$ is independent of $(W_k)_{k \in I_{\ell}}$. In particular $P_i$ is independent of $W_i$ ($P_i \perp W_i$) for all $i \in \Hnull$.
\end{enumerate}

The conclusion may be strengthened if instead $((P_i, X_i))_{i \in [m]}$ satisfy Assumption~\ref{assumption:distrib}:

\begin{enumerate}
	\item[(a')]  For all $\ell \in [K]$, $(P_i)_{i \in \Hnull \cap I_{\ell}}$ is independent of $(W_i)_{i \in \Hnull \cap I_{\ell}}$.
	\item[(b')] For all $\ell \in [K]$, $(P_i)_{i \in \Hnull \cap I_{\ell}}$ are jointly independent and super-uniform conditionally on $(W_i)_{i \in \Hnull \cap I_{\ell}}$.
\end{enumerate}
\end{lem}

\begin{proof}
We prove (a); the other statements follow similarly. Fix $\ell \in [K]$ and let $i \in \Hnull \cap I_{\ell}$. By definition of honesty (Specification~\ref{assumption:honest_weights}), $(W_k)_{k \in I_{\ell}}$ is a function only of $(P_i)_{i \in I_{\ell}^c}$ and $\mathbf{X} = (X_i)_{i \in [m]}$. It thus suffices to argue that $P_i$ is independent of $((P_i)_{i \in I_{\ell}^c}, \mathbf{X})$. Writing the latter as $((P_i)_{i \in I_{\ell}^c}, (X_i)_{i \in I_{\ell}^c}, (X_i)_{i \in I_{\ell}})$ we conclude as a consequence of parts (a) and (b) of Assumption~\ref{assumption:distrib_dep}. 
\end{proof}

%------------------------------------------------------------------------------------
\subsection{The IHW-BH procedure under independence: Proof of Theorem~\ref{thm:IHWc}}
\label{subsec:IHW_BH_proof}
%------------------------------------------------------------------------------------
\begin{proof}
Let $\mathbf{W}$ be the weights and $\hat{k}$ the number of discoveries after applying IHW-BH at level $\alpha$ and with censoring level $\tau$. Also write $\mathbf{X} = (X_1, \dotsc, X_m)$, $\mathbf{P} = (P_1, \dotsc, P_m)$ and $\ind(\mathbf{P} \leq \tau) = (\ind(P_1 \leq \tau ), \dotsc, \ind(P_m \leq \tau))$. Here $\ind(P_i \leq \tau)$ is the indicator function that is $1$ when $P_i \leq \tau$ and $0$ otherwise. 

We first give a high level idea regarding the proof. To bound the $\FDR$ we seek to bound expectations of $\ind(H_i \text{ rejected})/(\hat{k} \lor 1)$, i.e., of $\ind(P_i \leq \alpha W_i \hat{k}/m, \; P_i\leq \tau)/(\hat{k} \lor 1)$ where $i$ is null.\footnote{We use the notation $a \lor b = \max\cb{a,b}$, $a \land b = \min\cb{a,b}$.} If $W_i, \hat{k}$ were independent of $P_i$, then we could directly upper bound this expectation by $\EE[(\alpha W_i \hat{k}/m)/(\hat{k} \lor 1)] \leq \EE[\alpha W_i/m]$ from which $\FDR$ control would follow by summing over all $i$. Honesty (Specification~\ref{assumption:honest_weights}) makes---in the way of Lemma~\ref{lem:Indep}---$P_i$ and its weight $W_i$ (for a single null $i$) independent. However, $P_i$ directly influences $\hat{k}$. This is true also for unweighted BH and weighted BH with deterministic weights, yet here $P_i$ also indirectly influences $\hat{k}$ through the weights $W_j,\; j \neq i$. Nevertheless, we will argue that the conclusion may still be salvaged: $\tau$-censoring (Specification~\ref{assumption:tau_stopped_weights}) ensures that on the event $\cb{P_i \leq \tau}$ the exact value of $P_i$ cannot influence weights $W_j,\; j \neq i$. Furthermore, it suffices to only consider the event $\cb{P_i \leq \tau}$ (in turn for each null $i$), since $i$ will never get rejected when $P_i > \tau$ (by Definition~\ref{defn:tau-wbh}).

We make the above intuition rigorous using a leave-one-out argument as in the proof idea of~\citet{li2019multiple}. Let us first pay attention to a single index $i \in [m]$. We denote by $k_i$ the number of discoveries of IHW-BH if $\mathbf{P}$ gets replaced by $\mathbf{P}_{i \mapsto 0} = (P_1, \dotsc, P_{i-1}, 0, P_{i+1}, P_m)$. Note that because the weights are $\tau$-censored (Specification~\ref{assumption:tau_stopped_weights}) , the tuple of weights $\mathbf{W}$ remains unchanged by replacing $P_i$ by $0$ on the event $\{ P_i \leq \tau \}$. Furthermore, by definition of the $\tau$-censored weighted BH procedure (Definition~\ref{defn:tau-wbh}), the rejection of $H_i$ (by IHW-BH applied to $\mathbf{P}$) implies that $P_i \leq \left(\frac{\alpha W_i \hat{k}}{m} \right) \land \tau$. In particular, the event $\{ P_i \leq \tau \}$ holds. Furthermore, for any $k\geq \hat{k}$, counting the entries of $\mathbf{P}$, respectively $\mathbf{P}_{i \mapsto 0}$, that are not greater than the corresponding entries of $\left(\frac{\alpha \mathbf{W}k}{m} \right) \land \tau$ must yield the same number. We conclude that: 
\begin{equation*} 
H_i \text{ rejected} \Rightarrow \hat{k} = k_i \geq 1.
\end{equation*}
Therefore,
\begin{equation*}  
H_i \text{ rejected} \Rightarrow P_i \leq \frac{\alpha W_i k_i}{m} \land \tau\,.
\end{equation*}
Note at this point that we can assume without loss of generality that $\PP[P_i \leq \tau] >0$  for all $i \in \Hnull$. Otherwise, just set $\Hnull' =  \{i \in \Hnull \mid \PP[P_i \leq \tau] >0\}$ and all the steps below will go through essentially unchanged with $\Hnull'$ replacing $\Hnull$.
\noindent For $i \in \Hnull$ and conditioning on the event $\{P_i \leq \tau\}$ and on the random vectors $\mathbf{W}, \mathbf{X} ,\mathbf{P}_{i \mapsto 0},\ind(\mathbf{P} \leq \tau )$, we get
\begin{equation*}
\begin{aligned}
&\PP[ H_i \text{ rejected} \mid   P_i \leq \tau, \mathbf{W}, \mathbf{X} ,\mathbf{P}_{i \mapsto 0}, \ind(\mathbf{P} \leq \tau)] \\
\leq\; &\PP[ P_i \leq  \frac{\alpha W_i k_i}{m} \land \tau \mid  P_i \leq \tau, \mathbf{W}, \mathbf{X} ,\mathbf{P}_{i \mapsto 0}, \ind(\mathbf{P} \leq \tau)]\\
\leq\;  &\frac{\alpha W_i k_i}{m \PP[P_i \leq \tau]}\,.
\end{aligned}
\end{equation*}
This follows because for $i \in \Hnull$ it holds that $P_i$ is super-uniform, $\PP[P_i \leq \tau] >0$ and $P_i$ is independent of $(\mathbf{P}_{i\mapsto 0}, \mathbf{X})$ and also because $k_i$, $\mathbf{W}$, $\ind(\mathbf{P} \leq \tau)$ are functions of $(\mathbf{P}_{i\mapsto 0}, \mathbf{X})$ on the event $\{ P_i \leq \tau \}$. It then follows that
\begin{equation*}
\begin{aligned}
&\EE\left[\frac{\ind(H_i \text{ rejected})}{\hat{k} \lor 1} \;\middle|\; P_i \leq \tau, \mathbf{W}, \mathbf{X} ,\mathbf{P}_{i \mapsto 0}, \ind(\mathbf{P} \leq \tau) \right]\\
=\; &\EE\left[\frac{\ind(H_i \text{ rejected})}{k_i \lor 1} \;\middle|\; P_i \leq \tau, \mathbf{W}, \mathbf{X} ,\mathbf{P}_{i \mapsto 0}, \ind(\mathbf{P} \leq \tau)\right]\\
\leq\; & \frac{\alpha W_i}{m\PP[P_i \leq \tau]}\,.
\end{aligned}
\end{equation*}
Moreover, by marginalization over $\mathbf{P}_{i \mapsto 0}$ and $\mathbf{X}$ (and noting again that $\ind(H_i \text{ rejected})=0$ when $\ind(P_i \leq \tau) = 0$),
\begin{equation*}
\EE\left[\frac{\ind(H_i \text{ rejected})}{\hat{k}\lor1} \mid \mathbf{W}, \ind(\mathbf{P} \leq \tau)  \right] \leq \frac{\alpha W_i}{m\PP[P_i \leq \tau]}\ind(P_i \leq \tau)\,.
\end{equation*}
In total, we thus get
\begin{equation*}
\EE[\FDP \mid \mathbf{W}, \ind(\mathbf{P} \leq \tau)] =\; \EE\left[ \frac{\sum_{i \in \Hnull} \ind(H_i \text{ rejected})}{\hat{k}\lor1} \mid \mathbf{W}, \ind(\mathbf{P} \leq \tau) \right] \leq\; \sum_{i \in \Hnull} \frac{\alpha W_i}{m\PP[P_i \leq \tau]}\ind(P_i \leq \tau)\,.
\end{equation*}
At this point we diverge from the proof of \citet{li2019multiple} and take advantage of honesty (Specification~\ref{assumption:honest_weights}) through Lemma~\ref{lem:Indep}.
\begin{equation*}
\begin{aligned}
\EE[\FDP] &= \EE[\EE[\FDP \mid \mathbf{W},\ind(\mathbf{P} \leq \tau)]] \\
&\leq \sum_{i\in \Hnull} \EE\left[\frac{\alpha W_i}{m\PP[P_i \leq \tau]}\ind(P_i \leq \tau)\right] \\
&= \sum_{i \in \Hnull} \frac{\alpha}{m\PP[P_i \leq \tau]}\EE\left[W_i\right]\EE\left[\ind(P_i \leq \tau)\right]\\
&\leq \frac{\alpha}{m}\EE\left[\sum_{i=1}^m W_i\right]\\
&= \alpha\,.
\end{aligned}
\end{equation*}

\noindent Going from the second to the third line, we used that for $i \in \Hnull$, $P_i$ is independent of $W_i$, which holds from Lemma~\ref{lem:Indep}(a'). In the last step, we used Part (b) of the Honesty specification.

\end{proof}

%---------------------------------------------------------------
\subsection{Counterexample to demonstrate that honesty of weights does not suffice for FDR control (due to M. Sklar)}
\label{sec:counterexample}
%---------------------------------------------------------------
In this section, we provide a counterexample that the result of Theorem~\ref{thm:IHWc} no longer holds if we drop the assumption of $\tau$-censored weighting. This is in contrast e.g., to the conclusion of Theorem~\ref{thm:IHW-bonf} for $k$-Bonferroni, wherein honesty of the weights suffices (along with distributional assumptions on $\allpairs$). 

Our agenda is as follows: for $m=4$, we construct $\allpairs$ under the global null such that Assumption~\ref{assumption:distrib} holds. Then we construct honest weights $W_i$ (Specification~\ref{assumption:honest_weights}) and finally we apply the weighted BH procedure at level $\alpha \in (0,1)$ (Definition~\ref{defn:tau-wbh} with $\tau=1$) with p-values $P_i$ and weights $W_i$. We will show this procedure does not control the $\FDR$ at the nominal level. 

We observe four independent and uniform (null) p-values $P_1, P_2, P_3, P_4$. Our covariates take values $X_i = i$. We partition the hypotheses into the folds $\cb{1,2}$ and $\cb{3,4}$. The (adversarial) honest weighting scheme is as follows: If $\frac{\alpha}{2} \leq P_1 \leq \alpha$, assign $W_3=2, W_4=0$. Otherwise assign $W_3=0, W_4=2$. Similarly if $\frac{\alpha}{2} \leq P_3 \leq \alpha$, then assign $W_1=2, W_2=0$ and otherwise $W_1=0, W_2=2$. These weights are honest; note that $W_i \geq 0$ for all $i$, $\sum_{i=1}^4 W_i =4$. 

To study the $\FDR$ of this procedure we partition the sample space according to the four possibilities for the weight assignment. Also note that due to the weighting scheme in the end we will be applying unweighted Benjamini-Hochberg to two hypotheses at level $\alpha$. For notational convenience we will write $\text{BH}(P_i, P_j)$ for the event that BH applied to $P_i,P_j$ at level $\alpha$ rejects at least one of these two p-values.\\

\noindent\textbf{Case 1:} Here we have $W_2=W_4=2$ and $W_1=W_3=0$. Thus we are just doing unweighted Benjamini-Hochberg on the p-values $P_2$ and $P_4$. Noting that occurence of this case depends only on $P_1, P_3$, we get by independence:
\begin{equation*} 
\PP[\text{Case 1 occurs}, \text{BH}(P_2, P_4)] = \PP[\text{Case 1 occurs}]\PP[\text{BH}(P_2, P_4)] = \left(1-\frac{\alpha}{2}\right)^2 \alpha\,.
\end{equation*}
\noindent\textbf{Case 2:} Now consider $W_1=W_3=2$ and $W_2=W_4=0$. In this case, we know that both $\frac{\alpha}{2} \leq P_1 \leq \alpha$ and $\frac{\alpha}{2} \leq P_3 \leq \alpha$. These in turn imply that $\text{BH}(P_1,P_3)$ also holds (in fact BH rejects both hypotheses). Thus:
\begin{equation*} 
\PP[\text{Case 2 occurs}, \text{BH}(P_1, P_3)] = \PP[\text{Case 2 occurs}] = \left(\frac{\alpha}{2}\right)^2\,. 
\end{equation*}
\noindent\textbf{Case 3:} Now let $W_1=W_4=2$ and $W_2=W_3=0$. Then:
\begin{equation*}
\begin{aligned}
\PP\left[\text{Case 3 occurs}, \text{BH}(P_1, P_4)\right] &= \PP\left[ P_1 \not\in \left[\frac{\alpha}{2}, \alpha\right], \frac{\alpha}{2} \leq P_3 \leq \alpha, \text{BH}(P_1, P_4)\right] \\
&= \PP\left[\frac{\alpha}{2} \leq P_3 \leq \alpha \right] \PP\left[ P_1 \not\in \left[\frac{\alpha}{2}, \alpha\right], \text{BH}(P_1,P_4)\right] \\
&= \frac{\alpha}{2}\left[\frac{\alpha}{2} + \frac{\alpha}{2}(1-\alpha)\right]\,.
\end{aligned}
\end{equation*}
The latter is true since if $P_1 \not\in [\frac{\alpha}{2}, \alpha]$, the only way BH will reject is if $P_1 < \frac{\alpha}{2}$ or $P_4 \leq \frac{\alpha}{2}$. Hence the event on the RHS can be written as the disjoint union of $\{P_1 < \alpha/2\}$ and $\{P_4 \leq \alpha/2, P_1 > \alpha \}$.\\

\noindent\textbf{Case 4:} By symmetry with Case 3, this contributes the same probability.\\
\noindent Summing up all 4 cases, we see that 
\begin{equation*}
\FDR = \text{FWER} = \alpha + \frac{\alpha^2}{4}(1-\alpha) > \alpha
\end{equation*}
Hence $\FDR$ is not controlled at the nominal level $\alpha$.

%---------------------------------------------------------------------------------------------
\subsection{The IHW-Storey procedure under independence: Proof of Theorem~\ref{thm:IHW-Storey}}
\label{subsec:IHWc_Storey_proof}
%---------------------------------------------------------------------------------------------
\begin{proof}
Take $i \in I_{\ell} \cap \Hnull$ and define the leave-one-out null proportion estimator (compare to Equation~\eqref{eq:pi0_adaptive}):
\begin{equation*}
\label{eq:pi0_adaptive_loo}
\hat{\pi}_{0,I_{\ell}}^{-i} = \frac{ \displaystyle\max_{j \in I_{\ell}} W_j + \sum_{j \in I_{\ell} \setminus \{i\}} W_j \ind(P_j > \tau')}{|I_{\ell}|(1-\tau')}\,.
\end{equation*}
Now note that on the event $\{P_i \leq \tau\}$ (since $\tau' \geq \tau$) we have that:
\begin{equation}
\label{eq:same_pi0s}
\hat{\pi}_{0,I_{\ell}} = \hat{\pi}_{0,I_{\ell}}^{-i}\,.
\end{equation}
Next, define
\begin{equation*}
\widetilde{W}_i = \frac{W_i}{\hat{\pi}_{0,I_{\ell}}^{-i}}\,.
\end{equation*}
\eqref{eq:same_pi0s} implies that running the $\tau$-censored, weighed BH procedure (Definition~\ref{defn:tau-wbh}) with p-values $P_i$ and weights $W_i/\hat{\pi}_{0,I_{\ell}}$ (i.e., the procedure whose $\FDR$ control we seek to prove) will have identical rejections if we replace the weights by $\widetilde{W}_i$. Hence we turn to study the procedure with weights $\widetilde{W}_i$. Proceeding as in the leave-one-out argument of the proof of Theorem~\ref{thm:IHWc} we get
\begin{equation*}
H_i \text{ rejected} \Rightarrow P_i \leq \frac{\alpha \widetilde{W}_i k_i}{m} \land \tau \,.
\end{equation*}
In fact, since $\hat{\pi}_{0,I_{\ell}}^{-i}$ does not depend on $P_i$ (it depends on $\mathbf{P}_{i \mapsto 0})$, all arguments of the proof of Theorem~\ref{thm:IHWc} go through unchanged with \smash{$\widetilde{W}_i$} replacing $W_i$. The only step we need to pay attention to is the last line: it no longer holds that
\begin{equation*} 
\sum_{i=1}^m \widetilde{W}_i = m \text{ almost surely}\,.
\end{equation*}
Indeed we are hoping that this sum is greater than $m$ so that we can gain power by the null-proportion adaptivity. Instead, it suffices to argue that
\begin{equation*}
\sum_{i \in \Hnull} \EE\left[\widetilde{W}_i\right] \leq m\,. 
\end{equation*}
And hence it also suffices to prove that for each fold $\ell$ the following holds
\begin{equation*}
\sum_{i \in \Hnull \cap I_{\ell}} \EE\left[\widetilde{W}_i\right] \leq |I_{\ell}|\,.
\end{equation*}
To prove this, we first recall from Lemma~\ref{lem:Indep}(a') that
\begin{equation*}
(P_i)_{i \in \Hnull \cap I_{\ell}} \perp (W_i)_{i \in \Hnull \cap I_{\ell}}\,.
\end{equation*}
For notational convenience we write $\mathbf{W}_{ \Hnull \cap I_{\ell}}$ for $(W_i)_{i \in \Hnull \cap I_{\ell}}$. Then:
\begin{equation*}
\begin{aligned}
\EE\left[\widetilde{W}_i  \;\middle|\;  (W_i)_{i \in \Hnull \cap I_{\ell}}\right] &= \EE\left[\frac{W_i}{\hat{\pi}_{0,I_{\ell}}^{-i}}  \;\middle|\;  \mathbf{W}_{ \Hnull \cap I_{\ell}}\right] \\
&= W_i\EE\left[\frac{1}{\hat{\pi}_{0,I_{\ell}}^{-i}} \;\middle|\;  \mathbf{W}_{ \Hnull \cap I_{\ell}}\right] \\
&= W_i \EE\left[\frac{|I_{\ell}|(1-\tau')}{\underset{j \in I_{\ell}}{\max} \, W_j + \sum\limits_{j \in I_{\ell} \setminus \{i\}} W_j\ind(P_j > \tau')  } \;\middle|\;  \mathbf{W}_{ \Hnull \cap I_{\ell}} \right] \\
&\leq W_i\,|I_{\ell}|\,(1-\tau')\EE\left[\frac{1}{\underset{j \in \Hnull\cap I_{\ell}}{\max} \, W_j + \sum\limits_{j \in \Hnull\cap I_{\ell} \setminus \{i\}} W_j \ind(P_j > \tau')  } \;\middle|\;  \mathbf{W}_{ \Hnull \cap I_{\ell}} \right] \\
&\leq W_i\,|I_{\ell}|\,(1-\tau') \frac{1}{(1-\tau') \sum\limits_{j \in \Hnull\cap I_{\ell}} W_j} \\
&= \frac{W_i\,|I_{\ell}|}{\sum\limits_{j \in \Hnull\cap I_{\ell}} W_j}\,.
\end{aligned}
\end{equation*}
In the penultimate line we used the Inverse Binomial Lemma (Lemma 3 in \citet{ramdas2017unified}), noting that conditionally on $\mathbf{W}_{ \Hnull \cap I_{\ell}}$, the weights in folds $\ell$ may be treated as deterministic and by Lemma~\ref{lem:Indep}(b') the p-values $(P_i)_{i \in \Hnull \cap I_{\ell}}$ are jointly independent and super-uniform. We conclude our proof by iterated expectation and summing over $i \in \Hnull \cap I_{\ell}$.
\end{proof}

%--------------------------------------------------------------------------
\subsection{The IHW-BY procedure under dependence: Proof of Theorem~\ref{thm:ihw-by}}
\label{subsec:IHWdep_proof}
%--------------------------------------------------------------------------

\begin{proof}
We will equivalently prove that applying the weighted Benjamini-Hochberg procedure (without censoring, i.e., $\tau=1$) at level $\alpha$ controls the $\FDR$ at level $\alpha\sum_{k=1}^m \frac{1}{k}$.\\
For a probability measure $\nu$ on $\mathbb R^+$, we define the reshaping function $\tilde{\beta}: \mathbb R^+ \to \mathbb R^+$~\citep{blanchard2008two, ramdas2017unified}:
\begin{equation*}
\tilde{\beta}(r) = \int_{0}^r x d\nu(x)\,. 
\end{equation*}
Furthermore, let $\hat{k}$ be the number of rejections of the IHW-BH procedure applied at level $\alpha$. Then for arbitrary $c>0$, $i\in \Hnull$ and on the event $\{W_i >0 \}$:
\begin{equation}
\label{eq:by_proof_key_step}
\EE\left[\frac{\ind\p{P_i \leq \frac{\scriptstyle c \alpha W_i}{\scriptstyle m}\tilde{\beta}(\hat{k})}}{\hat{k}\lor 1} \;\middle|\; W_i \right] = 
\frac{c\alpha W_i}{m} \,\EE\left[\frac{\ind\p{P_i \leq \frac{\scriptstyle c\alpha W_i}{\scriptstyle m}\tilde{\beta}(\hat{k})}}{\frac{\scriptstyle c\alpha W_i}{\scriptstyle m}(\hat{k}\lor 1)} \;\middle|\; W_i \right]\leq 
\frac{c\alpha W_i}{m}\,.
\end{equation}
The inequality follows from Lemma 3.2.~(iii) in \citet{blanchard2008two} (also Lemma 1(c) in \citet{ramdas2017unified}), which we reproduce in a slightly modified form here for the reader's convenience:
\begin{lem}
\label{lem:superuniform}
Let $U$ a super-uniform random variable and $S >0$ another random variable, then for all fixed $t >0$:

\begin{equation*}
\EE\left[\frac{\ind(U \leq t\tilde{\beta}(S))}{tS} \right] \leq 1
\end{equation*}
\end{lem}
\noindent We recover~\eqref{eq:by_proof_key_step} by applying Lemma~\ref{lem:superuniform} conditionally on $W_i$ with $U=P_i$, $S=\hat{k}\lor1$ and $t=\frac{c\alpha W_i}{m}$. To do so, note that we may treat $\frac{c\alpha W_i}{m}$ as a constant conditionally on $W_i$ and that $P_i \mid W_i$ is super-uniform, since $P_i \perp W_i$ by Lemma~\ref{lem:Indep}(a) and $P_i$ is unconditionally super-uniform.

Inequality~\eqref{eq:by_proof_key_step} also holds true almost surely on the event $\{W_i = 0 \}$, as the distribution of $P_i \mid W_i$ cannot have a point mass at $0$, since this would contradict super-uniformity. Thus we also get unconditionally that
\begin{equation*}
\EE\left[\frac{\ind\p{P_i \leq \frac{\scriptstyle c\alpha W_i}{\scriptstyle m}\tilde{\beta}(\hat{k})}}{\hat{k}\lor 1}\right] \leq \EE\left[\frac{c\alpha W_i}{m}\right]\,.
\end{equation*}
Now, consider the special case in which we use the measure $\nu(x) = \frac{1}{\sum_{k=1}^m \frac{1}{k}}\sum_{k=1}^m \frac{1}{k} \delta_k(x)$, where $\delta_k$ is the point mass at $k$. Then the reshaping function takes the form $\tilde{\beta}(r) = \frac{r}{\sum_{k=1}^m \frac{1}{k}}$ for $r \in \mathbb N_{\geq 0}$.  Applying the above result with this $\tilde{\beta}$ and $c= \sum_{k=1}^m \frac{1}{k}$ we get
\begin{equation*}
\EE\left[\frac{\ind\p{P_i \leq \frac{\scriptstyle \alpha W_i\hat{k}}{\scriptstyle m}}}{\hat{k}\lor 1}\right] \leq \frac{\alpha\sum_{k=1}^m \frac{1}{k}}{m}\EE[W_i]\,
\end{equation*}
We conclude by using that $\sum_{i=1}^m W_i =m \text{ almost surely}$ as follows
\begin{equation*}
\begin{aligned}
\EE[\FDP] &= \sum_{i \in \Hnull}\EE\left[\frac{\ind\p{P_i \leq \frac{\scriptstyle \alpha W_i\hat{k}}{\scriptstyle m}}}{\hat{k}\lor 1}\right]  \\
&\leq \frac{\alpha\sum_{k=1}^m \frac{1}{k}}{m}\sum_{i \in \Hnull}\EE\left[ W_i\right]\\
&\leq \frac{\alpha\sum_{k=1}^m \frac{1}{k}}{m}\EE\left[\sum_{i=1}^m W_i\right]\\
&= \alpha\sum_{k=1}^m \frac{1}{k}\,.
\end{aligned}
\end{equation*}
Note that the above proof extends to applying the weighted BH procedure with arbitrary reshaping function $\tilde{\beta}$ as in \citet{blanchard2008two, ramdas2017unified}.

\end{proof}

%-------------------------------------------------------------------------------------------------------------------
\subsection{Counterexample to demonstrate that BY with $\tau$-censored data-driven weights does not control $\FDR$}
\label{subsec:counterexample_by_tau}
%-------------------------------------------------------------------------------------------------------------------

For our counterexample, we consider the following $\tau$-censored way of assigning data-driven weights: assign weight $W_i =0 $ to all hypotheses with $p$-value greater than $\tau$ and distribute the remaining weight equally across all hypotheses with p-value $\leq \tau$ in any given fold. This weighting procedure satisfies $\tau$-censoring (Specification~\ref{assumption:tau_stopped_weights}) as it only uses whether a p-value is below or above $\tau$; however it does not satisfy honesty (Specification~\ref{assumption:honest_weights}). Finally, we apply the weighted Benjamini-Yekutieli procedure with p-values $P_i$ and weights $W_i$.

\begin{proof}
The result for this counterexample depends on $m, \tau, \alpha$. We make the following simplifying assumptions on these: first, to avoid issues with rounding, we assume that $\tau \in \mathbb Q$ and $m$ is such that $m\cdot \tau \in \mathbb N$. Furthermore, we assume that $\alpha \leq \tau$.\footnote{ $\FDR$ control is also violated when $\alpha > \tau$: just replace $\tau$ by $\max\cb{\alpha,\tau}$ in the following arguments.}

Below, we will construct a joint distribution on $\allpairs$ such that Assumption~\ref{assumption:distrib_dep} holds with one fold, i.e., $K=1$ and $I_1 = [m]$. We discuss the case of two independent folds at the end of the proof.

First, we draw $X_i \simiid U[0,1]$ and independent of the p-values $(P_i)_{i \in [m]}$. The joint distribution of the p-values is constructed (details below) so that exactly $m\tau$ p-values are $\leq \tau$. This means that the $m\tau$ hypotheses with p-value $\leq \tau$ are assigned weights $m/(m \tau) = 1/\tau$ and so letting $\alpha_{BY} = \alpha\big/\sum_{j=1}^{m}\frac{1}{j}$, then weighted BY will reject at least $k$ hypotheses if:
$$
\begin{aligned}
& P_j \leq \frac{\alpha_{BY} \cdot k \cdot W_j}{m} \text{ for at least } k \text { indices }\subset \{1,\dotsc,m\} \\
\Longleftrightarrow\;\; & P_j \leq   \frac{\alpha_{BY} \cdot k \cdot \frac{1}{\tau}}{m} \land \tau \text{ for at least } k \text { indices }\subset \{1,\dotsc,m\}
\end{aligned}
$$
Next we define $q_k = \frac{\alpha_{BY} \cdot k}{m \tau}, k \geq 1$, $q_0=0$. Then weighted BY will make at least $k$ rejections (for $k \leq m\tau$) if:
\begin{equation}
\label{eq:by_rejections_rule}
 P_j \leq  q_k \text{ for at least } k \text { indices }\subset \{1,\dotsc,m\}
\end{equation}
It remains to provide a distribution on p-values $(P_1,\dotsc,P_m)$ such that Assumption~\ref{assumption:distrib_dep} is satisfied and such that~\eqref{eq:by_rejections_rule} with $k \geq 1$ occurs frequently enough so that $\FDR$ control is violated. To this end, we generate the $m$ p-values hierarchically\footnote{Our construction is a modification of an unpublished proof of the worst-case behavior of BH under dependence by Emmanuel Cand{\`e}s and Rina Foygel Barber. This proof has appeared in the STATS300C lecture notes of Emmanuel Cand{\`e}s, available at \url{https://statweb.stanford.edu/~candes/teaching/stats300c/}.} as follows:

\begin{enumerate}
\item We draw a set of indices $\mathcal{T} \subset \{1,\dotsc,m\}$ of cardinality $m \tau$ uniformly at random from $\{1,\dotsc,m\}$.
\item For $i \notin \mathcal{T}$, we draw $P_i \sim U[\tau, 1]$.
\item For $i \in \mathcal{T}$, we instead proceed as follows:
\begin{enumerate}
  \item We draw $\Kproof \in \{0,\dotsc, m\tau\}$ from the following distribution:
     $$ \PP[ \Kproof = k] = m \frac{q_k - q_{k-1}}{k} = \frac{\alpha_{BY}}{\tau k}, \; k=1,\dotsc, m\tau, \;\; \PP[\Kproof=0] = 1- \frac{\alpha_{BY}}{\tau}\sum_{j=1}^{m\tau} \frac{1}{j}$$
  \item We draw a set of indices $\mathcal{S} \subset \mathcal{T}$ of cardinality $\Kproof$ uniformly at random from  $\mathcal{T}$.
  \item For $i \in \mathcal{S}$, we draw $P_i \sim U[q_{\Kproof-1}, q_{\Kproof}]$.
  \item For $i \in \mathcal{T}\setminus\mathcal{S}$, we draw $P_i \sim U[\alpha_{BY},\tau]$. 
\end{enumerate}
\end{enumerate}
Let us note that when $\Kproof \geq 1$, then $\abs{\mathcal{S}}=\Kproof$ and so there will be $\Kproof$ p-values in the interval $U[q_{\Kproof-1}, q_{\Kproof}]$, and so by~\eqref{eq:by_rejections_rule} these p-values will be rejected leading to a $\FDP$ equal to $1$. The only situation in which we will make no rejections is on the event that $\Kproof=0$ and so $\FDP \geq \ind(\Kproof \geq 1)$. Thus:
\begin{equation}
\label{eq:by_fdr_loss}
\FDR \geq 1 - \PP[ \Kproof = 0]=  \frac{\alpha_{BY}}{\tau}\sum_{j=1}^{m\tau} \frac{1}{j} = \frac{\alpha}{\tau} \cdot \frac{\sum_{j=1}^{m\tau} \frac{1}{j}}{\sum_{j=1}^{m} \frac{1}{j}} \geq  \frac{\alpha}{\tau}\cdot \frac{\log(m\tau + 1)}{\log(m)+1}
\end{equation}
Note that for large enough $m$ this approaches $\alpha/\tau$ and so indeed, for $\tau <1$, $\FDR$ is not controlled.

There remains one step to conclude the proof: we need to check that the p-values generated above indeed are all (marginally) uniform. Fix an arbitrary $i \in \{1,\dotsc,m\}$. Note that conditionally on $\Kproof, \mathcal{S}, \mathcal{T}$, the distribution of the p-value $P_i$ is as follows:
$$
P_i \sim \left\{\begin{matrix}
 U[0, q_1] &\text {if } i \in \mathcal{S},\;\Kproof=1\\
 U[q_1, q_2] &\text {if } i \in \mathcal{S},\; \Kproof=2\\
\vdots & \\
 U[q_{m\tau -1}, \alpha_{BY}] &\text {if } i \in \mathcal{S},\; \Kproof=m\tau\\
 U[\alpha_{BY}, \tau]  &\text {if } i \in \mathcal{T}\setminus \mathcal{S} \\ 
 U[\tau, 1] &\text {if } i \notin \mathcal{T}
\end{matrix}\right.
$$
Let us compute the probabilities of the events above:
$$ \PP[ i \in \mathcal{S},\; \Kproof=k] = \PP[  i \in \mathcal{S} \cond \Kproof=k] \; \PP[\Kproof=k] = \frac{k}{m} \cdot m \frac{q_k - q_{k-1}}{k}=  q_k - q_{k-1}$$
$$ \PP[ i \in \mathcal{T}\setminus \mathcal{S}] = \PP[ i \in \mathcal{T}] - \PP[i \in \mathcal{S}] = \tau - \sum_{j=1}^{m \tau}(q_j - q_{j-1}) = \tau - q_{m \tau} = \tau - \alpha_{BY}$$
$$ \PP[ i \notin \mathcal{T}] = 1 -\tau $$
This means that:
$$
P_i \sim \left\{\begin{matrix}
 U[0, q_1] &\text { with probability } q_1\\
 U[q_1, q_2] &\text { with probability } q_2 - q_1\\
\vdots & \\
 U[q_{m\tau -1}, \alpha_{BY}] &\text{ with probability } \alpha_{BY}-q_{m\tau -1}\\
 U[\alpha_{BY}, \tau]  &\text{ with probability } \tau - \alpha_{BY}\\ 
 U[\tau, 1]  &\text{ with probability } 1-\tau
\end{matrix}\right.
$$
This is precisely the uniform distribution, i.e., $P_i \sim U[0,1]$.

Let us finally conclude by discussing how to extend this construction to the case of two independent folds. Let $m=2m'$ for $m' \in \mathbb N$ and assume that $m'\tau \in \mathbb N$. Let us take the two folds to be $I_1 = [m']$ and $I_2 = [m]\setminus [m']$. We may apply the construction above independently to each fold. Now let $A_{\ell}$, $\ell \in \cb{1,2}$ be the event that BY rejects at least one hypothesis in fold $\ell$, even after setting the p-values in the other fold to $1$. Then repeating the arguments leading up to~\eqref{eq:by_fdr_loss}, we find that $\PP[A_{\ell}] \geq \alpha'/2$, where $\alpha':=\alpha/\tau\cdot \log(m'\tau + 1)/(\log(2m')+1)$. Since $\FDR \geq \PP[A_1 \cup A_2]$ and the events $A_1$ and $A_2$ are independent, we find that $\FDR \geq \alpha'/(\alpha'+1)$. This is strictly larger than $\alpha$, for example when $\tau < 1$, $m'$ is large and $\alpha$ is small.

\end{proof}
%--------------------------------------------------
\section{Proofs for IHW-BH asymptotics}
\label{sec:asymp_proof}
%--------------------------------------------------

For our asymptotics, we make the following regularity assumption:
\begin{assumption}[Regularity of conditional two-groups model]
\label{assumption:conditional_twogroups_asymptotics}
The conditional two-groups model~\eqref{eq:conditional_twogroups} satisfies:
\begin{enumerate}[label=(\alph*)]
\item $F_{\text{alt}}(t \mid X_i=x)$ is $L(x)$-Lipschitz continuous in $t$ for all $x \in \mathcal{X}$, i.e., 
$$\abs{F_{\text{alt}}(t \mid X_i=x) - F_{\text{alt}}(t' \mid X_i=x)} \leq L(x)\abs{t-t'} \text{ for all } t,t' \in [0,1],x \in \mathcal{X}$$
$L(\cdot)$ satisfies $\int{L^2(x)d\PP^X(x)} < \infty$ and furthermore $F_{\text{alt}}(0 \mid X_i=x) =0$ for all $x$.
\item $F_{\text{alt}}(t \mid X_i=x)$ is strictly concave in $t$ for all $x$.
\item There exists $t' \in (0,1]$ such that $\frac{t'}{F(t' \mid X_i=x)} \leq \alpha' $ for an $\alpha' < \alpha$ and for all $x \in \mathcal{X}$.
\end{enumerate}
\end{assumption}
Part~(a) is a technical assumption restricting the smoothness of $F_{\text{alt}}(\cdot \mid X_i=x)$; it allows for the smoothness to vary as $x \in \mathcal{X}$ varies. Part (b) is a common assumption in multiple testing; see also the discussion and references in Section~\ref{sec:powerful_weighting_rules}. The assumption (in the setting without covariates) appears for example in Lemma 1 and Theorem 2 of~\citet{genovese2006false}. Part (c) is also an assumption made for $\FDR$ asymptotics without covariates (e.g., it appears in Theorem 4 of~\citet*{storey2004strong}). It is, however, less innocuous than Parts (a,b); for example it excludes the global null and the case $\pi_0(x)=1$.

\newcommand{\wtfun}{\mathscr{W}}
\newcommand{\op}{o_{\PP}(1)}
\newcommand{\Vbh}{\widehat{V}^{\text{BH}}}
\newcommand{\Fbh}{F_0^{\text{BH}}}

\newcommand{\wtfunloo}{\hat{\wtfun}^{([m]\setminus I)}}
\newcommand{\wtfunlooell}{\hat{\wtfun}^{([m]\setminus I_{\ell})}}
\newcommand{\wtfunm}{\hat{\wtfun}^{([m])}}

\newcommand{\hatFDPihw}{\widehat{\FDP}^{\text{IHW}}}
\newcommand{\tihw}{\hat{t}^{\text{IHW}}}

\newcommand{\TotalRIHW}{R^{\text{IHW}}}
\newcommand{\TotalVIHW}{V^{\text{IHW}}}
\newcommand{\FDPihw}{\FDP^{\text{IHW}}}
\newcommand{\tnaive}{\hat{t}^{\text{Naive}}}
\newcommand{\FDPnaive}{\FDP^{\text{Naive}}}
\newcommand{\hatFDPnaive}{\widehat{\FDP}^{\text{Naive}}}

\newcommand{\tstar}{t^*(\wtfun^*)}

\paragraph{Some remarks on notation:} In this section we use a different typeface for the weight function, i.e., we write $\wtfun: \mathcal{X} \to \RR_{\geq 0}$ and $\hat{\wtfun}^{(I)}$ for the weight function learned based on data $((P_i, X_i))_{i \in I}$. This ensures that the notation is unambiguous and not conflicting with the notation used in \supplementname~\ref{sec:finite_sample_proofs} for finite-sample results. We also use the notation $a_m = o(1)$ for a deterministic sequence $a_m$ satisfying $a_m \to 0, \text{ as } m \to \infty$ and $Z_m = \op$ for a sequence of random variables $Z_m$ that converge to $0$ in probability as $m \to \infty$.

\subsection{Proof of Proposition~\ref{prop:asymp}(a)}

\begin{proof}
We first make a few assumptions on the data-generating mechanism (while making sure that Assumption~\ref{assumption:conditional_twogroups_asymptotics} still holds): We first assume that $\mathcal{X}, \mathbb P^{X}$ are such that $X_1, \dotsc, X_n$ are all unequal with probability $1$; this is true for example when $\mathbb P^{X}$ is absolutely continuous w.r.t. the Lebesgue measure on $\RR^p$. Next we assume that for $\pi_1(x) = 1 - \pi_0(x)$ it holds that $\EEs{\pi_1(X_i)} < \delta$ for some $\delta >0$; i.e., there are not too many alternative hypotheses.  Finally we assume that we run weighted BH at $\alpha \in (0, 1/2)$.

Our application of naive weighted BH is as follows: We let $k_m = \floor{\alpha m / 2}$ and $\mathcal{J}_m$ the index set of $k_m$ hypotheses with smallest p-values. We will assign weight $m/k_m$ to these and all other hypotheses will receive weight $0$. This is equivalent to applying BH directly to the $k_m$ smallest p-values (while ignoring their selection).

Formally, in terms of Specification~\ref{specif:wt_scheme}, the weighting function takes the form:
$$ \hat{\wtfun}^{([m])}(x) =  \ind\p{ x \notin \cb{X_j, j \in [m]}} \; + \; \frac{m}{k_m} \ind\p{ x \in \cb{X_j, j \in \mathcal{J}_m}} ; $$
This satisfies the conditions of Specification~\ref{specif:wt_scheme}: \smash{$\int \hat{\wtfun}^{([m])}(x)d\PP^X(x)=1$} almost surely for all $m$ and second, $\sup_{x \in \mathcal{X}} \hat{\wtfun}^{([m])}(x) = m/k_m = m/\floor{\alpha m / 2} \leq 4/\alpha$ as soon as $\alpha m \geq 2$, which is stronger than the requirement on the growth of \smash{$\int \hat{\wtfun}^{([m])}(x)^2d\PP^X(x)$} in~\eqref{eq:weights_technical_condition} (this is a formal verification; condition~\eqref{eq:weights_technical_condition} pertains to the out-of-sample behavior of the weighting function with respect to a fresh draw $X_i \sim \PP^X$).

Writing $P_{(1)} \leq P_{(2)} \leq \dotsc \leq P_{(m)}$ for the order statistics of $P_1,\dotsc,P_m$,
consider the events $A_m = \cb{ P_{(k_m)} \leq \alpha}$ and $B_m = \cb{ \sum_{i=1}^m H_i \leq 1.1 \cdot \delta \cdot m}$.  

On the event $A_m$, weighted BH will reject all hypotheses in $\mathcal{J}_m$, since by definition of $A_m$ it holds that $P_{(k_m)} \leq \alpha = (k_m \cdot \alpha/m)\cdot(m/k_m) = (k_m \cdot \alpha/m)\cdot W_{(k_m)}$. On the other hand, on the event $B_m$, there will be at least $k_m - 1.1 \cdot \delta \cdot m$ false rejections (i.e., all rejections minus an upper bound on the number of alternative hypotheses). Thus on $A_m \cap B_m$ and for large enough $m$ (we slightly enlarge $2.2 = 2 \cdot 1.1$ to $2.3$ to account for rounding in the definition of $k_m$):

$$\FDP_m \geq \frac{k_m - 1.1 \cdot \delta \cdot m}{k_m} \geq 1 - \frac{2.3  \delta }{\alpha} $$
We will next argue that $\PPs{A_m}, \PPs{B_m} \to 1$ as $m \to \infty$ and thus:

$$\liminf_{m \to \infty} \FDR_m \geq 1 - \frac{2.3  \delta }{\alpha}$$
The latter will in general be $> \alpha$ for small enough $\delta$, so that naive weighted BH does not control $\FDR$.

Let us prove the claims for $A_m$ and $B_m$. For $B_m$, the result follows by noting that $\sum_{i=1}^m H_i \sim \text{Binomial}(m, \EE{\pi_1(X_i)})$, as well as an application of Chernoff's bound. For $A_m$, we note that by Assumption~\ref{assumption:conditional_twogroups_asymptotics}(b), it follows that $P_{(k_m)}$ is stochastically smaller than $\tilde{P}_{(k_m)}$, defined as the $k_m$-th smallest order statistic of a sample of $m$ i.i.d. uniform random variables $\tilde{P}_1, \dotsc, \tilde{P}_m$. Note that $\tilde{P}_{(k_m)}$ is distributed as $\text{Beta}(k_m, m+1-k_m)$ which has expectation $k_m/(m+1) \leq \frac{\alpha}{2}$. Hence:

$$ \PPs{A_m} \geq \PPs{\text{Beta}(k_m, m+1-k_m) \leq \alpha} \to 1 \text{ as } m \to \infty$$
The last convergence follows from concentration of a Beta random variable (say, by an application of Chebyshev's inequality.)

\end{proof}

\subsection{Proof of Proposition~\ref{prop:asymp}(b)}

\begin{proof}

We first give a sketch of the proof:

\begin{enumerate}
\item \textbf{Analysis for a single fold and a deterministic weighting function:} This serves as a warm-up. The analysis is very similar to asymptotics e.g., in~\citet{storey2004strong}, adapted to the setting with covariates and a weighting function.
\item \textbf{Analysis for a single fold with data-driven weighting function learned out-of-fold}: Here we refine the analysis from Step 1 to account for the data-driven nature of the weighting function. The fundamental nature of the arguments however is the same as in Step 1.
\item \textbf{Aggregating results across folds:} We give an equivalent formulation of the IHW-BH rejection rule in terms of empirical processes. Then, by combining results shown in Step 2, we demonstrate FDR control. 
\end{enumerate}

\paragraph{Single fold, deterministic weighting function:}

We first study a single fold, say $I = I_{\ell}$ (that grows with $m$), and a deterministic weighting function with the following properties:
\begin{equation}
\label{eq:fixed_wtfun}
\wtfun: \mathcal{X} \to \RR_{\geq 0},\; \int \wtfun(x) d\PP^X(x) = 1,\; \int \wtfun(x)^2 d\PP^X(x) \leq \Gamma < \infty
\end{equation}
We introduce notation for processes indexed by a threshold $t \in [0,1]$, the weighting function $\wtfun$ and the set $I \subset [m]$ indexing the hypotheses in the single fold under study.
\begin{equation*}
\begin{aligned}
&R(t, \wtfun; I) = \sum_{i \in I}  \ind(P_i \leq t \wtfun(X_i))\\
&V(t, \wtfun; I) = \sum_{i \in I}  \ind\p{H_i=0} \ind(P_i \leq t \wtfun(X_i))\\
&\Vbh(t, \wtfun; I) = t\cdot\sum_{i \in I} \wtfun(X_i)\\
&F(t, \wtfun) = \PPs{ P_i \leq t \wtfun(X_i)}\\
&F_0(t, \wtfun) =  \PPs{ P_i \leq t \wtfun(X_i); H_i = 0}\\
&\Fbh(t, \wtfun) = t\cdot\EEs{\wtfun(X_i)}
\end{aligned}
\end{equation*}
The goal will be to relate the empirical processes to their population counterparts through uniform (in $t$) laws of large numbers. We require one more definition to account for normalization of weights so that $\sum_{i \in I} W_i = \abs{I}$
\begin{equation}
\label{eq:hatc}
 \hat{c}_{I,\wtfun} = \abs{I}/\sum_{i \in I} \wtfun(X_i)
\end{equation}
Next pick a deterministic sequence $0 < \varepsilon_m = o(1) \text{ as } m\to\infty$ such that $\PPs{\abs{\hat{c}_{I,\wtfun} - 1} > \varepsilon_{m}} = o(1)$; such a sequence exists by the law of large numbers. Then for a $o_{\PP}(1)$ term that is uniform in $t \in [0,1]$, it holds that:
\begin{equation}
\label{eq:gliv_cantelli_for_R}
\begin{aligned}
R(\hat{c}_{I,\wtfun}  \cdot t, \wtfun; I)\ind\p{ \hat{c}_{I,\wtfun} < 1 + \varepsilon_{m}}/\abs{I} &\stackrel{(i)}{\leq} R((1+\varepsilon_m)t, \wtfun; I)/\abs{I}  \\
&\stackrel{(ii)}{=}  F((1+\varepsilon_m)t, \wtfun) + \op \\
&\stackrel{(iii)}{\leq} F(t, \wtfun) + \op + \Gamma^{1/2}(1 + \int L(x)^2 d\PP^X(x))^{1/2}\varepsilon_m
\end{aligned}
\end{equation}
$(i)$ follows by monotonicity of $R(t, \wtfun; I)$ in $t$. $(ii)$ follows from the Glivenko-Cantelli theorem applied to the i.i.d. $P_i/\wtfun(X_i)$\footnote{We set the above to $\infty$ if $\wtfun(X_i)=0$.}. $(iii)$ follows from Assumption~\ref{assumption:conditional_twogroups_asymptotics}(a), as follows: first note that $F(t \mid X_i=x)$ must be $\max\cb{1, L(x)}$ Lipschitz in $t$ as it is a convex combination of a $L(x)$-Lipschitz function and a $1$-Lipschitz function (the identity). Next
$$
\begin{aligned}
\abs{F((1+\varepsilon_m)t, \wtfun) -  F(t, \wtfun)} &= \abs{\EEs{ F((1+\varepsilon_m)t\cdot \wtfun(X_i) \mid X_i) - F(t\cdot \wtfun(X_i) \mid X_i)}} \\
&\leq \EEs{\abs{F((1+\varepsilon_m)t\cdot \wtfun(X_i) \mid X_i) - F(t\cdot \wtfun(X_i) \mid X_i)}} \\
&\leq \EEs{ \max\cb{1,L(X_i)}\varepsilon_m t \wtfun(X_i)} \\
&\leq  \varepsilon_m t \EEs{\max\cb{1,L^2(X_i)}}^{1/2}\EEs{\wtfun^2(X_i)}^{1/2} \\
& \leq \Gamma^{1/2} (1 + \int L(x)^2 d\PP^X(x))^{1/2} \varepsilon_m
\end{aligned}
$$
Applying the same argument in the reverse direction we also get for the same (uniform in $t$) $\op$ term:
$$R(\hat{c}_{I,\wtfun}  \cdot t, \wtfun; I)\ind\p{ \hat{c}_{I,\wtfun} > 1 - \varepsilon_{m}}/\abs{I}  \geq  F(t, \wtfun) - o_{\PP}(1) - \Gamma^{1/2}(1 + \int L(x)^2 d\PP^X(x))^{1/2}\varepsilon_m$$
Combining the two results, noting that $R(t, \hat{c}_\wtfun \cdot \wtfun; I) =  R(\hat{c}_\wtfun  \cdot t, \wtfun; I)$ and by choice of $\varepsilon_m$ we conclude that:
\begin{equation}
 \label{eq:r_glivenko_cantelli}
 \sup_{t \in [0,1]} \abs{ R(t, \hat{c}_{I,\wtfun}\cdot\wtfun; I)/\abs{I} - F(t, \wtfun)} = \op
\end{equation}
We may analogously prove that:
\begin{equation}
 \label{eq:r_glivenko_cantelli_part2}
\begin{aligned}
&\sup_{t \in [0,1]} \abs{ V(t, \hat{c}_{I,\wtfun}\cdot\wtfun; I)/\abs{I} - F_0(t, \wtfun)} = \op
\end{aligned}
\end{equation}
Also note that:
\begin{equation}
\label{eq:r_glivenko_cantelli_part2}
\Vbh(t, \hat{c}_{I,\wtfun}\cdot\wtfun; I)/\abs{I} = \Fbh(t, \wtfun)  = t \; \text{ for all }\; t
\end{equation}
It also deterministically holds that $F_0(t, \wtfun) \leq \Fbh(t, \wtfun)$ for all $t,\wtfun$ and so
\begin{equation}
\sup_{t \in [0,1]}\p{F_0(t, \wtfun) - \Vbh(t, \hat{c}_{I,\wtfun}\cdot\wtfun; I)/\abs{I}} \leq \op
\end{equation}

\paragraph{Single fold, data-driven weighting function:} Above we worked with a deterministic weighting function $\wtfun$. However, for IHW we use the weighting function $\wtfunloo$ learned out-of-fold. It turns out that the conclusions hold verbatim, i.e.,
\begin{equation}
 \begin{aligned}
 \label{eq:r_glivenko_cantelli_loo}
&\sup_{t \in [0,1]} \abs{ R(t, \hat{c}_{I,\wtfunloo}\cdot\wtfunloo; I)/\abs{I} - F(t, \wtfunloo)} = \op \\
&\sup_{t \in [0,1]} \abs{ V(t, \hat{c}_{I,\wtfunloo}\cdot\wtfunloo; I)/\abs{I} - F_0(t, \wtfunloo)} = \op \\
&\; \Vbh(t, \hat{c}_{I,\wtfunloo}\cdot\wtfunloo; I)/\abs{I} = \Fbh(t, \wtfunloo) = t \\
&\sup_{t \in [0,1]}\p{F_0(t, \wtfunloo) - \Vbh(t, \hat{c}_{I,\wtfunloo}\cdot\wtfunloo; I)/\abs{I}} \leq \op
 \end{aligned}
\end{equation}
To adapt the proof for deterministic $\wtfun$ to a proof for data-driven $\wtfunloo$ (where $\wtfunloo$ depends on data outside of fold $I=I_{\ell}$, cf. Specification~\ref{specif:wt_scheme}) we make the following observations:
\begin{enumerate}
    \item We conduct the analysis conditionally on data in the other folds $\mathcal{D}_{[m]\setminus I} = ((P_i,X_i,H_i))_{i \in [m]\setminus I}$. For example, to show the first result in~\eqref{eq:r_glivenko_cantelli_loo} it suffices to show (see arguments below) that for a sequence $\eta_m \to 0$:
    \begin{equation}
    \label{eq:conditional_result}
    \PPs{ \sup_{t \in [0,1]} \abs{ R(t, \hat{c}_{I,\wtfunloo}\cdot\wtfunloo; I)/\abs{I} - F(t, \wtfunloo)} > \eta_m  \;\middle|\; \mathcal{D}_{[m]\setminus I}} = \op
    \end{equation}
    Such a conditional convergence statement also implies unconditional convergence (cf. Lemma 6.1. in~\citet{chernozhukov2017double}), i.e.,
    $$\PPs{ \sup_{t \in [0,1]} \abs{ R(t, \hat{c}_{I,\wtfunloo}\cdot\wtfunloo; I)/\abs{I} - F(t, \wtfunloo)} > \eta_m} = o(1)$$
    The first result in~\eqref{eq:r_glivenko_cantelli_loo} then follows.
    \item It can be assumed without loss of generality that $\int \wtfunloo(x) d\PP^X(x) = 1$ for all $m$; otherwise we may redefine the weight function as $\wtfunloo / \int \wtfunloo(x) d\PP^X(x)$. This is only a formal modification; the IHW-BH procedure applied remains the same, as the weights will subsequently be rescaled to sum to $\abs{I}$ in fold $I$ (this is captured here by the multiplication with $\hat{c}_{I,\wtfunloo}$). 
    \item To establish~\eqref{eq:conditional_result}, the argument used for a deterministic weighting function (e.g., in~\eqref{eq:gliv_cantelli_for_R}) applies as long as we pay attention to controlling the two probabilistically negligible terms. In particular, we need to check that for (deterministic sequences) $\eta_m', \eta_m'' = o(1)$ that
    $$ \PPs{\abs{\hat{c}_{I,\wtfunloo} - 1} > \eta_{m}' \;\middle|\; \mathcal{D}_{[m]\setminus I}} = \op$$
    and
    $$ \PPs{ \sup_{t \in [0,1]} \abs{ R(t, \wtfunloo; I)/\abs{I}  -  F(t; \wtfunloo)} > \eta_{m}''  \;\middle|\; \mathcal{D}_{[m]\setminus I}} = \op$$
    In the deterministic case, the corresponding results were a consequence of the law of large numbers, respectively the Glivenko-Cantelli theorem. In the conditional case we may establish these results directly. For the first one we note that  by Chebyshev's inequality (conditionally on $\mathcal{D}_{[m]\setminus I}$) it holds almost surely for any $\delta > 0$ that
    $$\PPs{\abs{\hat{c}_{I,\wtfunloo}^{-1} - 1} > \delta \;\middle|\; \mathcal{D}_{[m]\setminus I}} \leq \frac{\int \wtfunloo(x)^2 d\PP^X(x)}{\delta^2 \abs{I}} \leq \frac{\Gamma}{\delta^2 \abs{I}}$$
    The conclusion follows. For the second result, we may replace the Glivenko-Cantelli theorem by an application of the Dvoretzky–Kiefer–Wolfowitz (DKW) inequality conditionally on $\mathcal{D}_{[m]\setminus I}$.
\end{enumerate}

\paragraph{Aggregating results across folds:}

Let us introduce some additional notation. 
\begin{equation}
\label{eq:fdpihw}
\hatFDPihw(t) = \frac{\sum_{\ell=1}^K \Vbh(t, \hat{c}_{I_{\ell},\wtfunlooell} \cdot \wtfunlooell; I_{\ell}) / m}{\max\cb{1, \sum_{\ell=1}^K R(t, \hat{c}_{I_{\ell},\wtfunlooell} \cdot \wtfunlooell; I_{\ell})} / m}
\end{equation}
\begin{equation}
\label{eq:tihw}
\tihw = \sup\cb{t \in [0,1] \mid \hatFDPihw(t) \leq \alpha}
\end{equation}
This implies (the denominator of $\hatFDPihw(t)$ is right-continuous and piecewise constant in $t$ with jumps, while the numerator is continuous) that
\begin{equation}
\label{eq:hatfdpihw}
\hatFDPihw(\tihw) \leq \alpha \text{ almost surely}
\end{equation}
The quantities allow us to express IHW-BH from an empirical process viewpoint (cf.~\citet{storey2004strong})
\begin{equation}
\label{eq:IHWbh_empirical_process}
\text{Reject } i \in I_{\ell}\;\;\Longleftrightarrow  \;\; P_i \leq \tihw \cdot  \hat{c}_{I_{\ell},\wtfunlooell} \cdot \wtfunlooell(X_i)
\end{equation}
Henceforth we make the additional assumption that $\alpha'$ in Assumption~\ref{assumption:conditional_twogroups_asymptotics}(c) further satisfies $\alpha' < \alpha/2$; this simplifies the step below but the proof goes through also for $\alpha'<\alpha$. With this simplification, we next argue that, for $t'' = t'/(4\Gamma)$ with $t'$ defined in Assumption~\ref{assumption:conditional_twogroups_asymptotics}(c)
\begin{equation}
\label{eq:tihw_not_small}
\PPs{\tihw < t''} = o(1) \text{ as } m \to \infty 
\end{equation}
Note that Assumption~\ref{assumption:conditional_twogroups_asymptotics} implies that $F(t| X_i = x) \geq t$ for all $t \in (0,1)$. Fixing a weighting function $\wtfun$ as in~\eqref{eq:fixed_wtfun}
$$
\begin{aligned}
F(t, \wtfun) &=  \int  F(t \wtfun(x) \mid X_i=x) d\PP^X(x) \\
              &\geq  \int  \ind\p{\wtfun(x) > 1/2} F(t \wtfun(x) \mid X_i=x) d\PP^X(x)\\
              &\geq  \int  \ind\p{\wtfun(x) > 1/2} F(t/2 \mid X_i=x) d\PP^X(x)\\
              &\geq t/2\int \ind\p{\wtfun(x) > 1/2} d\PP^X(x)\\
              &\geq   t/2 \cdot (1/4) \frac{(\int \wtfun(x) d\PP^X(x))^2}{\int \wtfun(x)^2 d\PP^X(x)} \\
              &\geq  t/(8 \cdot \Gamma) 
\end{aligned}
$$
In the penultimate step we used the Paley–Zygmund inequality. This lower bound holds uniformly over weighting functions satisfying~\eqref{eq:fixed_wtfun}. In the current setting, this implies that for a constant $c>0$
$$ \sum_{\ell=1}^K \abs{I_{\ell}}/m \cdot F(t''', \wtfunlooell) \;> \; c  \; \text{ for all } t''' \geq t'' \text{  almost surely for all }m$$
In conjunction with \eqref{eq:r_glivenko_cantelli_loo} this yields (we need the preceding claim to make sure the denominators in the expression below do not vanish)
\begin{equation}
\label{eq:uniform_hatfdp}
\sup_{t \in [t'', 1]}\abs{\hatFDPihw(t) - \frac{\sum_{\ell=1}^K \abs{I_{\ell}}/m \cdot \Fbh(t, \wtfunlooell)}{ \sum_{\ell=1}^K \abs{I_{\ell}}/m \cdot F(t, \wtfunlooell) }} = \op
\end{equation}
Fixing again a $\wtfun$ as in~\eqref{eq:fixed_wtfun}, we find using Cauchy-Schwarz and Markov's inequality, that
\begin{equation}
\label{eq:truncated_weight_mean}
\begin{aligned}
\int \wtfun(x) \ind\p{\wtfun(x) \leq 4\Gamma}d\PP^X(x) &= 1 - \int \wtfun(x) \ind\p{\wtfun(x) > 4\Gamma}d\PP^X(x) \\
&\geq 1-  \big(\int \wtfun(x)^2 d\PP^X(x)\big)^{1/2} \big(\int \ind\p{\wtfun(x) > 4\Gamma}d\PP^X(x)\big)^{1/2}\\
&\geq 1 - \Gamma^{1/2} 1/(4\Gamma)^{1/2} \\
& = \frac{1}{2}
\end{aligned}
\end{equation}
Next, we find that:
\begin{equation}
\label{eq:local_to_integrated_marginal_fdr}
\begin{aligned}
\Fbh(t'', \wtfun) &\leq \int t'' \wtfun(x) d\PP^X(x)  \\
& \stackrel{(i)}{\leq}  2\int t''\wtfun(x) \ind\p{\wtfun(x) \leq 4\Gamma} d\PP^X(x)  \\
& = 2\int (t'/4\Gamma)\wtfun(x)\ind\p{\wtfun(x) \leq 4\Gamma}d\PP^X(x) \\
& \stackrel{(ii)}{\leq} 2\alpha' \int \wtfun(x)/(4\Gamma) \cdot F(t' \mid X_i=x) \ind\p{\wtfun(x) \leq 4\Gamma}d\PP^X(x) \\
& \stackrel{(iii)}{\leq} 2\alpha' \int  F( \wtfun(x)/(4\Gamma) \cdot t' \mid X_i=x)\ind\p{\wtfun(x) \leq 4\Gamma}d\PP^X(x) \\ 
& \leq 2\alpha'  \int  F( t''\cdot \wtfun(x) \mid X_i=x) d\PP^X(x) \\
& = 2\alpha' F(t'', \wtfun)
\end{aligned}
\end{equation}
Step $(i)$ follows from~\eqref{eq:truncated_weight_mean}, step $(ii)$ follows by the definition of $\alpha'$ in Assumption~\ref{assumption:conditional_twogroups_asymptotics}(c) and $(iii)$ from concavity of $F(\cdot \mid x)$ and the fact that $\wtfun(x)/(4\Gamma)\leq 1$ on the stated event. Next, rearranging~\eqref{eq:local_to_integrated_marginal_fdr}
$$ \frac{\Fbh(t'', \wtfun)}{F(t'', \wtfun)} \leq 2\alpha'$$
Since this holds for an arbitrary weighting function~\eqref{eq:fixed_wtfun}, we also get that
$$\frac{\sum_{\ell=1}^K \abs{I_{\ell}}/m \cdot \Fbh(t'', \wtfunlooell)}{ \sum_{\ell=1}^K \abs{I_{\ell}}/m \cdot F(t'', \wtfunlooell)} \leq 2\alpha' < \alpha$$
And so along with~\eqref{eq:uniform_hatfdp} we see that~\eqref{eq:tihw_not_small} holds:
$$ \PPs{ \tihw \geq t''} \geq \PPs{ \hatFDPihw(t'') \leq \alpha} = 1 - o(1) \text{ as } m\to\infty$$
We are almost ready to prove $\FDR$ control. By~\eqref{eq:IHWbh_empirical_process}, we see that the rejections of IHW-BH are precisely equal to:
$$\TotalRIHW = \sum_{\ell=1}^K R(\tihw, \hat{c}_{I_{\ell},\wtfunlooell} \cdot \wtfunlooell; I_{\ell})$$
The false rejections are equal to:
$$ \TotalVIHW = \sum_{\ell=1}^K V(\tihw, \hat{c}_{I_{\ell},\wtfunlooell} \cdot \wtfunlooell; I_{\ell})$$
So:
$$
\begin{aligned}
\FDPihw \cdot \ind(\tihw \geq t'') &= \frac{\TotalVIHW}{\max\cb{1, \TotalRIHW}}\cdot \ind(\tihw \geq t'') \\
&= \frac{ \sum_{\ell=1}^K V(\tihw, \hat{c}_{I_{\ell},\wtfunlooell} \cdot \wtfunlooell; I_{\ell})}{\max\cb{1,  \sum_{\ell=1}^K R(\tihw, \hat{c}_{I_{\ell},\wtfunlooell} \cdot \wtfunlooell; I_{\ell})}}\cdot \ind(\tihw \geq t'') \\
&=   \frac{\sum_{\ell=1}^K \abs{I_{\ell}} \cdot F_0(\tihw, \wtfunlooell)}{\max\cb{1,  \sum_{\ell=1}^K R(\tihw, \hat{c}_{I_{\ell},\wtfunlooell} \cdot \wtfunlooell; I_{\ell})}}\cdot \ind(\tihw \geq t'') \; + \;\op \\
&\leq \frac{\sum_{\ell=1}^K \abs{I_{\ell}} \cdot \Fbh(\tihw, \wtfunlooell)}{\max\cb{1,  \sum_{\ell=1}^K R(\tihw, \hat{c}_{I_{\ell},\wtfunlooell} \cdot \wtfunlooell; I_{\ell})}}\cdot \ind(\tihw \geq t'') \; + \;\op \\
&\leq \frac{\sum_{\ell=1}^K \Vbh(\tihw, \hat{c}_{I_{\ell},\wtfunlooell} \cdot \wtfunlooell; I_{\ell})}{\max\cb{1,  \sum_{\ell=1}^K R(\tihw, \hat{c}_{I_{\ell},\wtfunlooell} \cdot \wtfunlooell; I_{\ell})}}\cdot \ind(\tihw \geq t'')  \; + \;\op\\
& =\hatFDPihw(\tihw) \cdot \ind(\tihw \geq t'') \; + \;\op \\
&\leq \alpha  \cdot 1\; + \;\op
\end{aligned}
$$
In the last step we used~\eqref{eq:hatfdpihw}. We carry along the constraint $\ind(\tihw \geq t'')$ to emphasize that the denominator divided by $m$ will remain bounded from below with probability converging to $1$. We conclude with the dominated convergence theorem ($\text{FDR}^{\text{IHW}} = \EEs{\FDPihw}$ and $\FDPihw \in [0,1]$) that
$$ \limsup_{m \to \infty} \text{FDR}^{\text{IHW}} \leq \alpha $$

\end{proof}

\subsection{Proof of Proposition~\ref{prop:asymp}(c)}

\begin{proof}
Let us introduce the asymptotic threshold of both procedures;
\begin{equation}
\label{eq:tstar}
\tstar =  \sup\cb{t \in [0,1] \; \cond \; \frac{ \Fbh(t, \wtfun^*)}{F(t, \wtfun^*)} \leq \alpha}
\end{equation}
Assumption~\ref{assumption:conditional_twogroups_asymptotics} ensures the existences of unique $\tstar \in (t'', 1)$ for which in fact equality is attained, i.e., $\Fbh(\tstar, \wtfun^*)/ F(\tstar, \wtfun^*) = \alpha$.

We use~\eqref{eq:power_defn} as our definition for power; the results for other notions such as $1-\text{FNR}$ where $\text{FNR}$ is the false nondiscovery rate are analogous. Our claim is that the power of both naive weighted BH and IHW-BH asymptotically is equal to
\begin{equation}
\label{eq:asymptotic_power}
 \lim_{m \to \infty} \text{Power}^{\text{IHW-BH}}_m =   \lim_{m \to \infty} \text{Power}^{\text{Naive}}_m = \frac{F(\tstar, \wtfun^*) - F_0(\tstar, \wtfun^*)}{\int (1-\pi_0(x)) d\PP^X(x)}  
\end{equation}
We start by analyzing Naive weighted BH. First, we may use continuity and Glivenko Cantelli arguments leading to~\eqref{eq:r_glivenko_cantelli} and~\eqref{eq:r_glivenko_cantelli_part2}, along with the assumption on uniform convergence (in probability) of $\wtfunm$, to show that
\begin{equation}
 \begin{aligned}
 \label{eq:r_glivenko_cantelli_naive}
&\sup_{t \in [0,1]} \abs{ R(t, \hat{c}_{[m],\wtfunm}\cdot\wtfunm; [m])/m - F(t, \wtfun^*)} = \op \\
&\sup_{t \in [0,1]} \abs{ V(t, \hat{c}_{[m],\wtfunm}\cdot\wtfunm; [m])/m - F_0(t, \wtfun^*)} = \op \\
&\;\Vbh(t, \hat{c}_{[m],\wtfunm}\cdot\wtfunm; [m])/m = \Fbh(t, \wtfun^*) = t
 \end{aligned}
\end{equation}
The empirical process interpretation of naive weighted BH (analogous to~\eqref{eq:IHWbh_empirical_process} for IHW-BH) is as follows. Define
\begin{equation}
\label{eq:fdpnaivebh}
\hatFDPnaive(t) = \frac{\Vbh(t, \hat{c}_{[m],\wtfunm}\cdot\wtfunm; [m])}{\max\cb{1, R(t, \hat{c}_{[m],\wtfunm}\cdot\wtfunm; [m])}}
\end{equation}
\begin{equation}
\label{eq:tnaive}
\tnaive = \sup\cb{t \in [0,1] \cond \hatFDPnaive(t) \leq \alpha}
\end{equation}
Then the naive weighted BH procedure rejection rule takes the following form;
\begin{equation}
\label{eq:naivebh_empirical_process}
\text{Reject } i \in [m]\;\; \Longleftrightarrow \;\;P_i \leq \tnaive \cdot  \hat{c}_{[m],\wtfunm} \cdot \wtfunm(X_i)
\end{equation}
Our next step is to show that $\tnaive = \tstar + \op$. Fix any $\delta \in (0, \tstar)$, then using Assumption~\ref{assumption:conditional_twogroups_asymptotics} and~\eqref{eq:r_glivenko_cantelli_naive} we deduce that:
$$ \PPs{ \abs{\tnaive - \tstar} \leq \delta} \geq  \PPs{ \hatFDPnaive(\tstar - \delta) < \alpha,\;  \inf_{\delta' > \delta}\cb{\hatFDPnaive(\tstar + \delta')} > \alpha} = 1 - o(1)$$
Then, another application of~\eqref{eq:r_glivenko_cantelli_naive} and continuity properties of $F(\cdot \mid X_i=x)$ demonstrates that:
$$
\begin{aligned}
&R(\tnaive, \hat{c}_{[m],\wtfunm}\cdot\wtfunm; [m])/m = F(\tstar, \wtfun^*) + \op, \\ 
&V(\tnaive, \hat{c}_{[m],\wtfunm}\cdot\wtfunm; [m])/m = F_0(\tstar, \wtfun^*) + \op
\end{aligned}
$$
By the law of large numbers: $\sum_{i=1}^m H_i/m = \int (1-\pi_0(x)) d\PP^X(x) + \op$. By definition, 
$$\text{Power}^{\text{Naive}}_m = \EEs{\frac{R(\tnaive, \hat{c}_{[m],\wtfunm}\cdot\wtfunm; [m])/m - V(\tnaive,  \hat{c}_{[m],\wtfunm}\cdot\wtfunm; [m])/m}{\max\cb{1,\sum_{i=1}^m H_i}/m}}$$ and so by dominated convergence (note the term within the expectation above is in $[0,1]$):
$$\text{Power}^{\text{Naive}}_m = (F(\tstar, \wtfun^*) - F_0(\tstar, \wtfun^*))\bigg/\int (1-\pi_0(x)) d\PP^X(x) + o(1)$$
The same argument also applies for IHW-BH, leveraging results proved already in part (b) of the proposition. In particular it follows as for naive weighted BH that also $\tihw = \tstar + \op$ by using \eqref{eq:r_glivenko_cantelli_loo} and Lipschitz continuity of $F(\cdot \mid x)$. We also note in passing that under the assumptions of part (c), we could have omitted the conditional analysis required in part (b) to prove \eqref{eq:r_glivenko_cantelli_loo}. Instead, a more direct argument (along the lines of~\eqref{eq:gliv_cantelli_for_R}) could be given by noting that the i.i.d. structure and convergence of the weighting mechanism imply that
$$ \max_{\ell =1}^K \norm{\wtfunlooell(\cdot) - W^*(\cdot)}_{\infty} = \op$$
Cross-weighting derives its flexibility from guarantees established in part (b), that hold even if the above convergence property of the learned weight function does not hold.
\end{proof}

\subsection{Proof of Corollary~\ref{cor:ihw_gbh_asymptotic}}
\label{subsec:ihw_gbh_asymptotics_proof}
Let $\nu(x) = \PP[X_i = x]$ for $x \in [G]$. We assume without loss of generality that $\nu(x) > 0$ for all $x \in [G]$; otherwise it suffices to restrict the covariate space to $[G]\setminus\{x\}$.

In the setting with a categorical covariate,~\eqref{eq:weights_technical_condition} is automatically satisfied for any weighting function $\wtfun(x) \geq 0$. To see this, first note that $\int \wtfun(x) d\mathbb P^X(x) \geq  \max_{x \in [G]} \cb{\wtfun(x)\nu(x)}$. It also holds that
$$ \int \wtfun(x)^2 d\mathbb P^X(x) \leq G \cdot \max_{x \in [G]}\cb{\wtfun(x)^2 \nu(x)} \leq \p{G\;\big/\min_{x \in [G]} \cb{\nu(x)}} \cdot \max_{x \in [G]}\cb{\wtfun(x)^2 \nu(x)^2}$$
Thus,~\eqref{eq:weights_technical_condition} holds with $\Gamma = G \;\big/ \min_{x \in [G]}\cb{\nu(x)}$. It remains to check part (c) of Proposition~\ref{prop:asymp}. Recall from Algorithms~\ref{alg:GBH}, \ref{alg:IHW-GBH}, that the weighting rules take the form, for $x \in [G]$:
$$\widehat{W}(x) \propto \frac{1-\widehat{\pi}_0(x)}{\widehat{\pi}_0(x)},\; \widehat{\pi}_0(x) := \frac{1+\sum_{i: X_i=x} \ind\p{P_i > \tau}}{\abs{\cb{i: X_i=x}}(1-\tau)}$$
We need to exhibit the weighting function towards which the aforementioned weighting function converges under Assumption~\ref{assumption:conditional_twogroups_asymptotics}. To this end, let us note that:
$$ \EE[\widehat{\pi}_0(x) \cond X_1,\dotsc,X_m] =  \frac{1+\abs{\{i: X_i=x\}} \cdot \sqb{\pi_0(x)\cdot\p{ 1-\tau} + (1-\pi_0(x))\cdot\p{ 1-F_{\text{alt}}(\tau \cond x)}}}{\abs{\{i: X_i=x\}}(1-\tau)} $$ 
Next, define:
$$ \pi^*_0(x) = \pi_0(x) + \frac{(1-\pi_0(x))\cdot\p{ 1-F_{\text{alt}}(\tau \cond x)}}{1-\tau}$$
Note that for $\tau \in (0,1)$, by assumptions, $\pi_0(x) < 1$ and $F_{\text{alt}}(\tau \cond x) > \tau$ and so $\pi^*_0(x) < 1$. To avoid dealing with the (unlikely in multiple testing applications) situation that $\pi^*_0(x) =0$ we further assume that either $\pi_0(x) > 0$ (i.e., there are at least some null hypotheses) or $F_{\text{alt}}(\tau \cond x) < 1$. Thus henceforth we assume that $\pi^*_0(x) \in (0,1)$. 

The asymptotic weight function is $W^*(x)$ defined as
$$W^*(x) = \frac{1-\pi^*_0(x)}{\pi^*_0(x)}\bigg / \p{\sum_{g=1}^G \nu(g) \cdot \frac{1-\pi^*_0(g)}{\pi^*_0(g)}}$$
Notice that indeed $\int W^*(x) d\mathbb P^X(x) = \sum_{g=1}^G \nu(g) W^*(g) = 1$. By application of the law of large numbers and the continuous mapping theorem, we may deduce that:
$$ \widehat{\pi}_0(x) = \pi^*_0(x) + \op,\; \frac{\abs{\{i: X_i=x\}}}{m} = \nu(x) + \op, \;\widehat{W}(x) = W^*(x) + \op$$     

We may conclude by noting that $\mathcal{X}=[G]$ is finite, and so $\widehat{W}(x) = W^*(x) + \op$ for all $x \in [G]$ implies that
$$ \norm{\widehat{W}(\cdot) - W^*(\cdot)}_{\infty} = \op$$

%---------------------------------------------------
\section{Multiple testing with local false discovery rates}
\label{sec:fdr}
%----------------------------------------------------

Consider the conditional two-groups model~\eqref{eq:conditional_twogroups} and assume that $F(t \mid x)$ has Lebesgue density $f(t\mid x)$ for all $x$. Then define the conditional local fdr:

\begin{equation}
\label{eq:conditional_local_fdr_def}
\fdr(t \mid x) = \frac{ \pi_0(x)}{f(t \mid x)}
\end{equation}
We make two observations: First, for any threshold function $s: \mathcal{X} \to [0,1]$, one may show that
\begin{equation}
\label{eq:local_fdr_to_tail}
\text{mFDR}(s) := \PP[H_i = 0 \mid P_i \leq s(X_i)] = \EE[\fdr(P_i \mid X_i) \mid P_i \leq s(X_i)].
\end{equation}
Equation~\eqref{eq:local_fdr_to_tail} implies that we can estimate the $\text{mFDR}$ of a procedure with decision threshold $s$ (i.e., of the procedure that rejects hypotheses that satisfy $P_i \leq s(X_i)$) by
\begin{equation}
\label{eq:local_fdr_estimator}
\widehat{\text{mFDR}}(s) = \frac{\sum\limits_{i=1}^m \fdr(P_i \mid X_i) \, \ind(P_i \leq s(X_i))}{\sum\limits_{i=1}^m \ind(P_i \leq s(X_i))}
\end{equation}
Second, optimality considerations for multiple testing under model~\eqref{eq:conditional_twogroups} dictate that hypotheses should be ranked by $\fdr(P_i | X_i)$~\citep{sun2007oracle, cai2009simultaneous}.  Putting these two ideas together, we arrive at the oracle multiple testing procedure in Algorithm~\ref{alg:Cfdr}.

\begin{algorithm}[ht]
  \caption{The local fdr multiple testing procedure \label{alg:Cfdr}}
  \Input{A nominal level $\alpha \in (0,1)$, $m$ p-values $P_1,\dotsc,P_m$ and covariates $X_1,\dotsc,X_m$.}
  Let $\Cfdr_i \coloneqq \fdr(P_i \mid X_i)$ \;
  Let $\Cfdr_{(1)}, \dotsc, \Cfdr_{(m)}$ be the order statistics of $\Cfdr_1, \dotsc, \Cfdr_m$ and let $\Cfdr_{(0)} \coloneqq 0$ \;
  Let $ k^* = \max \left\{k \mid \frac{1}{k}\sum_{i=1}^k \Cfdr_{(i)} \leq \alpha \;,\;  1 \leq k \leq m \right\}$. If the latter set is empty, let $k^* = 0$. \;
   Reject all hypotheses with $\Cfdr_i \leq \Cfdr_{(k^*)}$
\end{algorithm}
Such a procedure indeed controls the $\FDR$ \citep{cai2009simultaneous}, if the conditional two-groups model~\eqref{eq:conditional_twogroups} is true and the oracle has access to the true model. Data-driven approximations to this procedure can be developed by plugging in estimates of the conditional densities $f(t\mid x)$ and $\pi_0(\cdot)$ \citep{cai2009simultaneous}.
Such a procedure can be shown to be asymptotically consistent, albeit no finite-sample results are available. 

%-----------------------------------------------------------------------
\section{Estimation and optimization of the conditional two-groups model}
\label{sec:estimation_and_optim}
%-----------------------------------------------------------------------

%-----------------------------------------------------------------------
\subsection{The nonparametric Grenander estimator}
\label{subsec:ihw_grenander}
%-----------------------------------------------------------------------

Our application of the Grenander estimator~\citep{grenander1956theory} to estimating the conditional two-groups model begins by binning the covariate $X_i$; for example through quantile-slicing or as the leaves of a tree. Henceforth we will assume $X_i$ is discrete and $X_i \in [G]$.

\subsubsection{Estimation}
To estimate $\widehat{F}^{-\ell}(\cdot \mid g)$, $g \in [G]$ we first form the ECDF (empirical cumulative distribution function) of the p-values $P_i$ with $i \notin I_{\ell}$ and $X_i =g$. Then we compute the least concave majorant of the ECDF. The latter operation can be computed fast through weighted isotonic regression; as implemented for example in the \texttt{gcmlcm} function of the R package \texttt{fdrtool}~\citep{strimmer2008fdrtool}. Furthermore, the computational complexity of fitting the Grenander estimator in one group is $O(m_g \cdot \log(m_g))$, , where $m_g= \#\{i:X_i = g\}$, and so it is of order $O(m\cdot \log(m))$ for all the groups.

The estimated $\widehat{F}^{-\ell}(\cdot \mid g)$ is a piecewise-linear, concave function. In particular, for a finite index set\footnote{We omit the out-of-fold specification $(-\ell)$ from subsequent notation when it improves readability.} $\mathcal{J}_g$ and real numbers $a_j^g, b_j^g$ for $j\in \mathcal{J}_G$ it holds that
\begin{equation}
	\label{eq:grenander_piecewise_linear}
  \widehat{F}^{-\ell}(t \mid g) = \min_{j \in \mathcal{J}_g}\{ a_j^g + b_j^g\, t\}
\end{equation}

For applications to $\FDR$ control we also need to estimate $\widehat{\pi}_0^{-\ell}(g)$. We set it to $\widehat{\pi}_0^{-\ell}(g)=1$ in all our experiments. An alternative would be to apply a $\pi_0$ estimator developed in the setting without groups (such as the estimator of~\citet{storey2004strong}) to the p-values $P_i$ with $i \notin I_{\ell}$ and $X_i =g$. This would yield $\widehat{\pi}_0^{-\ell}(g)$.

\subsubsection{Optimization through linear programming}
\label{subsubsec:grenander_lp}

Using the Grenander estimator simplifies subsequent optimization in two ways: first, the optimization variable is $G$-dimensional --instead of $\abs{I_{\ell}}$-dimensional-- as all $i \in I_{\ell}$ with $X_i = g$ receive the same weight. Second,~\eqref{eq:grenander_piecewise_linear} enables us to cast the underlying convex optimization problems as linear programs by introducing additional variables $\bar{F}_g \in [0,1]$ ($g=1,\dotsc,G$).
\paragraph{Optimization~\eqref{eq:wbonf_opt} for $k$-Bonferroni:} Let $k_{\alpha} = \alpha k /m$. We then solve the following linear program (LP) with optimization variables $(w_g, \bar{F}_g),\; g=1,\dotsc,G$:

\begin{equation}
\label{eq:wbonf_opt_grenander}
\begin{aligned}
&\text{maximize}\;\;  &\sum_{g=1}^G \abs{\cb{i \in I_{\ell}: X_i = g}} \cdot \bar{F}_g \\
&\text{s.t.}\;\; &\bar{F}_g
\leq a_j^g + b_j^g \cdot k_{\alpha} \cdot w_g, \;j \in \mathcal{J}_g,\; g=1,\dotsc,G \\
& &\sum_{g=1}^G \abs{\cb{i \in I_{\ell}: X_i = g}} \cdot w_g = \abs{I_{\ell}}\\
& &w_g \geq 0,\; g=1,\dotsc,G
\end{aligned}
\end{equation}
In our implementation, we solve this LP problem with the open-source \texttt{SYMPHONY/Clp} solver of the COIN-OR project~\citepsup{lougee2003coinor} with the interface provided by the R/Bioconductor package \texttt{lpsymphony}~\citepsup{lpsymphony}. Hypothesis $i$ in fold $I_{\ell}$ with covariate $X_i = g$ is then assigned weight $w_g$, where ($w_1, \dotsc, w_G$) is the optimal weight vector of the above optimization problem.

\paragraph{Optimization~\eqref{eq:wbh_opt} for BH:} The optimization here is similar with the difference that we optimize directly over the thresholds $t_g, \; g=1,\dotsc, G$ and we also enforce the plug-in FDR constraint. Concretely, we solve the linear program with optimization variables $(t_g, \bar{F}_g),\; g=1,\dotsc,G$:
\begin{equation}
\label{eq:wbh_opt_grenander}
\begin{aligned}
&\text{maximize}\;\;  &\sum_{g=1}^G \abs{\cb{i \in I_{\ell}: X_i = g}} \cdot \bar{F}_g \\
&\text{s.t.}\;\; &\bar{F}_g
\leq a_j^g + b_j^g \cdot t_g, \;j \in \mathcal{J}_g,\; g=1,\dotsc,G \\
& &\sum_{g=1}^G \abs{\cb{i \in I_{\ell}: X_i = g}}\p{\widehat{\pi}_0^{-\ell}(g) \cdot t_g - \alpha \cdot \bar{F}_g} \; \leq \; 0\\
& &0 \leq t_g \leq 1,\; g=1,\dotsc,G
\end{aligned}
\end{equation}
Letting $(t_1,\dotsc,t_G)$ the optimal threshold vector, we then let
$$w_g = \abs{I_{\ell}} \cdot t_g \bigg/ \p{ \sum_{g=1}^G  \abs{\cb{i \in I_{\ell}: X_i = g}} t_g},$$
unless all $t_g=0$, in which case we set all $w_g=1$. $w_g$ is the weight assigned to hypotheses $i \in I_{\ell}$ with $X_i=g$.

\paragraph{Convex constraints on the weights:} For both linear programs~\eqref{eq:wbonf_opt_grenander} and ~\eqref{eq:wbh_opt_grenander} it is possible to incorporate additional linear constraints (so that the problems remain linear programs) that enforce weight functions of lower complexity. A concrete example is to enforce low total variation of the weight vector $(w_1,\dotsc,w_G)$ i.e., to enforce $\sum_{g=2}^G \abs{w_g - w_{g-1}} \leq \lambda$, for $\lambda \geq 0$. This may be directly incorporated into~\eqref{eq:wbonf_opt_grenander}. We may also add this constraint to problem~\eqref{eq:wbh_opt_grenander} in terms of $t_1,\dotsc,t_G$ as follows
\begin{equation}
\label{eq:grenander_tv_constraint}
 \sum_{g=2}^G \abs{t_g - t_{g-1}} \leq \frac{\lambda}{\abs{I_{\ell}}}\sum_{g=1}^G \abs{\cb{i \in I_{\ell}: X_i = g}} t_g
\end{equation}
Throughout this work we set $\lambda = \infty$ (i.e., we do not add the above total variation constraints), with the exception of the application described in \supplementname~\ref{sec:hqtl_suppl}.

\subsubsection{Direct optimization}
\label{subsubsec:direct_opt}
Here we describe an alternative optimization scheme that does not require the use of a linear programming solver and has guarantees on its computational complexity. For our numerical examples, however, the linear programming approach is fast enough.

We describe our algorithm for solving the $k$-Bonferroni objective~\eqref{eq:wbonf_opt}; the steps for the BH objective~\eqref{eq:wbh_opt_grenander} are similar.

Let~\eqref{eq:grenander_piecewise_linear} be the fitted Grenander estimator in group $g$. Let the non-zero slopes in group $g$ be sorted as \smash{$b_1^g > \dotsc > b_{\abs{\mathcal{J}_g}}^g$=0} and let $s_j^g, j \in \mathcal{J}_g$ be the points at which the slope changes, i.e., the slope is equal to $b_j^g$ in the interval $(s_{j-1}^g, s_j^g)$. At the boundaries we define $s_0^g = 0$ and $s_{\abs{\mathcal{J}_g}}^g = 1$. Further, consider the set:
\begin{equation}
\label{eq:sorted_Bs}
\mathcal{B} =  \cb{ b_j^g: \;j \in \mathcal{J}_g, \; g \in [G]}
\end{equation}
Algorithm~\ref{alg:wbonf_opt} provides a computational routine for optimizing the objective~\eqref{eq:wbonf_opt} with computational complexity upper bounded by $O(m \cdot \log(m))$.

\begin{algorithm}
    \label{alg:wbonf_opt}
    \caption{Optimization of the $k$-Bonferroni objective~\eqref{eq:wbonf_opt} when $X$-conditional distributions are estimated with Grenander's method.}
    \Input{the number of type-I errors to protect against $k$\\
           the nominal level $\alpha$\\
           the total number of tests $m$\\
           the number of tests within fold $\ell$ and each group $g \in [G]$, $\tilde{m}_{g}= \abs{\cb{i \in I_{\ell}: X_i = g}}$\\
           the fitted Grenander estimator~\eqref{eq:grenander_piecewise_linear} with slope change points $s_j^g$ and slopes $b_j^g$\\
           the set $\mathcal{B}$ defined in~\eqref{eq:sorted_Bs}
           }
      \hrulealg
     Sort the set $\mathcal{B}$.\;
     Let $k_{\alpha} = \alpha \cdot k /m$ and $\tilde{m} = \sum_{g=1}^G \tilde{m}_g$.\;
    \While{$\abs{\mathcal{B}} > 1$}{
    	Let $\lambda$ be a middle element in $\mathcal{B}$ (i.e., the median if $\abs{\mathcal{B}}$ is odd).\;
    	\For{$g \in [G]$}{
    	\uIf{there exists $j \in \mathcal{J}_g$ such that $\lambda = b_j^{g}$}{
    		let  $\ell_{g}(\lambda) = s_{j-1}^g$,  $u_{g}(\lambda) = s_{j}^g$.\;
    	}\uElseIf{there exists $j \in \mathcal{J}_g$ such that $\lambda \in (b_{j+1}^{g}, b_{j}^{g})$}{
    	    let $\ell_{g}(\lambda) = u_{g}(\lambda) = s_j^g$.\;
    	}\uElseIf{$\lambda > b_j^g$ for all $j \in \mathcal{J}_g$ }{
    		let $\ell_{g}(\lambda) = u_{g}(\lambda) = 0$.\;
    	}
    	}
    	Compute: $\weightbudget_{\ell}(\lambda) = \sum_{g=1}^G  (\tilde{m}_g/\tilde{m}) \ell_g(\lambda) / k_{\alpha}$ and \\$\;\;\;\;\quad\quad\quad\;\;\weightbudget_{u}(\lambda) = \sum_{g=1}^G (\tilde{m}_g/\tilde{m})  u_g(\lambda) / k_{\alpha}$\;
    	\;
				
		\uIf{$\weightbudget_{\ell}(\lambda) > 1$}{
			$\mathcal{B} \leftarrow  \{ b \in \mathcal{B}: b > \lambda\}\;$
		}\uElseIf{$\weightbudget_{u}(\lambda) < 1$}{
			$\mathcal{B} \leftarrow  \{ b \in \mathcal{B}: b < \lambda\}\;$
	    }\uElseIf{$1 \in [\weightbudget_{\ell}(\lambda), \weightbudget_{u}(\lambda)]$}{
	    	$\mathcal{B} \leftarrow  \{ \lambda\}$\;
	    }
    }
    Let $\lambda^*$ the unique element remaining in $\mathcal{B}$.\;
    Let $c=(\weightbudget_{u}(\lambda^*) -1)/(\weightbudget_{u}(\lambda^*) - \weightbudget_{\ell}(\lambda^*))$.\;
    Let $t_g =  c  \ell_{g}(\lambda^*) + (1-c)u_{g}(\lambda^*),\; g \in [G]$. \;
    Return the weights $w_g^* = t_g / k_{\alpha},\; g \in [G]$.\;
\end{algorithm}

\noindent The following proof verifies the correctness of the algorithm above and the worst-case computational complexity.

\begin{proof}
We need to first check that the algorithm terminates, i.e., that there exists a $\lambda^*$ so that $1 \in [\weightbudget_{\ell}(\lambda), \weightbudget_{u}(\lambda)]$ To this end, note that if we choose $\lambda = \max \mathcal{B}$, then all $\ell_g(\lambda) = 0$ and so $\weightbudget_{\ell}(\lambda) =0$. On the other hand, letting $\lambda=0$, then we can pick all $u_g(\lambda)=1$, i.e., $\weightbudget_{u}(\lambda) = 1/k_{\alpha} > 1$. It remains to observe that for adjacent $\lambda_{j}, \lambda_{j+1}$ in $\mathcal{B}$, it holds that
$$\weightbudget_{\ell}(\lambda_{j+1}) =  \weightbudget_{u}(\lambda_{j}),$$
and also that for all $\lambda$, $\weightbudget_{\ell}(\lambda) \leq \weightbudget_{u}(\lambda)$. As the algorithm terminates, we may now check computational complexity. First, note that $\abs{\mathcal{B}} = O(m)$ since the Grenander estimator can only jump at support points of the per-group empirical distribution function. Thus the initial sorting step of $\mathcal{B}$ requires at most $O(m\log(m))$ operations. The `while' loop of the algorithm proceeds by bisection of $\mathcal{B}$, hence will comprise of at most $O(\log(m))$ iterations and the cost of each iteration step is at most $O(m)$. Computation after the while loop is negligible ($O(G)$ operations). Thus, the total complexity of this algorithm is $O(m\log(m))$ at most.

Second, we need to check the Karush–Kuhn–Tucker (KKT)~\citepsup{rockafellar1970convex} conditions for convex programming to verify the optimality of the weights returned by Algorithm~\ref{alg:wbonf_opt}. Let $\bm{\nu} = (\nu_1,\dotsc,\nu_G)$ the dual variables corresponding to the non-negativity constraint and $\lambda$ the dual variable corresponding to the weight-budget constraint. The Lagrangian takes the form
$$ \mathcal{L}(\mathbf{w}, \lambda, \bm{\nu}) = \sum_{g=1}^G  \tilde{m}_{g} \cdot \widehat{F}^{-\ell}\p{w_g \cdot k_{\alpha}} + \bm{\nu}^\top \mathbf{w} - \lambda \cdot k_{\alpha}\p{ \sum_{g=1}^G \tilde{m}_g w_g - \tilde{m}}.$$
We seek to specify dual and primal optimal variables. We set the dual $\lambda^*$ and primal $w_g^*$ as described in the last steps of Algorithm~\ref{alg:wbonf_opt}. For the dual variables corresponding to the non-negativity constraints, we set $\nu^*_g = 0$ if $w_g^* >0$ and $\nu_g^* = \tilde{m}_g \cdot k_{\alpha}\p{\lambda^* - b_1^g}$ if $w_g = 0$. \textbf{Complementary slackness} thus holds by construction. Furthermore, note that when $w_g^*=0$, then Algorithm~\ref{alg:wbonf_opt} ensures that $\lambda^* \geq b_1^g$ and so $\nu_g^* \geq 0$. Hence for all $g$, $\nu_g^* \geq 0$ and so \textbf{dual feasibility} holds. \textbf{Primal feasibility}, i.e., $w_g^* \geq 0$ and $\sum \tilde{m}_{g} w_g^* = \tilde{m}$ also hold by construction.

It remains to check that \textbf{stationarity} holds. Let us take take the superdifferential of the Lagrangian along the $g$-th coordinate, where we keep $g\in [G]$ fixed.
$$\partial_g L(\mathbf{w}, \lambda, \mathbf{\nu}) = \tilde{m}_{g} \cdot  k_{\alpha} \cdot \p{ \partial \widehat{F}^{-\ell}\p{k_{\alpha} \cdot w_g \cond g} - \lambda} + \nu_g$$
We next distinguish two cases according to the value of $w_g^*$.\\
\noindent \textbf{Case 1, $w_g^*=0$:} In this case, $b_1^g \in \partial \widehat{F}^{-\ell}\p{w_g^* \cdot k_{\alpha} \cond g}$ and $\nu_g^*$ is defined precisely so that $0 \in \partial_g L(\mathbf{w}^*, \lambda^*, \mathbf{\nu}^*)$.\\
\noindent \textbf{Case 2, $w_g^*>0$:}  First let us quickly study $\partial \widehat{F}^{-\ell}(t \cond g)$ for $t \in (0,1)$. If $t = s_j^g$ for $j \in \mathcal{J}_g$, then $\partial \widehat{F}^{-\ell}(t \cond g) = [b_{j+1}^g, b_j^g]$ and if $t \in (s_{j-1}^g, s_{j}^g)$, then $\partial \widehat{F}^{-\ell}(t \cond g) = \cb{b_j^g}$. In both cases, it holds that $\lambda^* \in \partial\widehat{F}^{-\ell}(t \cond g)$ and so, since $\nu_g^*=0$, it again follows that $0 \in \partial_g L(\mathbf{w}^*, \lambda^*, \mathbf{\nu}^*)$.
\end{proof}

%-----------------------------------------------------------------------
\subsection{Beta-uniform mixture GLM}
\label{subsec:ihw_betamix}
%-----------------------------------------------------------------------
In this section we consider the conditional two-groups model~\eqref{eq:conditional_twogroups} with parametrization~\eqref{eq:betamix_simulation_model}, where we assume throughout that $\beta(x) < 1$ and $0<\pi_0(x) < 1$ hold strictly. We first explain how to estimate the parameters of the conditional two-groups model given access to $X_i$ and censored p-values ($P_i\ind(P_i > \tau)$) and then we explain the optimization procedure for deriving optimal weights.

As a preliminary step, we introduce explicit notation for the CDF and pdf of the $\text{Beta}(\beta,1)$ distribution
$$F_{\beta}(t) = t^{\beta},\; f_{\beta}(t) = \beta t^{\beta - 1}$$

\subsubsection{Estimation}
In this section we let $Y_i = -\log(P_i)$. Our goal is estimation based on the censored data outside fold $\ell$, i.e., $D_{-\ell}(\tau) = ((X_i, P_i\ind(P_i > \tau)))_{i \in I_{\ell}^c}$\footnote{The case $\tau=0$ corresponds to no censoring. Without cross-weighting we use the data corresponding to all indices $i=1,\dotsc,m$.}. We will proceed by maximum likelihood estimation and optimize the (non-convex) objective through the EM algorithm. The full-data (i.e., if we could observe $((X_i, P_i,H_i))_{i \in I_{\ell}^c}$) log-likelihood decouples into the sum of the log-likelihood of two generalized linear models (GLMs); a binomial GLM and a Gamma GLM (cf. (16) in \citet{lei2016adapt})
\begin{equation}
\label{eq:betamix_loglikelihood}
\begin{aligned}
\ell(a, a_0, b, b_0; \mathbf{P}, \mathbf{X}, \mathbf{H}) = &\sum_{i \in I_{\ell}^c} \p{ -H_i(a_0 + a^\top X_i) - \log\sqb{1+\exp\p{-a_0 - a^\top X_i}}}\\
& \;+\;\sum_{i \in I_{\ell}^c} \p{H_i\sqb{-Y_i(b_0 + b^\top X_i) + \log\p{b_0 + b^\top X_i}}}
\end{aligned}
\end{equation}
During the $r$-th iteration of EM, we keep track of the imputed data (E-step; more below) $\hat{Y}^{(r)}$, $\hat{H}^{(r)}$. Furthermore in principle we should keep track of the parameters $\hat{a}_0^{(r)}, \hat{a}^{(r)}, \hat{b}_0^{(r)}, \hat{b}^{(r)}$. Instead, we keep track of $\hat{\pi_0}^{(r)}(x) =  \expit(-\hat{a}_0^{(r)} - \hat{a}^{(r)\top} x)$ and $\hat{\beta}^{(r)}(x) =  \hat{b}_0^{(r)} + \hat{b}^{(r)\top} x$ both evaluated at $X_i, i \in I_{\ell}^c$.

We now describe the details of the EM algorithm.

\paragraph{E-step:} For the $r$-th E-step, we need to compute:
$$\EEs[\hat{\pi}_0^{(r-1)}, \hat{\beta}^{(r-1)}]{\ell(a, a_0, b, b_0; \mathbf{P}, \mathbf{X}, \mathbf{H}) \mid D_{-\ell}(\tau)}$$
This boils down to computing 
$$\hat{H}_i^{(r)} = \EEs[\hat{\pi}_0^{(r-1)}, \hat{\beta}^{(r-1)}]{ H_i \mid D_{-\ell}(\tau)} \text{ and } \hat{Y}_i^{(r)} = \EEs[\hat{\pi}_0^{(r-1)}, \hat{\beta}^{(r-1)}]{ Y_i \mid H_i=1,  D_{-\ell}(\tau)}$$
and plugging these into~\eqref{eq:betamix_loglikelihood} in lieu of $H_i, Y_i$.

\begin{itemize}
\item[--] $\hat{H}_i^{(r)}$ update: 

$$\hat{H}_i^{(r)} = \left\{\begin{matrix} &1- \frac{\hat{\pi}_0^{(r-1)}\cdot \tau}{\hat{\pi}_0^{(r-1)}\cdot \tau  + (1-\hat{\pi}_0^{(r-1)})\cdot F_{\hat{\beta}^{(r-1)}(X_i)}(\tau)} &, \text{ if } P_i \leq \tau \\
& 1- \frac{\hat{\pi}_0^{(r-1)}}{\hat{\pi}_0^{(r-1)}  + (1-\hat{\pi}_0^{(r-1)})\cdot f_{\hat{\beta}^{(r-1)}(X_i)}(P_i)} &, \text{ if } P_i > \tau \end{matrix}\right.$$

\item[--] $\hat{Y}_i^{(r)}$ update: 
 $$\hat{Y}_i^{(r)} = \left\{\begin{matrix} & \frac{1}{\hat{\beta}^{(r-1)}(X_i)} - \log(\tau) &, \text{ if } P_i \leq \tau \\ &Y_i&, \text{ if } P_i > \tau \end{matrix}\right.$$

\end{itemize}

\paragraph{M-step:} As already alluded to, the M-step consists of fitting two GLMs, a (quasi)binomial GLM ($\hat{H}_i^{(r)} \in [0,1]$ takes on fractional values) and a weighted Gamma GLM. In R pseudocode, the M-step is as follows:
\begin{itemize}
\item[--] $$\hat{\pi}_0^{(r)} = \text{predict}(\text{glm}(1-\hat{H}^{(r)} \sim 1+X,\; \text{family=quasibinomial()}), \; \text{type=``response''})$$
 
\item[--] $$\hat{\beta}^{(r)} = \text{predict}(\text{glm}(\hat{Y}^{(r)} \sim 1+X,\; \text{family=Gamma(), weights=}\hat{H}^{(r)}),  \; \text{type=``link''})$$ 
\end{itemize}
In this step we also seek to ensure $\beta(x) < 1, 0<\pi_0(x) < 1$ (so that strict concavity of the estimated p-value distribution holds conditionally on all $x$). To this end we introduce parameters $\pi_{0,\text{min}}, \pi_{0,\text{max}}$ and $\beta_{\text{max}}$ and clamp the $\beta(x), \pi_0(x)$ estimates ($\clamp(x;a,b) = \max\cb{ \min\cb{x, b}, a}$) to the above ranges. In our implementation we use $\pi_{0,\text{min}}=0.1, \pi_{0,\text{max}}=0.99$ and $\beta_{\text{max}}=0.9$ (we have not needed to lower bound $\beta$ in our experiments).     

\paragraph{Initialization:}

\newcommand{\taubl}{\tau^{\text{BL}}}

\begin{itemize}
	\item[--] $\hat{Y}^{(0)}$: We initialize $\hat{Y}_i$ by $Y_i$ if
	 $P_i$ is not censored and by $-\log(\tau/2)$ otherwise.
$$\hat{Y}^{(0)}_i = Y_i \ind\p{P_i > \tau} -\log(\tau/2)\ind\p{P_i \leq \tau}$$
    \item[--] $\hat{\pi}_0^{(0)}$: The $\hat{\pi}_0(X_i)$ are initialized through the procedure of~\citet{boca2018direct}. First, let $\taubl \geq \tau$; in our simulations we use $\taubl=0.5$. Then we fit a logistic regression of $\ind\p{P_i \geq \taubl}$ onto $X_i$, let $\hat{\PP}[P_i \geq \taubl \mid X_i]$ the fitted probabilities and finally we set $$\hat{\pi}_0^{(0)}(X_i) = \clamp(\;\hat{\PP}[P_i \geq \taubl \mid X_i]/(1-\tau);\; \pi_{0,\text{min}},\; \pi_{0,\text{max}})$$
	\item[--] $\hat{H}^{(0)}$: We first compute the adjusted p-values $\text{adj}P_i$ of the BH procedure applied to $P_i \lor \tau$ (i.e., in R pseudocode: \texttt{p.adjust(pmax(Ps, tau), method="BH")}) and then we set:
	$$ \hat{H}_i^{(0)} = 1 - \text{adj}P_i \cdot \hat{\pi}_0^{(0)}(X_i)$$
\end{itemize}

\paragraph{Output:} Let $r^*$ the final iteration of the EM algorithm, we keep $\hat{\beta}^{-\ell}(\cdot) = \hat{\beta}^{(r^*)}(\cdot)$ and  $\hat{\pi}_0^{(r^*)}(\cdot)$. These fully specify the estimated conditional distribution 
$$\widehat{F}^{-\ell}(t \mid x) = \hat{\pi}_0^{(r^*)}(x)\cdot t  +  (1-\hat{\pi}_0^{(r^*)}(x))\cdot F_{\hat{\beta}^{-\ell}(x)}(t).$$
When learning weights for IHW BH as in~\eqref{eq:wbh_opt}, we set the estimated conditional distribution as above, while keeping $\hat{\pi}_0^{-\ell}(\cdot) = 1$ instead of using the output from the EM algorithm. An alternative, following~\citet{markitsis2010censored}, would be to set
$$\hat{\pi}_0^{-\ell}(x) = \hat{\pi}_0^{(r^*)}(x) + (1-\hat{\pi}_0^{(r^*)}(x))\cdot f_{\hat{\beta}^{-\ell}(x)}(1).$$

\subsubsection{Optimization}
\label{subsubsec:betamix_opt}
The estimated conditional distributions and densities take the form:
$$ \widehat{F}^{-\ell}(t \mid X_i=x) = \hat{\pi}_0^{-\ell}(x) \cdot t + (1-\hat{\pi}_0^{-\ell}(x))t^{\hat{\beta}^{-\ell}(x)}$$
$$ \hat{f}^{-\ell}(t \mid X_i=x) = \hat{\pi}_0^{-\ell}(x) + (1-\hat{\pi}_0^{-\ell}(x))\cdot\hat{\beta}^{-\ell}(x)\cdot t^{\hat{\beta}^{-\ell}(x)-1}$$

\paragraph{Optimization~\eqref{eq:wbonf_opt} for $k$-Bonferroni:}
We seek to maximize $\sum_{i \in I_{\ell}} \widehat{F}^{-\ell}(k_{\alpha} \cdot w_i \mid X_i)$ subject to $w_i \geq 0$, $\sum_{i \in I_{\ell}} w_i = \abs{I_{\ell}}$, where $k_{\alpha} = \alpha k /m$. This a convex optimization problem and furthermore strong duality is attained, e.g., by Slater's condition (also note that the program is feasible; take $w_i =1$).

Assume momentarily that the optimizer satisfies $w_i >0$ for all $i$. Let $\lambda$ be the Lagrange multiplier corresponding to the constraint $\sum_{i \in I_{\ell}} w_i = \abs{I_{\ell}}$ . Then, differentiating the Lagrangian with respect to $w_i$, we see that it must hold that:
$$ \hat{f}^{-\ell}(k_{\alpha} \cdot w_i \mid X_i)k_{\alpha} - \lambda \stackrel{!}{=} 0$$
So:
$$  \hat{\pi}_0^{-\ell}(x) + (1-\hat{\pi}_0^{-\ell}(x))\cdot\hat{\beta}^{-\ell}(x)\cdot (k_{\alpha} \cdot w_i )^{\hat{\beta}^{-\ell}(x)-1}  \stackrel{!}{=}  \lambda/k_{\alpha} $$
Since $\hat{\beta}^{-\ell}(X_i) < 1$ and $\hat{\pi}^{-\ell}_0(X_i) < 1$ by our estimation procedure, we may solve the equation above analytically for $w_i > 0$. We call this solution $w_i(\lambda)$. Then we use bisection over $\lambda$ to find $\lambda^*$ such that the equality constraint is satisfied, i.e., $\sum_{i \in I_{\ell}} w_i(\lambda^*) = \abs{I_{\ell}}$. Then the optimizing weights are $w_i = w_i(\lambda^*)$. 

We may derive the computational complexity of the optimization step as follows: We can minimize the Lagrangian analytically in $O(m)$ operations. To find the optimal dual variable $\lambda^*$ we need to use bisection. Thus, we need roughly $O(m \cdot \log(1/\delta))$ operations, where $\delta$ is a parameter controlling tolerance (accuracy).

\paragraph{Optimization~\eqref{eq:wbh_opt} for BH:} Here we seek to maximize \smash{$\sum_{i \in I_{\ell}} \widehat{F}^{-\ell}(t_i \mid X_i)$ over $t_i \geq 0$} subject to \smash{$\sum_{i \in I_{\ell}} \hat{\pi}^{-\ell}_0(X_i)t_i \leq \alpha \sum_{i \in I_{\ell}}\widehat{F}^{-\ell}\left(t_i \mid X_i \right)$}. We may directly verify the conditions of Theorem 2 in~\citet{lei2016adapt} (which ensures strong duality) and conclude that there exists $\lambda \in (0,1)$ such that at the optimal solution:
$$\fdr^{-\ell}(t_i \mid X_i) \stackrel{!}{=}  \lambda \text{ for all } i \in I_{\ell}$$
Here \smash{$\fdr^{-\ell}(t_i \mid X_i)$} is defined as in~\eqref{eq:conditional_local_fdr_def} with population quantities replaced by estimated ones. Rearranging, this implies that:
$$ \hat{f}^{-\ell}(t_i \mid X_i) \stackrel{!}{=} \hat{\pi}^{-\ell}_0(X_i)/\lambda$$
As already described for $k$-Bonferroni, for each fixed $\lambda$ we may solve the above expression analytically for $t_i$, say by $t_i(\lambda) > 0$. Then it only remains to use bisection to find $\lambda^*$ such that
$$\sum_{i \in I_{\ell}} \hat{\pi}^{-\ell}_0(X_i)t_i(\lambda^*) = \alpha \sum_{i \in I_{\ell}}\widehat{F}^{-\ell}\left(t_i(\lambda^*) \mid X_i \right).$$
Finally, hypothesis $i \in I_{\ell}$ is assigned weight $W_i = \abs{I_{\ell}}\cdot t_i(\lambda^*)\big/\sum_{j \in I_{\ell}}t_j(\lambda^*)$.

We note that here, just as for $k$-Bonferroni, the computational complexity scales as $O(m \cdot \log(1/\delta))$ operations, where $\delta$ is a parameter controlling tolerance (accuracy).

%------------------------------------------------------------------------
\section{More details on the data application of Section~\ref{sec:hQTLexample}}
\label{sec:hqtl_suppl}
%------------------------------------------------------------------------

For the hQTL example, we used the dataset described in~\cite{grubert2015genetic} and
looked for associations between SNPs and the histone modification mark (H3K27ac) on human
Chromosomes 1 and 2. p-values for association were calculated using Matrix eQTL~\citep{shabalin2012matrix}.  

As a covariate we used the linear genomic distance between the SNP and the
ChIP-seq signal, which we discretized using non-uniform binning: the bins corresponded to genomic segments of length $10$~kb (kilobase) bins up to $300$~kb (i.e., the categories were $0-10$~kb, $10-20$~kb, $\dotsc$, $290-300$~kb), to segments of length $100$~kb up to $1$~Mb and finally to segments of length $10$~Mb for the rest of the hypotheses. The longest genomic distance between SNPs and H3K27ac was approximately equal to $24$~Mb.

For the application of IHW-BY (cf. Theorem~\ref{thm:ihw-by}), we split hypotheses into two folds corresponding to the two chromosomes. Honest weights are learned within each fold with the strategy described in Section~\ref{subsec:bh_weights} and \supplementname~\ref{subsec:ihw_grenander} based on the Grenander estimator. Note that we set $\hat{\pi}_0^{-\ell}=1$. Furthermore, we apply a mild constraint on the total variation of the learnd weights, i.e. by including the constraint~\eqref{eq:grenander_tv_constraint} with $\lambda = 2000$ in the linear programming problem~\eqref{eq:wbh_opt_grenander}.

%----------------------------------------------------------
\section{Choice and examples of informative covariates}
\label{sec:covariate_examples}
%----------------------------------------------------------

Covariates that can take the role of $X_i$ in the conditional two-groups model~\eqref{eq:conditional_twogroups} are available in many multiple testing applications of practical interest, and in this section we discuss a range of examples. We will group them into domain-specific and statistical covariates. Whereas the former derive from an understanding of the data-generating process, the latter reflect mathematical properties of the specific test procedure used to compute the p-values. Domain-specific covariates are often informative about prior probabilities (i.e., the function $\pi_0(x)$ depends on $x$), statistical covariates about the power of the test and thus the shape of the alternative distribution function $F_{\text{alt}}(\cdot \mid X_i=x)$. The categorization is informal, loose and partially overlapping.

For a given application, there will often be more than one possible choice of covariate. In our formulation of the conditional two-groups model~\eqref{eq:conditional_twogroups}, we assume for simplicity of notation that $X_i$ is either one particular choice, or the combination of several original covariates into a single ``effective'' covariate, e.g., by taking the Cartesian product. The details of how to select or combine will depend on the application and the data and are beyond the scope of this paper.

%----------------------------------------------------------
\subsection{Domain-specific covariates}
\label{subsec:domain_spec_covariates}
%----------------------------------------------------------

In many scientific applications, informative covariates are apparent to domain scientists due to mechanistic insight or prior experience. Examples include:

\begin{itemize}
\item \textbf{Genomic distance between SNPs and peaks.} This is the covariate in our motivating example in Figure~\ref{fig:H3K27ac_histograms} and Section~\ref{sec:hQTLexample}. The p-values are from testing the association between SNPs and H3K27ac peak heights across different individuals from the human population. The choice of covariate is motivated by the expectation that many of the true instances where a DNA polymorphism affects a H3K27ac peak are short-range, so that $\pi_0$ for hypotheses with a short distance is smaller than for those where SNP and peak are far apart.
\item \textbf{Physical distance between pairs of firing neurons}. It is now possible to simultaneously measure the activity of many neurons, and there is interest in determining whether two neurons are firing in synchrony \citep{scott2014false}. We know that neurons in close proximity are a-priori more likely to be interacting, thus, the distance between neurons can be used as a covariate for association tests between pairs of neurons.
\item \textbf{Gene expression patterns in nearby genetic variants}. Genome-wide association studies (GWAS) look for statistical associations between genetic variants in a population with prevalence of a disease. Once discovered, such an association can be the basis for a follow-up mechanistic study. Sample size and power tend to be limiting bottlenecks of many GWAS due to multiple testing and to the study's expense. Power can be increased by considering (phenotype-unrelated) gene expression patterns around the loci of the genetic variants \citepsup{baillie2018_shared_activity_patterns}. 

\item \textbf{P-values from a distinct but related experiment}. For example, \citetsup{fortney2015genome} used data from previous, independent GWAS for related diseases to increase the power of a GWAS study of a longevity phenotype.
\end{itemize}

In a different context---multivariate regression rather than hypothesis testing---the widespread existence of such covariates was observed by \citetsup{wiel2016better}, who used the term ``co-data'' for them and developed a weighted ridge regression procedure, with data-driven penalization weights.

%----------------------------------------------------------
\subsection{Statistical covariates}
\label{subsec:stat_covariates}
%----------------------------------------------------------
In single hypothesis testing, classical theory \citep{lehmann2005testing} dictates that the whole dataset should be reduced to a sufficient statistic, which in turn can be used to derive the best test statistic under optimality considerations. Everything else, can be discarded or should be conditioned on. This data compression comes without any loss of statistical power.

However, the $m$ resulting p-values for the individual tests are in general not able to capture how one should weigh the hypotheses relative to each other to arrive at an optimal multiple testing protocol \citepsup{storey2007optimal}. The consequence is that information irrelevant for single hypothesis testing can be embedded in the conditional two-groups framework and can help increase the power of the resulting multiple testing procedure; sometimes dramatically so.

%----------------------------------------------------------
\subsubsection{Sample size}\label{subsec:sample_size_N_i}
%----------------------------------------------------------
A generic covariate, likely to be useful whenever it differs across tests, is the sample size $N_i$. Note that if the test statistic is continuous and the null hypothesis is simple, then the p-value $P_i$ under the null is uniformly distributed independently of $N_i$. Often, there is no reason to expect that the prior probability of a hypothesis being true depends on $N_i$. However, the alternative distribution will depend on $N_i$: for higher sample size, we have more power.

A simple, but generic and instructive example is as follows: consider a series of one-sided $z$-tests in which we observe independent $Y_1^i, \dotsc, Y_{N_i}^i \sim \mathcal{N}(\mu_i, 1)$, where $\mu_i > 0$ if $H_i=1$ and $\mu_i=0$ otherwise.  We can use $P_i = 1-\Phi\left(N_i^{1/2} \; \overline{Y^i}\right)$ as our statistic, where $\overline{Y^i}$ is the sample average of $Y_1^i, \dotsc, Y_{N_i}^i$. Then the alternative distribution of the $i$-th test is

\begin{equation}\label{eq:sample_size_example_Falti}
F_{\text{alt}, i}(t) = 1 - \Phi(\Phi^{-1}(1-t) - \sqrt{N_i}\,\mu_i).
\end{equation}

Now consider the case in which $\pi_{0,i}= \pi_0$ and $\mu_i=\mu H_0\;\forall i$, i.e., a common prior probability and a common effect size. In this case, Equation~\eqref{eq:sample_size_example_Falti} leads to the conditional two-groups model with covariate $N_i$ and $F_{\text{alt}, i}(t) = F_{\text{alt}}(t \mid N_i)$. Then, to maximize discoveries and thus power, hypotheses with large sample sizes $N_i$ should be prioritized. The methods described here are able to accomplish this automatically.

\begin{remark}
\label{remark:streetlight}
At this point, readers might ask themselves whether this is desirable -- since, in practice, different effect sizes $\mu_i$ may be present. Prioritizing hypotheses with large sample sizes $N_i$ will lead to a trade-off where some discoveries with smaller  $N_i$  but higher $\mu_i$ are missed, for the benefit of making more discoveries with larger $N_i$ but smaller $\mu_i$. Yet, the former might be more valuable to us. In a way, one can draw analogies to the streetlight effect: if we have lost our keys during a walk at night and have no idea where it happened, it makes sense to start searching under the streetlight, where it is easiest to see. However, if we do have guesses where we might have dropped them, it makes sense to combine these guesses with the ease of seeing in each place to arrive at an optimal search schedule.
\end{remark}

\begin{remark}
The optimal weights are not necessarily a monotonic function of the sample size. With IHW, it is possible that hypotheses with covariates associated with very large sample size (or effect size) are down-weighted relative to more intermediate hypotheses. This phenomenon is called \emph{size-investing} \citep{roeder2007improving,pena2011power,ignatiadis2016data}. The intuition is that higher weights should be preferentially allocated where they make most difference -- and little to hypotheses that are anyway exceedingly easy or hard to reject.
\end{remark}

%-----------------------------------------------------------------------
\subsubsection{Overall variance (independent of label) in ANOVA tests}
\label{subsubsec:ttest}
%------------------------------------------------------------------------
In Section~\ref{subsec:cars_sims} we demonstrated a covariate that can be used to improve power in the simultaneous two-sample testing problem for equality of means in the case of known variances. Here we extend the discussion to the case of unknown variances; cf.~\citet{cai2016cars} for a comprehensive treatment of more general forms of this problem.

Our data is drawn from model~\eqref{eq:two_sample}. We are interested in testing $H_i: \mu_{\grA,i} = \mu_{\grB,i}$ and do not know $\sigma_i$. The optimal test statistic for this situation is the two-sample $t$-statistic:

\begin{equation}\label{eq:two_sample_t_stat}
T_i = \sqrt{n} \frac{\overline{\grA_i}-\overline{\grB_i}}{\sqrt{S_{\grA,i}^2 + S_{\grB,i}^2}}, 
\end{equation}
\noindent where $\overline{\grA_i}$ and $\overline{\grB_i}$ are the sample means and $S_{\grA,i}^2$ and $S_{\grB,i}^2$ the sample variances. In addition, denote by $ \hat{\mu}_i \coloneqq \frac{1}{2}\left(\overline{\grA_i} + \overline{\grB_i}\right)$ and $S_i^2$ the sample mean and sample variance after pooling all observations ($\grA_{i,1},\dotsc, \grA_{i,n}, \grB_{i,1},\dotsc,\grB_{i,n}$) and forgetting their labels.

Now note that under the null hypothesis, $\mu_{\grA,i} = \mu_{\grB,i} = \mu_i$ and $\grA_{i,1}, \dotsc, \grA_{i,n}$, $\grB_{i,1}, \dotsc, \grB_{i,n} \sim \mathcal{N}(\mu_i, \sigma_i^2)$ i.i.d. Then, $(\hat{\mu}_i, S_i^2)$ is a complete sufficient statistic for the experiment, while $T_i$ is ancillary for $(\mu_i, \sigma_i^2)$. Thus, by Basu's theorem, $(\hat{\mu}_i, S_i^2)$ is independent of $T_i$ and we can use it as a covariate. 

Now consider $S_i^2$ in particular and note that under the null it is distributed as a scaled $\chi^2$-distribution. On the other hand, under the alternative, we expect $S_i^2$ to take larger values with high probability, especially if $|\mu_{\grA,i}-\mu_{\grB,i}|$ is large. Therefore, if we are doing $m$ $t$-tests, each with unknown variance $\sigma_i^2$ and if we assume $\sigma_i \sim G$ from a concentrated distribution $G$, then hypotheses with high $S_i^2$ are more likely to be true alternatives (and also likely to be alternatives with high power). Thus, the overall variance (ignoring sample labels) is not only independent of the p-values under the null hypothesis, but also informative about the alternatives. Using it as a covariate can lead to a large power increase in simultaneous two-sample $t$-tests \citep{bourgon2010independent, ignatiadis2016data}. The result extends to more complex ANOVA settings.

For a second example of the usefulness of $(\hat{\mu}_i, S_i^2)$ in this setting, consider the screening statistic $|\hat{\mu}_i|/S_i$. This can be interpreted as a statistic for the null hypothesis $\mu_{\grA,i} = \mu_{\grB,i} = 0$. If we believe a-priori that for many of the hypotheses $i$ with $\mu_{\grA,i} = \mu_{\grB,i}$ a sparsity condition holds, so that in fact $\mu_{\grA,i} = \mu_{\grB,i} = 0$ \citep{liu2014incorporation}, then large values of this statistic are more likely to correspond to alternatives, cf. Section~\ref{subsec:cars_sims}.

\begin{remark}
In single hypothesis testing, there is nothing to be gained from $(\hat{\mu}_i, S_i^2)$. Its usefulness only emerges in the multiple testing setup.
\end{remark}

%----------------------------------------------------------
\subsubsection{Ratio of number of observations in each group in two-sample tests}
%----------------------------------------------------------
For yet another example, revisit the two-sample situation, but now assume that for the $i$-th hypothesis, we have $n_{1,i}$ observations of the first population and $n_{2,i}$ observations from the second population, such that $n_{1,i}+n_{2,i} = n_i$. Then $n_{1,i}\,n_{2,i}/n_i^2$ is a statistic which is related to the alternative distribution, with values close to $\frac{1}{4}$ implying higher power~\citep{roquain2009optimal}. This statistic is also related to the Minor Allele Frequency (MAF) in genome-wide association studies~\citep{boca2018direct}.

%----------------------------------------------------------
\subsubsection{Sign of estimated effect size}
%----------------------------------------------------------
As a final example of a statistical covariate, consider a two-sided test where the null distribution is symmetric and the test-statistic is the absolute value of a symmetric statistic $T_i$. Then, the sign of $T_i$ is independent of the p-value under the null hypothesis. However, we might a-priori believe that among the alternatives, more have one or the other sign of effect size. Thus, the sign can be used as an informative covariate.  Previous uses of stratification by sign to improve power include the SAM (significance analysis of microarrays) procedure~\citepsup{tusher2001significance}.

\end{document}